\documentstyle[12pt,fleqn]{book}

\font\twlgot =eufm10 scaled \magstep1
\font\egtgot =eufm8
\font\sevgot =eufm7

\font\twlmsb =msbm10 scaled \magstep1
\font\egtmsb =msbm8
\font\sevmsb =msbm7

\newfam\gotfam
\def\pgot{\fam\gotfam\twlgot}
\textfont\gotfam\twlgot
\scriptfont\gotfam\egtgot
\scriptscriptfont\gotfam\sevgot
\def\got{\protect\pgot}

\newfam\msbfam
\textfont\msbfam\twlmsb
\scriptfont\msbfam\egtmsb
\scriptscriptfont\msbfam\sevmsb
\def\Bbb{\protect\pBbb}
\def\pBbb{\relax\ifmmode\expandafter\Bb\else\typeout{You cann't use
Bbb in text mode}\fi}
\def\Bb #1{{\fam\msbfam\relax#1}}

\def\op#1{\mathop{{\it\fam0} #1}\limits}

\newcommand{\id}{{\rm Id\,}}

\newcommand{\pr}{{\rm pr}}

\newcommand{\Id}{{\rm Id}}
\newcommand{\Ker}{{\rm Ker\,}}
\newcommand{\im}{{\rm Im\, }}
\newcommand{\re}{{\rm Re\, }}
\newcommand{\hm}{{\rm Hom\,}}
\newcommand{\dif}{{\rm Diff\,}}

\newcommand{\nm}[1]{|{#1}|}
\newcommand{\sU}{{\{U\}}}
\newcommand{\lng}{\langle}
\newcommand{\rng}{\rangle}
\newcommand{\bll}{\bullet}

\newcommand{\rx}{{\rm x}}

\newcommand{\bite}{\begin{itemize}}
\newcommand{\eite}{\end{itemize}}
\newcommand{\benu}{\begin{enumerate}}
\newcommand{\eenu}{\end{enumerate}}
\newcommand{\bde}{\begin{description}}
\newcommand{\ede}{\end{description}}
\newcommand{\bquo}{\begin{quote}}
\newcommand{\equo}{\end{quote}}
\newcommand{\bquot}{\begin{quotation}}
\newcommand{\equot}{\end{quotation}}

\newcommand{\eqref}[1]{(\ref{#1})}
\newcommand{\beq}{\begin{equation}}
\newcommand{\eeq}{\end{equation}}
\newcommand{\ben}{\begin{eqnarray}}
\newcommand{\een}{\end{eqnarray}}
\newcommand{\be}{\begin{eqnarray*}}
\newcommand{\ee}{\end{eqnarray*}}
\newcommand{\bea}{\begin{eqalph}}
\newcommand{\eea}{\end{eqalph}}

\newcommand{\nw}[1]{[{#1}]}
\newcommand{\cB}{{\cal B}}
\newcommand{\gU}{{\got U}}
\newcommand{\gO}{{\got O}}
\newcommand{\gA}{{\got A}}

\newcommand{\gP}{{\got P}}
\newcommand{\cG}{{\got g}}
\newcommand{\gd}{{\got d}}
\newcommand{\gj}{{\got J}}

\newcommand{\gS}{{\got S}}

\newcommand{\gR}{{\got R}}

\newcommand{\gf}{{\got f}}

\newcommand{\cJ}{{\cal J}}
\newcommand{\cA}{{\cal A}}
\newcommand{\cO}{{\cal O}}
\newcommand{\cT}{{\cal T}}
\newcommand{\cP}{{\cal P}}
\newcommand{\cR}{{\cal R}}

\newcommand{\cV}{{\cal V}}

\newcommand{\cE}{{\cal E}}
\newcommand{\cH}{{\cal H}}

\newcommand{\cC}{{\cal C}}
\newcommand{\cZ}{{\cal Z}}
\newcommand{\cI}{{\cal I}}

\newcommand{\cK}{{\cal K}}
\newcommand{\cM}{{\cal M}}
\newcommand{\cN}{{\cal N}}
\newcommand{\ccG}{{\cal G}}

\newcommand{\cS}{{\cal S}}
\newcommand{\cD}{{\cal D}}

\newcommand{\bL}{{\bf L}}
\newcommand{\bb}{{\bf 1}}

\newcommand{\bp}{{\bf p}}

\newcommand{\al}{\alpha}
\newcommand{\bt}{\beta}
\newcommand{\dl}{\delta}
\newcommand{\la}{\lambda}
\newcommand{\La}{\Lambda}
\newcommand{\f}{\phi}
\newcommand{\vf}{\varphi}

\newcommand{\p}{\pi}

\newcommand{\om}{\omega}
\newcommand{\Om}{\Omega}
\newcommand{\m}{\mu}

\newcommand{\g}{\gamma}
\newcommand{\G}{\Gamma}
\newcommand{\e}{\epsilon}
\newcommand{\ve}{\varepsilon}
\newcommand{\thh}{\theta}

\newcommand{\up}{\upsilon}
\newcommand{\vt}{\vartheta}

\newcommand{\si}{\sigma}
\newcommand{\Si}{\Sigma}

\newcommand{\w}{\wedge}

\newcommand{\wt}{\widetilde}
\newcommand{\wh}{\widehat}
\newcommand{\ol}{\overline}

\newcommand{\dr}{\partial}

\newcommand{\mar}[1]{}

\newcommand{\lto}{{\leftarrow}}

\newcommand{\lla}{\op\longleftarrow}
\newcommand{\ar}{\op\longrightarrow}

\newcommand{\ot}{\otimes}

\let\ssection=\section
\renewcommand{\section}{\setcounter{equation}{0}\ssection}

\newcounter{eqalph}[section]
\newcounter{equationa}[section]
\newcounter{example}[section]
\newcounter{remark}[section]
\newcounter{theorem}[section]
\newcounter{proposition}[section]
\newcounter{lemma}[section]
\newcounter{corollary}[section]
\newcounter{definition}[section]

\def\theremark{\arabic{chapter}.\arabic{section}.\arabic{remark}}

\def\thedefinition{\arabic{chapter}.\arabic{section}.\arabic{definition}}

\newenvironment{proof}{\noindent {\it Outline of proof:}\small
}{{\footnotesize\it
QED}\medskip}
\newenvironment{ex}{\refstepcounter{remark}\medskip\noindent{\bf
Example
\theremark:}\small }{$\diamondsuit$ \medskip}
\newenvironment{rem}{\refstepcounter{remark}\medskip\noindent{\bf
Remark
\theremark:}\small }{$\diamondsuit$\medskip}
\newenvironment{theo}{\refstepcounter{definition}
\medskip\noindent{\sc Theorem \thedefinition}:}{$\Box$\medskip}
\newenvironment{prop}{\refstepcounter{definition}\medskip\noindent{\sc
Proposition \thedefinition}:}{$\Box$\medskip}

\newenvironment{defi}{\refstepcounter{definition}\medskip\noindent{\sc
Definition \thedefinition}:}{$\Box$\medskip}

\newenvironment{eqalph}{\stepcounter{equation}
\setcounter{equationa}{\value{equation}}
\setcounter{equation}{0}

\begin{eqnarray}}{\end{eqnarray}
\setcounter{equation}{\value{equationa}}}

\makeindex

\begin{document}

\hbox{}

\thispagestyle{empty}

\setcounter{page}{0}

\vskip 3cm

\begin{center}

{\large \bf Lectures on Differential Geometry of Modules and
Rings}

\bigskip
\bigskip
\bigskip

{\sc G. Sardanashvily}
\bigskip

Department of Theoretical Physics, Moscow State University,
Moscow, Russia

\bigskip
\bigskip
\bigskip
\bigskip
\bigskip
\bigskip

{\bf Abstract}
\bigskip
\end{center}

\noindent Generalizing differential geometry of smooth vector
bundles formulated in algebraic terms of the ring of smooth
functions, its derivations and the Koszul connection, one can
define differential operators, differential calculus and
connections on modules over arbitrary commutative, graded
commutative and noncommutative rings. For instance, this is the
case of quantum theory, SUSY theory and noncommutative geometry,
respectively. The relevant material on this subject is summarized.

%\pagenumbering{arabic}

%\tableofcontents

\newpage

\centerline{\large \bf Contents}
\bigskip
\bigskip

\contentsline {chapter}{Introduction}{3} \contentsline
{chapter}{\numberline {1}Commutative geometry}{5} \contentsline
{section}{\numberline {1.1}Commutative algebra}{5} \contentsline
{section}{\numberline {1.2}Differential operators on modules and
rings}{10} \contentsline {section}{\numberline {1.3}Connections on
modules and rings}{13} \contentsline {section}{\numberline
{1.4}Differential calculus over a commutative ring}{17}
\contentsline {section}{\numberline {1.5}Local-ringed spaces}{19}
\contentsline {section}{\numberline {1.6}Differential geometry of
$C^\infty (X)$-modules}{21} \contentsline {section}{\numberline
{1.7}Connections on local-ringed spaces}{25} \contentsline
{chapter}{\numberline {2}Geometry of quantum systems}{29}
\contentsline {section}{\numberline {2.1}Geometry of Banach
manifolds}{29} \contentsline {section}{\numberline {2.2}Geometry
of Hilbert manifolds}{34} \contentsline {section}{\numberline
{2.3}Hilbert and $C^*$-algebra bundles}{40} \contentsline
{section}{\numberline {2.4}Connections on Hilbert and
$C^*$-algebra bundles}{44} \contentsline {section}{\numberline
{2.5}Instantwise quantization}{45} \contentsline
{section}{\numberline {2.6}Berry connection}{48} \contentsline
{chapter}{\numberline {3}Supergeometry}{53} \contentsline
{section}{\numberline {3.1}Graded tensor calculus}{53}
\contentsline {section}{\numberline {3.2}Graded differential
calculus and connections}{56} \contentsline {section}{\numberline
{3.3}Geometry of graded manifolds}{62} \contentsline
{section}{\numberline {3.4}Supermanifolds}{69} \contentsline
{section}{\numberline {3.5}Supervector bundles}{74} \contentsline
{section}{\numberline {3.6}Superconnections}{77} \contentsline
{chapter}{\numberline {4}Non-commutative geometry}{81}
\contentsline {section}{\numberline {4.1}Modules over
$C^*$-algebras}{81} \contentsline {section}{\numberline
{4.2}Non-commutative differential calculus}{84} \contentsline
{section}{\numberline {4.3}Differential operators in
non-commutative geometry}{87} \contentsline {section}{\numberline
{4.4}Connections in non-commutative geometry}{92} \contentsline
{section}{\numberline {4.5}Matrix geometry}{95} \contentsline
{section}{\numberline {4.6}Connes' non-commutative geometry}{97}
\contentsline {section}{\numberline {4.7}Differential calculus
over Hopf algebras}{100} \contentsline {chapter}{\numberline
{5}Appendix. Cohomology}{111} \contentsline {section}{\numberline
{5.1}Cohomology of complexes}{111} \contentsline
{section}{\numberline {5.2}Cohomology of Lie algebras}{114}
\contentsline {section}{\numberline {5.3}Sheaf cohomology}{116}
\contentsline {chapter}{Bibliography}{125} \contentsline
{chapter}{Index}{131}

%------------------------------------------------------------------------------

\addcontentsline{toc}{chapter}{Introduction}

\chapter*{Introduction}

Geometry of classical mechanics and field theory is mainly
differential geometry of finite-dimensional smooth manifolds and
fibre bundles \cite{book09,book98,sard09}.

At the same time, the standard mathematical language of quantum
mechanics and field theory has been long far from geometry. In the
last twenty years, the incremental development of new physical
ideas in quantum theory (including super- and BRST symmetries,
geometric and deformation quantization, topological field theory,
anomalies, non-commutativity, strings and branes) has called into
play advanced geometric techniques, based on the deep interplay
between algebra, geometry and topology \cite{book05,sardgeom}.

Geometry in quantum systems speaks mainly the algebraic language
of rings, modules and sheaves due to the fact that the basic
ingredients in the differential calculus and differential geometry
of smooth manifolds (except non-linear differential operators) can
be restarted in a pure algebraic way.

Let $X$ be a smooth manifold and $C^\infty(X)$ the ring  of smooth
real functions on $X$. A key point is that, by virtue of the
well-known Serre--Swan theorem, a $C^\infty(X)$-module is finitely
generated projective iff it is isomorphic to the module of
sections of some smooth vector bundle over $X$. Moreover, this
isomorphism is a categorial equivalence. Therefore, differential
geometry of smooth vector bundles can be adequately formulated in
algebraic terms of the ring $C^\infty(X)$, its derivations and the
Koszul connection (Section 1.6).

In a general setting, let $\cK$ be a commutative ring, $\cA$ an
arbitrary commutative $\cK$-ring, and $P$, $Q$ some $\cA$-modules.
The $\cK$-linear $Q$-valued differential operators on $P$ can be
defined \cite{grot,kras}. The representative objects of the
functors $Q\to \dif_s(P,Q)$ are the jet modules $\cJ^sP$ of $P$.
Using the first order jet module $\cJ^1P$, one also restarts the
notion of a connection on an $\cA$-module $P$
\cite{kosz60,book00}. Such a connection assigns to each derivation
$\tau\in\gd\cA$ of a $\cK$-ring $\cA$ a first order $P$-valued
differential operator $\nabla_\tau$ on $P$ obeying the Leibniz
rule
\be
\nabla_\tau(ap)=\tau(a)p+a\nabla_\tau(p).
\ee
Similarly, connections on local-ringed spaces are introduced
(Section 1.7).

As was mentioned above, if $P$ is a $C^\infty(X)$-module of
sections of a smooth vector bundle $Y\to X$, we come to the
familiar notions of a linear differential operator on $Y$, the
jets of sections of $Y\to X$ and a linear connection on $Y\to X$.

In quantum theory, Banach and Hilbert manifolds, Hilbert bundles
and bundles of $C^*$-algebras over smooth manifolds are
considered. Their differential geometry also is formulated as
geometry of modules, in particular, $C^\infty(X)$-modules (Chapter
2).

Let $\cK$ be a commutative ring, $\cA$ a (commutative or
non-commutative) $\cK$-ring, and $\cZ(\cA)$ the center of $\cA$.
Derivations of $\cA$ make up a Lie $\cK$-algebra $\gd\cA$. Let us
consider the Chevalley--Eilenberg complex of $\cK$-multilinear
morphisms of $\gd\cA$ to $\cA$, seen as a $\gd\cA$-module
\cite{fuks,book00}. Its subcomplex $\cO^*(\gd\cA,d)$ of
$\cZ(\cA)$-multilinear morphisms is a differential graded algebra,
called the Chevalley--Eilenberg differential calculus over $\cA$.
It contains the minimal differential calculus $\cO^*\cA$ generated
by elements $da$, $a\in\cA$. If $\cA$ is the $\Bbb R$-ring
$C^\infty(X)$ of smooth real functions on a smooth manifold $X$,
the module $\gd C^\infty(X)$ of its derivations is the Lie algebra
of vector fields on $X$ and the Chevalley--Eilenberg differential
calculus over $C^\infty(X)$ is exactly the algebra of exterior
forms on a manifold $X$ where the Chevalley--Eilenberg coboundary
operator $d$ coincides with the exterior differential, i.e.,
$\cO^*(\gd C^\infty(X),d)$ is the familiar de Rham complex. In a
general setting, one therefore can think of elements of the
Chevalley--Eilenberg differential calculus $\cO^k(\gd\cA,d)$ over
an algebra $\cA$ as being differential forms over $\cA$.

Similarly, the Chevalley--Eilenberg differential calculus over a
graded commutative ring is constructed \cite{fuks,book09,sard09a}.
This is the case of SUSY theory and supergeometry (Chapter 3). In
supergeometry, connections on graded manifolds and supervector
bundles are defined as those on graded modules over a graded
commutative ring and graded local-ringed spaces are defined
\cite{bart,book00,sard09a}.

Non-commutative geometry also is developed as a generalization of
the calculus in commutative rings of smooth functions
\cite{conn,land}. In a general setting, any non-commutative
$\cK$-ring $\cA$ over a commutative ring $\cK$ can be called into
play. One can consider the above mentioned Chevalley--Eilenberg
differential calculus $\cO^*\cA$ over $\cA$, differential
operators and connections on $\cA$-modules (but not their jets)
(Chapter 4). If the derivation $\cK$-module $\gd\cA$ is a finite
projective module with respect to the center of $\cA$, one can
treat the triple $(\cA,\gd\cA,\cO^*\cA)$ as a non-commutative
space.

Note that, in non-commutative geometry, different definitions of a
differential operator on modules over a non-commutative ring have
been suggested \cite{bor97,dublmp,lunts}, but there is no a
satisfactory construction of higher order differential operator
\cite{book05,sard07} Roughly speaking, the difficulty lies in the
fact that, if $\dr$ is a derivation of a non-commutative
$\cK$-ring $\cA$, the product $a\dr$, $a\in\cA$, need not be so.
There are also different definitions of a connection on modules
over a non-commutative ring \cite{dub,land,book00}.

\chapter{Commutative geometry}

We start with differential calculus and differential geometry of
modules over commutative rings. In particular, this is the case of
differential geometry of smooth manifolds and smooth vector
bundles (Section 1.6).

\section{Commutative algebra}

In this Section, the relevant basics on modules over commutative
algebras is summarized  \cite{lang,mcl}.

An {\it algebra} $\cA$ \index{algebra} is an additive group which
is additionally provided with distributive multiplication. All
algebras  throughout the Lectures are associative, unless they are
Lie algebras.  A {\it ring} \index{ring} is a {\it unital}
algebra, \index{algebra!unital} i.e., it contains a unit element
$\bb$. Unless otherwise stated, we assume that $\bb\neq 0$, i.e.,
a ring does not reduce to the zero element. One says that $A$ is a
{\it division algebra} \index{division algebra} if it has no a
divisor of zero, i.e., $ab=0$, $a,b\in A$, implies either $a=0$ or
$b=0$. Non-zero elements of a ring form a multiplicative monoid.
If this multiplicative monoid is a multiplicative group, one says
that the ring has a multiplicative inverse. A ring $A$ has a
multiplicative inverse iff it is a division algebra. A {\it field}
\index{field} is a commutative ring whose non-zero elements make
up a multiplicative group.

A subset $\cI$ of an algebra $\cA$ is called a left (resp. right)
{\it ideal} \index{ideal} if it is a subgroup of the additive
group $\cA$ and $ab\in \cI$ (resp. $ba\in\cI$) for all $a\in \cA$,
$b\in \cI$. If $\cI$ is both a left and right ideal, it is called
a two-sided ideal. An ideal is a subalgebra, but a {\it proper}
ideal \index{ideal!proper} (i.e., $\cI\neq \cA$) of a ring is not
a subring because it does not contain a unit element.

Let $\cA$ be a commutative ring. Of course, its ideals are
two-sided. Its proper ideal is said to be {\it maximal}
\index{ideal!maximal} if it does not belong to another proper
ideal. A commutative ring $\cA$ is called {\it local} \index{local
ring} if it has a unique maximal ideal. This ideal consists of all
non-invertible elements of $\cA$. A proper two-sided ideal $\cI$
of a commutative ring is called {\it prime} \index{ideal!prime} if
$ab\in \cI$ implies either $a\in\cI$ or $b\in\cI$. Any maximal
two-sided ideal is prime. Given a two-sided ideal $\cI\subset
\cA$, the additive factor group $\cA/\cI$ is an algebra, called
the {\it factor algebra}. \index{factor algebra} If $\cA$ is a
ring, then $\cA/\cI$ is so. If $\cI$ is a prime ideal, the factor
ring $\cA/\cI$ has no divisor of zero, and it is a field if $\cI$
is a maximal ideal.

Given an algebra $\cA$, an additive group $P$ is said to be a left
(resp. right) {\it $\cA$-module} \index{module} if it is provided
with distributive multiplication $\cA\times P\to P$ by elements of
$\cA$ such that $(ab)p=a(bp)$ (resp. $(ab)p=b(ap)$) for all
$a,b\in\cA$ and $p\in P$.  If $\cA$ is a ring, one additionally
assumes that $\bb p=p=p\bb$ for all $p\in P$. Left and right
module structures are usually written by means of left and right
multiplications $(a,p)\mapsto ap$ and $(a,p)\mapsto pa$,
respectively. If $P$ is both a left module over an algebra $\cA$
and a right module over an algebra $\cA'$, it is called an
$(\cA-\cA')$-bimodule (an $\cA$-bimodule if $\cA=\cA'$).
\index{bimodule} If $\cA$ is a commutative algebra, an
$(\cA-\cA)$-bimodule $P$ is said to be {\it commutative}
\index{bimodule!commutative} if $ap=pa$ for all $a\in \cA$ and
$p\in P$. Any left or right module over a commutative algebra
$\cA$ can be brought into a commutative bimodule. Therefore,
unless otherwise stated, any module over a commutative algebra
$\cA$ is called an $\cA$-module.

A module over a field is called a {\it vector space}.
\index{vector space} If an algebra $\cA$ is a module over a ring
$\cK$, it is said to be a {\it $\cK$-algebra}.
\index{$\cK$-algebra} Any algebra can be seen as a $\Bbb
Z$-algebra.

\begin{rem} \label{w120} \mar{w120}
Any $\cK$-algebra $\cA$ can be extended to a unital algebra $\wt
\cA$ by the adjunction of the identity $\bb$ to $\cA$. The algebra
$\wt\cA$, called the {\it unital extension} \index{unital
extension} of $\cA$, is defined as the direct sum of $\cK$-modules
$\cK\oplus \cA$ provided with the multiplication
\be
(\la_1,a_1)(\la_2,a_2)=(\la_1\la_2,\la_1a_2+\la_2 a_1 +a_1a_2),
\quad \la_1,\la_2 \in \cK, \quad a_1,a_2\in \cA.
\ee
Elements of $\wt\cA$ can be written as $(\la,a)=\la\bb+a$,
$\la\in\cK$, $a\in\cA$.

Let us note that, if $\cA$ is a unital algebra, the identity
$\bb_\cA$ in $\cA$ fails to be that in $\wt\cA$. In this case, the
algebra $\wt\cA$ is isomorphic to the product of $\cA$ and the
algebra $\cK(\bb-\bb_\cA)$.
\end{rem}

In this Chapter, all associative algebras are assumed to be
commutative, unless they are graded.

The following are standard constructions of new modules from old
ones.

$\bullet$ The {\it direct sum} \index{direct sum!of modules}
$P_1\oplus P_2$ of $\cA$-modules $P_1$ and $P_2$ is the additive
group $P_1\times P_2$ provided with the $\cA$-module structure
\be
a(p_1,p_2)=(ap_1,ap_2), \qquad p_{1,2}\in P_{1,2}, \qquad a\in\cA.
\ee
Let $\{P_i\}_{i\in I}$ be a set of modules. Their direct sum
$\oplus P_i$ consists of elements $(\ldots, p_i,\ldots)$ of the
Cartesian product $\prod P_i$ such that $p_i\neq 0$ at most for a
finite number of indices $i\in I$.

$\bullet$ The {\it tensor product} \index{tensor product!of
modules} $P\ot Q$ of $\cA$-modules $P$ and $Q$ is an additive
group which is generated by elements $p\ot q$, $p\in P$, $q\in Q$,
obeying the relations
\be
&& (p+p')\ot q =p\ot q + p'\ot q, \quad p\ot(q+q')=p\ot q+p\ot q', \\
&&  pa\ot q= p\ot aq, \qquad p\in P, \qquad q\in Q, \qquad a\in\cA,
\ee
and it is provided with the $\cA$-module structure
\be
a(p\ot q)=(ap)\ot q=p\ot (qa)=(p\ot q)a.
\ee
If the ring $\cA$ is treated as an $\cA$-module, the tensor
product  $\cA\ot_\cA Q$ is canonically isomorphic to $Q$ via the
assignment
\be
\cA\ot_\cA Q\ni a\ot q \leftrightarrow aq\in Q.
\ee

$\bullet$ Given a submodule $Q$ of an $\cA$-module $P$, the
quotient $P/Q$ of the additive group $P$ with respect to its
subgroup $Q$ is also provided with an $\cA$-module structure. It
is called a {\it factor module}. \index{factor module}

$\bullet$ The set $\hm_\cA(P,Q)$ of $\cA$-linear morphisms of an
$\cA$-module $P$ to an $\cA$-module $Q$ is naturally an
$\cA$-module. The $\cA$-module $P^*=\hm_\cA(P,\cA)$ is called the
{\it dual} \index{module!dual} of an $\cA$-module $P$. There is a
natural monomorphism  $P\to P^{**}$.

An $\cA$-module $P$ is called {\it free} \index{module!free} if it
has a {\it basis}, \index{basis!for a module} i.e., a linearly
independent subset $I\subset P$ spanning $P$ such that each
element of $P$ has a unique expression as a linear combination of
elements of $I$ with a finite number of non-zero coefficients from
an algebra $\cA$. Any vector space is free. Any module is
isomorphic to a quotient of a free module. A module is said to be
{\it finitely generated} \index{module!finitely generated} (or of
{\it finite rank}) \index{module!of finite rank} if it is a
quotient of a free module with a finite basis.

One says that a module $P$ is {\it projective}
\index{module!projective} if it is a direct summand of a free
module, i.e., there exists a module $Q$ such that $P\oplus Q$ is a
free module. A module $P$ is projective iff $P=\bp S$ where $S$ is
a free module and $\bp$ is a projector of $S$, i.e., $\bp^2=\bp$.
If $P$ is a projective module of finite rank over a ring, then its
dual $P^*$ is so, and $P^{**}$ is isomorphic to $P$.

\begin{theo}
Any projective module over a local ring is free.
\end{theo}

Now we focus on exact sequences, direct and inverse limits of
modules \cite{mcl,massey}.

A composition of module morphisms
\be
P\ar^i Q\ar^j T
\ee
is said to be {\it exact} \index{exact sequence!of modules} at $Q$
if $\Ker j=\im i$. A composition of module morphisms
\mar{spr13}\beq
0\to P\ar^i Q\ar^j T\to 0 \label{spr13}
\eeq
is called a {\it short exact sequence} \index{exact
sequence!short} if it is exact at all the terms $P$, $Q$, and $T$.
This condition implies that: (i) $i$ is a monomorphism,  (ii)
$\Ker j=\im i$, and (iii) $j$ is an epimorphism onto the quotient
$T=Q/P$.

\begin{theo} \label{spr183} \mar{spr183}
Given an exact sequence of modules (\ref{spr13}) and another
$\cA$-module $R$, the sequence of modules
\mar{spr900}\beq
0\to\hm_\cA(T,R)\ar^{j^*} \hm_\cA(Q,R)\ar^{i^*} \hm(P,R)
\label{spr900}
\eeq
is exact at the first and second terms, i.e., $j^*$ is a
monomorphism, but $i^*$ need not be an epimorphism.
\end{theo}

One says that the exact sequence (\ref{spr13}) is {\it split}
\index{splitting of an exact sequence} if there exists a
monomorphism $s:T\to Q$ such that $j\circ s=\id T$ or,
equivalently,
\be
Q=i(P)\oplus s(T) \cong P\oplus T.
\ee
The exact sequence (\ref{spr13}) is always split if $T$ is a
projective module.

A {\it directed set} \index{directed set} $I$ is a set with an
order relation $<$ which satisfies the following three conditions:

(i) $i<i$, for all $i\in I$;

(ii) if $i<j$ and $j< k$, then $i<k$;

(iii) for any $i,j\in I$, there exists $k\in I$ such that $i<k$
and $j<k$.

\noindent It may happen that $i\neq j$, but $i<j$ and $j<i$
simultaneously.

A family of modules $\{P_i\}_{i\in I}$ (over the same algebra),
indexed by a directed set $I$, is called a {\it direct system}
\index{direct system of modules} if, for any pair $i<j$, there
exists a morphism $r^i_j:P_i\to P_j$ such that
\be
r^i_i=\id P_i, \qquad r^i_j\circ r^j_k=r^i_k, \qquad i<j<k.
\ee
A direct system of modules admits a {\it direct limit}.
\index{direct limit} This is a module $P_\infty$ together with
morphisms $r^i_\infty: P_i\to P_\infty$ such that
$r^i_\infty=r^j_\infty\circ r^i_j$ for all $i<j$. The module
$P_\infty$ consists of elements of the direct sum
$\op\oplus_{I}P_i$ modulo the identification of elements of $P_i$
with their images in $P_j$ for all $i<j$. An example of a direct
system is a {\it direct sequence} \index{direct sequence}
\mar{spr1}\beq
P_0\ar P_1\ar \cdots P_i\ar^{r^i_{i+1}}\cdots, \qquad I=\Bbb N.
\label{spr1}
\eeq
It should be noted that direct limits also exist in the categories
of commutative algebras and rings, but not in categories whose
objects are non-Abelian groups.

\begin{theo} \label{spr170} \mar{spr170}
Direct limits commute with direct sums and tensor products of
modules. Namely, let $\{P_i\}$ and $\{Q_i\}$ be two direct systems
of modules over the same algebra which are indexed by the same
directed set $I$, and let $P_\infty$ and $Q_\infty$ be their
direct limits. Then the direct limits of the direct systems
$\{P_i\oplus Q_i\}$ and $\{P_i\ot Q_i\}$ are $P_\infty\oplus
Q_\infty$ and $P_\infty\ot Q_\infty$, respectively.
\end{theo}

A morphism of a direct system $\{P_i, r^i_j\}_I$ to a direct
system $\{Q_{i'}, \rho^{i'}_{j'}\}_{I'}$ consists of an order
preserving map $f:I\to I'$ and morphisms $F_i:P_i\to Q_{f(i)}$
which obey the compatibility conditions
\be
\rho^{f(i)}_{f(j)}\circ F_i=F_j\circ r^i_j.
\ee
If $P_\infty$ and $Q_\infty$ are limits of these direct systems,
there exists a unique morphism $F_\infty: P_\infty\to Q_\infty$
such that
\be
\rho^{f(i)}_\infty\circ F_i=F_\infty\circ r^i_\infty.
\ee
Moreover, direct limits preserve monomorphisms and epimorphisms.
To be precise, if all $F_i:P_i\to Q_{f(i)}$ are monomorphisms or
epimorphisms, so is $\Phi_\infty:P_\infty\to Q_\infty$. As a
consequence, the following holds.

\begin{theo} \label{dlim1} \mar{dlim1}
Let short exact sequences
\mar{spr186}\beq
0\to P_i\ar^{F_i} Q_i\ar^{\Phi_i} T_i\to 0 \label{spr186}
\eeq
for all $i\in I$ define a short exact sequence of direct systems
of modules $\{P_i\}_I$, $\{Q_i\}_I$, and $\{T_i\}_I$ which are
indexed by the same directed set $I$. Then there exists a short
exact sequence of their direct limits
\mar{spr187}\beq
0\to P_\infty\ar^{F_\infty} Q_\infty\ar^{\Phi_\infty} T_\infty\to
0. \label{spr187}
\eeq
\end{theo}

In particular, the direct limit of factor modules $Q_i/P_i$ is the
factor module $Q_\infty/P_\infty$. By virtue of Theorem
\ref{spr170}, if all the exact sequences (\ref{spr186}) are split,
the exact sequence (\ref{spr187}) is well.

\begin{ex} \label{ws40} \mar{ws40} Let $P$ be an $\cA$-module.
We denote $P^{\ot k}=\op\ot^kP$. Let us consider the direct system
of $\cA$-modules with respect to monomorphisms
\be
\cA\ar (\cA\oplus P)\ar \cdots  (\cA\oplus P\oplus\cdots\oplus
P^{\ot k})\ar \cdots.
\ee
Its direct limit
\mar{spr620}\beq
\ot P=\cA\oplus P\oplus\cdots\oplus P^{\ot k}\oplus\cdots
\label{spr620}
\eeq
is an $\Bbb N$-graded $\cA$-algebra with respect to the tensor
product $\ot$. It is called the {\it tensor algebra} \index{tensor
algebra} of a module $P$. Its quotient with respect to the ideal
generated by elements $p\ot p'+p'\ot p$, $p,p'\in P$, is an $\Bbb
N$-graded commutative algebra, called the {\it exterior algebra}
\index{exterior algebra} of a module $P$.
\end{ex}

We restrict our consideration of inverse systems of modules to
{\it inverse sequences} \index{inverse sequence}
\mar{spr2}\beq
P^0\lla P^1\lla \cdots P^i\op\lla^{\pi^{i+1}_i}\cdots.
\label{spr2}
\eeq
Its {\it inductive limit} \index{inductive limit} ({\it the
inverse limit}) \index{inverse limit} is a module $P^\infty$
together with morphisms $\pi^\infty_i: P^\infty\to P^i$ such that
$\pi^\infty_i=\pi^j_i\circ \pi^\infty_j$ for all $i<j$. It
consists of elements $(\ldots,p^i,\ldots)$, $p^i\in P^i$, of the
Cartesian product $\prod P^i$ such that $p^i=\pi^j_i(p^j)$ for all
$i<j$.

\begin{theo} \label{spr3} \mar{spr3}
Inductive limits preserve monomorphisms, but not epimorphisms. If
a sequence
\be
0\to P^i\ar^{F^i} Q^i\ar^{\Phi^i} T^i, \qquad i\in\Bbb N,
\ee
of inverse systems of modules $\{P^i\}$, $\{Q^i\}$ and $\{T^i\}$
is exact, so is the sequence of the inductive limits
\be
0\to P^\infty\ar^{F^\infty} Q^\infty\ar^{\Phi^\infty} T^\infty.
\ee
\end{theo}

In contrast with direct limits, the inductive ones exist in the
category of groups which are not necessarily commutative.

\begin{ex} \label{t1} \mar{t1}
Let $\{P_i\}$ be a direct sequence of modules. Given another
module $Q$, the modules $\hm(P_i,Q)$ make up an inverse system
such that its inductive limit is isomorphic to $\hm(P_\infty,Q)$.
\end{ex}

\section{Differential operators on modules and rings}

This Section addresses the notion of a (linear) differential
operator on a module over a commutative ring
\cite{grot,kras,book09}.

Let $\cK$ be a commutative ring and $\cA$ a commutative
$\cK$-ring. Let $P$ and $Q$ be $\cA$-modules. The $\cK$-module
$\hm_\cK (P,Q)$ of $\cK$-module homomorphisms $\Phi:P\to Q$ can be
endowed with the two different $\cA$-module structures
\mar{5.29}\beq
(a\Phi)(p)= a\Phi(p),  \qquad  (\Phi\bll a)(p) = \Phi (a p),\qquad
a\in \cA, \quad p\in P. \label{5.29}
\eeq
For the sake of convenience, we will refer to the second one as
the $\cA^\bll$-module structure. Let us put
\mar{spr172}\beq
\dl_a\Phi= a\Phi -\Phi\bll a, \qquad a\in\cA. \label{spr172}
\eeq

\begin{defi} \label{ws131} \mar{ws131}
An element $\Delta\in\hm_\cK(P,Q)$ is called a $Q$-valued {\it
differential operator} \index{differential operator!on a module}
of order $s$ on $P$ if
\be
\dl_{a_0}\circ\cdots\circ\dl_{a_s}\Delta=0
\ee
for any tuple of $s+1$ elements $a_0,\ldots,a_s$ of $\cA$.
\end{defi}

The set $\dif_s(P,Q)$ of these operators inherits the $\cA$- and
$\cA^\bll$-module structures (\ref{5.29}). Of course, an $s$-order
differential operator is also of $(s+1)$-order.

In particular, zero order differential operators obey the
condition
\be
\dl_a \Delta(p)=a\Delta(p)-\Delta(ap)=0, \qquad a\in\cA, \qquad
p\in P,
\ee
and, consequently, they coincide with $\cA$-module morphisms $P\to
Q$. A first order differential operator $\Delta$ satisfies the
condition
\mar{ws106}\beq
\dl_b\circ\dl_a\,\Delta(p)= ba\Delta(p) -b\Delta(ap)
-a\Delta(bp)+\Delta(abp) =0, \quad a,b\in\cA. \label{ws106}
\eeq

The following fact reduces the study of $Q$-valued differential
operators on an $\cA$-module $P$ to that of $Q$-valued
differential operators on the ring $\cA$.

\begin{prop} \label{ws109} \mar{ws109}
Let us consider the $\cA$-module morphism
\mar{n2}\beq
h_s: \dif_s(\cA,Q)\to Q, \qquad h_s(\Delta)=\Delta(\bb).
\label{n2}
\eeq
Any $Q$-valued $s$-order differential operator $\Delta\in
\dif_s(P,Q)$ on $P$ uniquely factorizes
\mar{n13}\beq
\Delta:P\ar^{\gf_\Delta} \dif_s(\cA,Q)\ar^{h_s} Q \label{n13}
\eeq
through the morphism $h_s$ (\ref{n2}) and some homomorphism
\mar{n0}\beq
\gf_\Delta: P\to \dif_s(\cA,Q), \qquad (\gf_\Delta
p)(a)=\Delta(ap), \qquad a\in \cA, \label{n0}
\eeq
of the $\cA$-module $P$ to the $\cA^\bll$-module $\dif_s(\cA,Q)$
\cite{kras}. The assignment $\Delta\mapsto\gf_\Delta$ defines the
isomorphism
\mar{n1}\beq
\dif_s(P,Q)=\hm_{\cA-\cA^\bll}(P,\dif_s(\cA,Q)). \label{n1}
\eeq
\end{prop}

Let $P=\cA$. Any zero order $Q$-valued differential operator
$\Delta$ on $\cA$ is defined by its value $\Delta(\bb)$. Then
there is an isomorphism $\dif_0(\cA,Q)=Q$ via the association
\be
Q\ni q\mapsto \Delta_q\in \dif_0(\cA,Q),
\ee
where $\Delta_q$ is given by the equality $\Delta_q(\bb)=q$. A
first order $Q$-valued differential operator $\Delta$ on $\cA$
fulfils the condition
\be
\Delta(ab)=b\Delta(a)+ a\Delta(b) -ba \Delta(\bb), \qquad
a,b\in\cA.
\ee
It is called a $Q$-valued {\it derivation} \index{derivation!with
values in a module} of $\cA$ if $\Delta(\bb)=0$, i.e., the {\it
Leibniz rule} \index{Leibniz rule}
\mar{+a20}\beq
\Delta(ab) = \Delta(a)b + a\Delta(b), \qquad  a,b\in \cA,
\label{+a20}
\eeq
holds. One obtains at once that any first order differential
operator on $\cA$ falls into the sum
\be
\Delta(a)= a\Delta(\bb) +[\Delta(a)-a\Delta(\bb)]
\ee
of the zero order differential operator $a\Delta(\bb)$ and the
derivation $\Delta(a)-a\Delta(\bb)$. If $\dr$ is a derivation of
$\cA$, then $a\dr$ is well for any $a\in \cA$. Hence, derivations
of $\cA$ constitute an $\cA$-module $\gd(\cA,Q)$, called the {\it
derivation module}. \index{derivation module} There is the
$\cA$-module decomposition
\mar{spr156'}\beq
\dif_1(\cA,Q) = Q \oplus\gd(\cA,Q). \label{spr156'}
\eeq

\begin{rem} \label{w70} \mar{w70}
Let us recall that, given a (non-commutative) $\cK$-algebra $\cA$
and an
 $\cA$-bimodule $Q$, by a $Q$-valued {\it derivation} of $\cA$
\index{derivation!of a non-commutative algebra} is meant a
$\cK$-module morphism $u:\cA\to Q$ which obeys the {\it Leibniz
rule} \index{Leibniz rule!non-commutative}
\mar{ws100}\beq
u(ab)=u(a)b +a u(b), \qquad a,b\in\cA. \label{ws100}
\eeq
It should be emphasized that this derivation rule differs from
that (\ref{ws10}) of graded derivations. A $Q$-valued derivation
$u$ of $A$ is called {\it inner} \index{derivation!inner} if there
exists an element $q\in Q$ such that $u(a)=qa-aq$.
\end{rem}

If $Q=\cA$, the derivation module $\gd\cA$ of $\cA$ is also a Lie
algebra over the ring $\cK$ with respect to the Lie bracket
\mar{+860}\beq
[u,u']=u\circ u' - u'\circ u, \qquad u,u'\in\cA. \label{+860}
\eeq
Accordingly, the decomposition (\ref{spr156'}) takes the form
\mar{spr156}\beq
\dif_1(\cA) = \cA \oplus\gd\cA. \label{spr156}
\eeq

An $s$-order differential operator on a module $P$ is represented
by a zero order differential operator on the module of $s$-order
jets of $P$ as follows.

Given an $\cA$-module $P$, let us consider the tensor product
$\cA\otimes_\cK P$ of $\cK$-modules $\cA$ and $P$. We put
\mar{spr173}\beq
\dl^b(a\otimes p)= (ba)\otimes p - a\otimes (b p), \qquad p\in P,
\qquad a,b\in\cA.  \label{spr173}
\eeq
Let us denote by $\m^{k+1}$ the submodule of $\cA\ot_\cK P$
generated by elements of the type
\be
\dl^{b_0}\circ \cdots \circ\dl^{b_k}(a\otimes p).
\ee
The {\it $k$-order jet module} $\cJ^k(P)$ \index{jet module} of a
module $P$ is defined as the quotient of the $\cK$-module
$\cA\otimes_\cK P$ by $\m^{k+1}$. We denote its elements
$c\ot_kp$.

In particular, the first order jet module $\cJ^1(P)$ consists of
elements $c\ot_1 p$ modulo the relations
\mar{mos041}\beq
\dl^a\circ \dl^b(\bb\ot_1 p)= ab\otimes_1 p -b\otimes_1 (ap)
-a\otimes_1 (bp) +\bb\ot_1(abp) =0. \label{mos041}
\eeq

The $\cK$-module $\cJ^k(P)$ is endowed with the $\cA$- and
$\cA^\bll$-module structures
\mar{+a21}\beq
b(a\ot_k p)= ba\ot_k p, \qquad b\bll(a\otimes_k p)= a\otimes_k
(bp). \label{+a21}
\eeq
There exists the module morphism
\mar{5.44}\beq
J^k: P\ni p\mapsto \bb\otimes_k p\in \cJ^k(P) \label{5.44}
\eeq
of the $\cA$-module $P$ to the $\cA^\bll$-module $\cJ^k(P)$ such
that $\cJ^k(P)$, seen as an $\cA$-module, is generated by elements
$J^kp$, $p\in P$.

Due to the natural monomorphisms $\m^r\to \m^s$ for all $r>s$,
there are $\cA$-module epimorphisms of jet modules
\mar{t4}\beq
\pi^{i+1}_i: \cJ^{i+1}(P)\to \cJ^i(P). \label{t4}
\eeq
In particular,
\mar{+a13}\beq
\pi^1_0:\cJ^1(P) \ni a\ot_1 p\mapsto ap \in P.\label{+a13}
\eeq

The above mentioned relation between differential operators on
modules and jets of modules is stated by the following theorem
\cite{kras}.

\begin{theo} \label{t6} \mar{t6}
Any $Q$-valued differential operator $\Delta$ of order $k$ on an
$\cA$-module $P$ factorizes uniquely
\be
\Delta: P\ar^{J^k} \cJ^k(P)\ar Q
\ee
through the morphism $J^k$ (\ref{5.44}) and some $\cA$-module
homomorphism ${\got f}^\Delta: \cJ^k(P)\to Q$.
\end{theo}

The proof is based on the fact that the morphism $J^k$
(\ref{5.44}) is a $k$-order $\cJ^k(P)$-valued differential
operator on $P$. Let us denote
\be
J: P\ni p\mapsto \bb\ot p\in \cA\ot P.
\ee
Then, for any ${\got f}\in\hm_\cA (\cA\ot P,Q)$, we obtain
\be
\dl_b({\got f}\circ J)(p)={\got f}(\dl^b(\bb\ot p)).
\ee
The correspondence $\Delta\mapsto {\got f}^\Delta$ defines an
$\cA$-module isomorphism
\mar{5.50}\beq
\dif_s(P,Q)=\hm_{\cA}(\cJ^s(P),Q). \label{5.50}
\eeq

\section{Connections on modules and rings}

We employ  the jets of modules in previous Section in order to
introduce connections on modules and commutative rings
\cite{book00}.

Let us consider the jet modules $\cJ^s=\cJ^s(\cA)$ of the ring
$\cA$ itself. In particular, the first order jet module $\cJ^1$
consists of the elements $a\otimes_1 b$, $a,b\in\cA$, subject to
the relations
\mar{5.53}\beq
ab\otimes_1 \bb -b\otimes_1 a -a\otimes_1 b +\bb\ot_1(ab) =0.
\label{5.53}
\eeq
The $\cA$- and $\cA^\bll$-module structures (\ref{+a21}) on
$\cJ^1$ read
\be
c(a\ot_1 b)=(ca)\ot_1 b,\qquad c\bll(a\ot_1 b)=
a\ot_1(cb)=(a\ot_1b)c.
\ee
Besides the monomorphism
\be
J^1: \cA\ni a\mapsto \bb\otimes_1 a\in \cJ^1
\ee
(\ref{5.44}), there is the $\cA$-module monomorphism
\be
i_1: \cA \ni a  \mapsto a\otimes_1 \bb\in \cJ^1.
\ee
With these monomorphisms, we have the canonical $\cA$-module
splitting
\mar{mos058}\ben
&& \cJ^1=i_1(\cA)\oplus \cO^1, \label{mos058} \\
&& aJ^1(b)= a\ot_1 b=ab\ot_1\bb + a(\bb\ot_1 b- b\ot_1\bb),
\nonumber
\een
where the $\cA$-module $\cO^1$ is generated by the elements
$\bb\ot_1 b-b\ot_1 \bb$ for all $b\in\cA$. Let us consider the
corresponding $\cA$-module epimorphism
\mar{+216}\beq
h^1:\cJ^1\ni \bb\ot_1 b\mapsto \bb\ot_1 b-b\ot_1 \bb\in \cO^1
\label{+216}
\eeq
and the composition
\mar{mos045}\beq
d^1=h^1\circ J_1: \cA \ni b \mapsto \bb\ot_1 b- b\ot_1\bb \in
\cO^1, \label{mos045}
\eeq
which is a $\cK$-module morphism. This is a $\cO^1$-valued
derivation of the $\cK$-ring $\cA$ which obeys the Leibniz rule
\be
d^1(ab)= \bb\ot_1 ab-ab\ot_1\bb +a\ot_1 b  -a\ot_1 b  =ad^1b +
(d^1a)b.
\ee
It follows from the relation (\ref{5.53}) that $ad^1b=(d^1b)a$ for
all $a,b\in \cA$. Thus, seen as an $\cA$-module, $\cO^1$ is
generated by the elements $d^1a$ for all $a\in\cA$.

Let $\cO^{1*}=\hm_{\cA}(\cO^1,\cA)$ be the dual  of the
$\cA$-module $\cO^1$. In view of the splittings (\ref{spr156}) and
(\ref{mos058}), the isomorphism (\ref{5.50}) reduces to the
duality relation
\mar{5.81,a}\ben
&& \gd\cA=\cO^{1*}, \label{5.81}\\
&& \gd\cA\ni u\leftrightarrow
   \f_u\in \cO^{1*}, \qquad \f_u(d^1a)=u(a), \qquad  a\in \cA.
   \label{5.81a}
\een
In a more direct way (see Proposition \ref{w130} below), the
isomorphism (\ref{5.81}) is derived from the facts that $\cO^1$ is
generated by elements $d^1a$, $a\in\cA$, and that $\f(d^1a)$ is a
derivation of $\cA$ for any $\f\in\cO^{1*}$. However, the morphism
\be
\cO^1 \to \cO^{1**}=\gd\cA^*
\ee
need not be an isomorphism.

Let us define the modules $\cO^k$, $k=2,\ldots$, as the exterior
products of the $\cA$-module $\cO^1$. There are the higher degree
generalizations
\mar{5.4}\ben
&& h^k: \cJ^1(\cO^{k-1})\to \cO^k, \nonumber\\
&& d^k=h^k\circ J^1: \cO^{k-1}\to \cO^k \label{5.4}
\een
of the morphisms (\ref{+216}) and (\ref{mos045}). The operators
(\ref{5.4}) are nilpotent, i.e., $d^k\circ d^{k-1}=0$. They  form
the cochain complex
\mar{55.63}\beq
0\to \cK\ar\cA\ar^{d^1}\cO^1\ar^{d^2} \cdots  \cO^k\ar^{d^{k+1}}
\cdots\,. \label{55.63}
\eeq

Let us return to the first order jet module $\cJ^1(P)$ of an
$\cA$-module $P$. It is isomorphic to the tensor product
\mar{mos074}\beq
\cJ^1(P)=\cJ^1\ot P, \qquad (a\ot_1 bp) \leftrightarrow (a\ot_1
b)\ot p. \label{mos074}
\eeq
Then the isomorphism (\ref{mos058}) leads to the splitting
\mar{mos071}\ben
&& \cJ^1(P)= (\cA\oplus \cO^1)\ot P=
(\cA \ot P)\oplus (\cO^1\ot P), \label{mos071}\\
&& a\ot_1 bp\leftrightarrow  (ab +ad^1(b))\ot p. \nonumber
\een
Applying the epimorphism $\pi^1_0$ (\ref{+a13}) to this splitting,
one obtains the short exact sequence of $\cA$- and
$\cA^\bll$-modules
\mar{+175}\ben
&& 0\ar \cO^1\ot P\to \cJ^1(P)\ar^{\pi^1_0} P\ar 0, \label{+175}\\
&&  (a\ot_1 b -ab\ot_1 \bb)\ot p\to
(c\ot_1 \bb+ a\ot_1 b -ab\ot_1 \bb)\ot p \to cp. \nonumber
\een
This exact sequence is canonically split by the $\cA^\bll$-module
morphism
\be
P\ni ap \mapsto \bb\ot ap= a\ot p + d^1(a)\ot p\in\cJ^1(P).
\ee
However, it need not be split by an $\cA$-module morphism, unless
$P$ is a projective $\cA$-module.

\begin{defi} \label{+176} \mar{+176}
A {\it connection} \index{connection!on a module} on an
$\cA$-module $P$ is defined as an $\cA$-module morphism
\mar{+179'}\beq
\G:P\to \cJ^1(P), \qquad \G (ap)=a\G(p), \label{+179'}
\eeq
which splits the exact sequence (\ref{+175}) or, equivalently, the
exact sequence
\mar{+183}\beq
0\to \cO^1\ot P\to (\cA\oplus \cO^1)\ot P\to P\to 0. \label{+183}
\eeq
\end{defi}

If a splitting $\G$ (\ref{+179'}) exists, it reads
\mar{+178}\beq
J^1p=\G(p) + \nabla(p), \label{+178}
\eeq
where $\nabla$ is the complementary morphism
\mar{+179}\beq
\nabla: P\to \cO^1\ot P, \qquad \nabla(p)= \bb\ot_1 p-
\G(p).\label{+179}
\eeq
Though this complementary morphism in fact is a {\it covariant
differential} \index{covariant differential!on a module} on the
module $P$, it is traditionally called a connection on a module.
It satisfies the {\it Leibniz rule} \index{Leibniz rule}
\mar{+180}\beq
\nabla(ap)= d^1a \ot p +a\nabla(p), \label{+180}
\eeq
i.e., $\nabla$ is  an $(\cO^1\ot P)$-valued first order
differential operator on $P$. Thus, we come to the equivalent
definition of a connection \cite{kosz60}.

\begin{defi} \label{+181} \mar{+181}
A {\it connection} \index{connection!on a module} on an
$\cA$-module $P$ is a $\cK$-module morphism $\nabla$ (\ref{+179})
which obeys the Leibniz rule (\ref{+180}). Sometimes, it is called
the {\it Koszul connection}. \index{Koszul connection}
\end{defi}

The morphism $\nabla$ (\ref{+179}) can be extended naturally to
the morphism
\be
\nabla: \cO^1\ot P\to \cO^2\ot P.
\ee
Then we have the morphism
\mar{+101}\beq
R=\nabla^2: P\to \cO^2 \ot P, \label{+101}
\eeq
called the {\it curvature} \index{curvature!of a connection!on a
module} of the connection $\nabla$ on a module $P$.

In view of the isomorphism (\ref{5.81}), any connection in
Definition \ref{+181} determines a connection in the following
sense.

\begin{defi} \label{1016} \mar{1016}
A {\it connection} \index{connection!on a module} on an
$\cA$-module $P$ is an $\cA$-module morphism
\mar{1017}\beq
\gd\cA\ni u\mapsto \nabla_u\in \dif_1(P,P) \label{1017}
\eeq
such that the first order differential operators $\nabla_u$ obey
the {\it Leibniz rule} \index{Leibniz rule}
\mar{1018}\beq
\nabla_u (ap)= u(a)p+ a\nabla_u(p), \quad a\in \cA, \quad p\in P.
\label{1018}
\eeq
\end{defi}

Definitions \ref{+181} and \ref{1016} are equivalent if
$\cO^1=\gd\cA^*$.

The {\it curvature} \index{curvature!of a connection!on a module}
of the connection (\ref{1017}) is defined as a zero order
differential operator
\mar{+100}\beq
R(u,u')=[\nabla_u,\nabla_{u'}] -\nabla_{[u,u']} \label{+100}
\eeq
on the module $P$ for all $u,u'\in \gd\cA$.

Let $P$ be a commutative $\cA$-ring and $\gd P$ the derivation
module of $P$ as a $\cK$-ring. Definition \ref{1016} is modified
as follows.

\begin{defi} \label{mos088} \mar{mos088}
A {\it connection} \index{connection!on a ring} on an $\cA$-ring
$P$ is an $\cA$-module morphism
\mar{mos090}\beq
\gd\cA\ni u\mapsto \nabla_u\in \gd P, \label{mos090}
\eeq
which is a connection on $P$ as an $\cA$-module, i.e., obeys the
Leinbniz rule (\ref{1018}).
\end{defi}

Two such connections $\nabla_u$ and $\nabla'_u$ differ from each
other in a derivation of the $\cA$-ring $P$, i.e., which vanishes
on $\cA\subset P$. The curvature of the connection (\ref{mos090})
is given by the formula (\ref{+100}).

\section{Differential calculus over a commutative ring}

In a general setting, the de Rham complex is defined as a cochain
complex which is also a differential graded algebra. By a
gradation throughout this Section is meant the $\Bbb N$-gradation.

A {\it graded algebra} \index{algebra!$\Bbb N$-graded} $\Om^*$
over a commutative ring $\cK$ is defined as a direct sum
\be
\Om^*= \op\oplus_k \Om^k
\ee
of $\cK$-modules $\Om^k$, provided with an associative
multiplication law $\al\cdot\bt$, $\al,\bt\in \Om^*$, such that
$\al\cdot\bt\in \Om^{|\al|+|\bt|}$, where $|\al|$ denotes the
degree of an element $\al\in \Om^{|\al|}$. In particular, it
follows that $\Om^0$ is a (non-commutative) $\cK$-algebra $\cA$,
while $\Om^{k>0}$ are $\cA$-bimodules and $\Om^*$ is an
$(\cA-\cA)$-algebra. A graded algebra  is said to be {\it graded
commutative} \index{algebra!$\Bbb N$-graded!commutative} if
\be
\al\cdot\bt=(-1)^{|\al||\bt|}\bt\cdot \al, \qquad \al,\bt\in
\Om^*.
\ee

A graded algebra $\Om^*$ is called a {\it differential graded
algebra} \index{differential graded algebra} if it is a cochain
complex of $\cK$-modules
\mar{spr260}\beq
0\to \cK\ar\cA\ar^\dl\Om^1\ar^\dl\cdots\Om^k\ar^\dl\cdots
\label{spr260}
\eeq
with respect to a coboundary operator $\dl$ which obeys the {\it
graded Leibniz rule} \index{Leibniz rule!graded}
\mar{1006}\beq
\dl(\al\cdot\bt)=\dl\al\cdot\bt +(-1)^{|\al|}\al\cdot \dl\bt.
\label{1006}
\eeq
In particular, $\dl:\cA\to \Om^1$ is a $\Om^1$-valued derivation
of a $\cK$-algebra $\cA$.

The cochain complex (\ref{spr260}) is the above mentioned {\it de
Rham complex} \index{de Rham complex} of the differential graded
algebra $(\Om^*,\dl)$. This algebra is also said to be a {\it
differential calculus} \index{differential calculus} over $\cA$.
Cohomology $H^*(\Om^*)$ of the complex (\ref{spr260}) is called
the {\it de Rham cohomology} \index{de Rham cohomology} of a
differential graded algebra. It is a graded algebra with respect
to the {\it cup-product} \index{cup-product}
\mar{spr268}\beq
[\al]\smile [\bt]=[\al\cdot\bt], \label{spr268}
\eeq
where $[\al]$ denotes the de Rham cohomology class of elements
$\al\in \Om^*$.

A morphism $\g$ between two differential graded algebras
$(\Om^*,\dl)$ and $(\Om'^*,\dl')$ is defined as a cochain
morphism, i.e., $\g\circ\dl=\g\circ \dl'$. It yields the
corresponding morphism of the de Rham cohomology groups of these
algebras.

One considers the minimal differential graded subalgebra
$\Om^*\cA$ of the differential graded algebra $\Om^*$ which
contains $\cA$. Seen as an $(\cA-\cA)$-algebra, it is generated by
the elements $\dl a$, $a\in \cA$, and consists of monomials
\be
\al=a_0\dl a_1\cdots \dl a_k, \qquad a_i\in \cA,
\ee
whose product obeys the {\it juxtaposition rule}
\index{juxtaposition rule}
\be
(a_0\dl a_1)\cdot (b_0\dl b_1)=a_0\dl (a_1b_0)\cdot \dl b_1-
a_0a_1\dl b_0\cdot \dl b_1
\ee
in accordance with the equality (\ref{1006}). The differential
graded algebra $(\Om^*\cA,\dl)$ is called the {\it minimal
differential calculus} \index{differential calculus!minimal} over
$\cA$.

Let us show that any commutative $\cK$-ring $\cA$ defines a
differential calculus.

As was mentioned above, the derivation module $\gd\cA$ of $\cA$ is
also a Lie $\cK$-algebra. Let us consider the extended
Chevalley--Eilenberg complex
\mar{ws102}\beq
0\to \cK\ar^{\rm in}C^*[\gd\cA;\cA] \label{ws102}
\eeq
of the Lie algebra $\gd\cA$ with coefficients in the ring $\cA$,
regarded as a $\gd\cA$-module \cite{book05}. This complex contains
a subcomplex $\cO^*[\gd\cA]$ of $\cA$-multilinear skew-symmetric
maps
\mar{+840'}\beq
\f^k:\op\times^k \gd\cA\to \cA \label{+840'}
\eeq
with respect to the Chevalley--Eilenberg coboundary operator
\mar{+840}\ben
&& d\f(u_0,\ldots,u_k)=\op\sum^k_{i=0}(-1)^iu_i
(\f(u_0,\ldots,\wh{u_i},\ldots,u_k)) +\label{+840}\\
&& \qquad \op\sum_{i<j} (-1)^{i+j}
\f([u_i,u_j],u_0,\ldots, \wh u_i, \ldots, \wh u_j,\ldots,u_k).
\nonumber
\een
Indeed, a direct verification shows that if $\f$ is an
$\cA$-multilinear map, so is $d\f$. In particular,
\mar{spr708}\ben
&& (d a)(u)=u(a), \qquad a\in\cA, \qquad u\in\gd\cA, \label{spr708}\\
&&(d\f)(u_0,u_1)= u_0(\f(u_1)) -u_1(\f(u_0))
\label{+921}\\
&& \qquad -\f([u_0,u_1]), \qquad \f\in \cO^1[\gd\cA], \nonumber \\
&& \cO^0[\gd\cA]=\cA, \qquad
\cO^1[\gd\cA]=\hm_\cA(\gd\cA,\cA).\nonumber
\een \mar{+921}
It follows that $d(\bb)=0$ and $d$ is a $\cO^1[\gd\cA]$-valued
derivation of $\cA$.

The graded module $\cO^*[\gd\cA]$ is provided with the structure
of a graded $\cA$-algebra with respect to the product
\mar{ws103}\ben
&& \f\w\f'(u_1,...,u_{r+s})= \label{ws103}\\
&& \qquad \op\sum_{i_1<\cdots<i_r;j_1<\cdots<j_s} {\rm
sgn}^{i_1\cdots i_rj_1\cdots j_s}_{1\cdots r+s} \f(u_{i_1},\ldots,
u_{i_r}) \f'(u_{j_1},\ldots,u_{j_s}), \nonumber \\
&& \f\in \cO^r[\gd\cA], \qquad \f'\in \cO^s[\gd\cA], \qquad u_k\in \gd\cA,
\nonumber
\een
where sgn$^{...}_{...}$ is the sign of a permutation. This product
obeys the relations
\mar{ws98,9}\ben
&& d(\f\w\f')=d(\f)\w\f' +(-1)^{|\f|}\f\w d(\f'),
\quad \f,\f'\in \cO^*[\gd\cA], \label{ws98}\\
&& \f\w \f' =(-1)^{|\f||\f'|}\f'\w \f. \label{ws99}
\een
By virtue of the first one, $(\cO^*[\gd\cA],d)$ is a differential
graded $\cK$-algebra, called the {\it Chevalley--Eilenberg
differential calculus} \index{Chevalley--Eilenberg!differential
calculus} over a $\cK$-ring $\cA$ \cite{book05}. The relation
(\ref{ws99}) shows that $\cO^*[\gd\cA]$ is a graded commutative
algebra.

The {\it minimal Chevalley--Eilenberg differential calculus}
$\cO^*\cA$  over a ring \index{Chevalley--Eilenberg!differential
calculus!minimal} $\cA$ consists of the monomials
\be
a_0da_1\w\cdots\w da_k, \qquad a_i\in\cA.
\ee
Its complex
\mar{t10}\beq
0\to\cK\ar \cA\ar^d\cO^1\cA\ar^d \cdots  \cO^k\cA\ar^d \cdots
\label{t10}
\eeq
is exactly the cochain complex (\ref{55.63}). Indeed, comparing
the equalities (\ref{5.81a}) and (\ref{spr708}) shows that $d^1=d$
on the $\cA$-module
\be
\cO^1\cA=\cO^1\subseteq \cO^1[\gd\cA]=\gd\cA^*,
\ee
generated by elements $d^1a$, $a\in\cA$. The complex (\ref{t10})
is called the {\it de Rham complex} \index{de Rham complex!of a
ring} of the $\cK$-ring $\cA$, and its cohomology $H^*(\cA)$ is
said to be the {\it de Rham cohomology} \index{de Rham
cohomology!of a ring} of $\cA$. This cohomology is a graded
commutative algebra with respect to the cup-product (\ref{spr268})
induced by the exterior product $\w$ of elements of $\cO^*\cA$ so
that
\mar{spr268'}\ben
&& [\f]\smile [\f']= [\f\w \f'], \label{spr268'}\\
&& [\f]\smile [\f']=(-1)^{|[\f]||[\f']|}[\f']\smile [\f]. \nonumber
\een

\section{Local-ringed spaces}

Local-ringed spaces are sheafs of local rings. For instance,
smooth manifolds, represented by sheaves of real smooth functions,
make up a subcategory of the category of local-ringed spaces.

A sheaf $\gR$ on a topological space $X$ is said to be a {\it
ringed space} \index{ringed space} if its stalk $\gR_x$ at each
point $x\in X$ is a commutative ring \cite{tenn}. A ringed space
is often denoted by a pair $(X,\gR)$ of a topological space $X$
and a sheaf $\gR$ of rings on $X$ which are called the {\it body}
\index{body!of a ringed space} and the {\it structure sheaf}
\index{structure sheaf!of a ringed space} of a ringed space,
respectively.

A ringed space is said to be a {\it local-ringed space}
\index{local-ringed space} (a {\it geometric space}
\index{geometric space} in the terminology of \cite{tenn}) if it
is a sheaf of local rings.

For instance, the sheaf $C^0_X$ of continuous real functions on a
topological space $X$ is a local-ringed space. Its stalk $C^0_x$,
$x\in X$, contains the unique maximal ideal of germs of functions
vanishing at $x$.

Morphisms of local-ringed spaces are defined to be particular
morphisms of sheaves on different topological spaces as follows.

Let $\vf:X\to X'$ be a continuous map. Given a sheaf $S$ on $X$,
its {\it direct image} \index{direct image of a sheaf} $\vf_*S$ on
$X'$ is generated by the presheaf of assignments
\be
X'\supset U'\mapsto S(\vf^{-1}(U'))
\ee
for any open subset $U'\subset X'$. Conversely, given a sheaf $S'$
on $X'$, its {\it inverse image} \index{inverse image of a sheaf}
$\vf^*S'$ on $X$ is defined as the pull-back onto $X$ of the
topological fibre bundle $S'$ over $X'$, i.e.,
$\vf^*S'_x=S_{\vf(x)}$. This sheaf is generated by the presheaf
which associates to any open $V\subset X$ the direct limit of
modules $S'(U)$ over all open subsets $U\subset X'$ such that
$V\subset f^{-1}(U)$.

\begin{ex} \label{spr201} \mar{spr201}
Let $i:X\to X'$ be a closed subspace of $X'$. Then $i_*S$ is a
unique sheaf on $X'$ such that
\be
i_*S|_X=S, \qquad i_*S|_{X'\setminus X}=0.
\ee
Indeed, if $x'\in X\subset X'$, then $i_*S(U')= S(U'\cap X)$ for
any open neighborhood $U$ of this point. If $x'\not\in X$, there
exists its neighborhood $U'$ such that $U'\cap X$ is empty, i.e.,
$i_*S(U')=0$. The sheaf $i_*S$ is called the {\it trivial
extension} \index{trivial extension of a sheaf} of the sheaf $S$.
\end{ex}

By a {\it morphism of ringed spaces} \index{morphism!of ringed
spaces} $(X,\gR)\to (X',\gR')$ is meant a pair $(\vf,\Phi)$ of a
continuous map $\vf:X\to X'$ and a sheaf morphism $\Phi:\gR'\to
\vf_*\gR$ or, equivalently, a sheaf morphisms $\vf^*\gR'\to \gR$
\cite{tenn}. Restricted to each stalk, a sheaf morphism $\Phi$ is
assumed to be a ring homomorphism. A morphism of ringed spaces is
said to be:

$\bullet$ a monomorphism if $\vf$ is an injection and $\Phi$ is an
epimorphism,

$\bullet$ an epimorphism if $\vf$ is a surjection, while $\Phi$ is
a monomorphism.

Let $(X,\gR)$ be a local-ringed space. By a {\it sheaf $\gd \gR$
of derivations} \index{sheaf!of derivations} of the sheaf $\gR$ is
meant a subsheaf of endomorphisms of $\gR$ such that any section
$u$ of $\gd \gR$ over an open subset $U\subset X$ is a derivation
of the ring $\gR(U)$. It should be emphasized that, since
(\ref{+212}) is not necessarily an isomorphism, a derivation of
the ring $\gR(U)$ need not be a section of the sheaf $\gd \gR|_U$.
Namely, it may happen that, given open sets $U'\subset U$, there
is no restriction morphism
\be
\gd (\gR(U)) \to\gd (\gR(U')).
\ee

Given a local-ringed space $(X,\gR)$, a sheaf $P$ on $X$ is called
a {\it sheaf of $\gR$-modules} \index{sheaf!of modules} if every
stalk $P_x$, $x\in X$, is an $\gR_x$-module or, equivalently, if
$P(U)$ is an $\gR(U)$-module for any open subset $U\subset X$. A
sheaf of $\gR$-modules $P$ is said to be {\it locally free}
\index{sheaf!locally free} if there exists an open neighborhood
$U$ of every point $x\in X$ such that $P(U)$ is a free
$\gR(U)$-module. If all these free modules are of finite rank
(resp. of the same finite rank), one says that $P$ is of {\it
finite type} \index{sheaf!locally free!of finite type} (resp. of
constant rank). \index{sheaf!locally free!of constant rank} The
structure module of a locally free sheaf is called a {\it locally
free module}. \index{locally free module}

The following is a generalization of Proposition \ref{spr256}
\cite{hir}.

\begin{prop} \label{spr256'} \mar{spr256'} Let $X$ be a paracompact
space which admits a partition of unity by elements of the
structure module $S(X)$ of some sheaf $S$ of real functions on
$X$. Let $P$ be a sheaf of $S$-modules. Then $P$ is fine and,
consequently, acyclic.
\end{prop}

\section{Differential geometry of $C^\infty(X)$-modules}

The sheaf $C^\infty_X$ of smooth real functions on a smooth
manifold $X$ provides an important example of a local-ringed
spaces.

\begin{rem} \label{ws1} \mar{ws1}
Throughout the Lectures, {\it smooth manifolds} \index{smooth
manifold} are finite-dimensional real manifolds, though
infinite-dimensional Banach and Hilbert manifolds are also smooth.
A smooth real manifold is customarily assumed to be Hausdorff and
{\it second-countable} \index{second-countable topological space}
(i.e., it has a countable base for topology). Consequently, it is
a locally compact space which is a union of a countable number of
compact subsets, a {\it separable} space \index{separable
topological space} (i.e., it has a countable dense subset), a
paracompact and completely regular space. Being paracompact, a
smooth manifold admits a partition of unity by smooth real
functions. One also can show that, given two disjoint closed
subsets $N$ and $N'$ of a smooth manifold $X$, there exists a
smooth function $f$ on $X$ such that $f|_N=0$ and $f|_{N'}=1$.
Unless otherwise stated, manifolds are assumed to be connected
and, consequently, arcwise  connected. We follow the notion of a
manifold without boundary.
\end{rem}

Similarly to the sheaf $C^0_X$ of continuous functions, the stalk
$C^\infty_x$ of the sheaf $C^\infty_X$ at a point $x\in X$ has a
unique maximal ideal of germs of smooth functions vanishing at
$x$. Though the sheaf $C^\infty_X$ is defined on a topological
space $X$, it fixes a unique smooth manifold structure on $X$ as
follows.

\begin{theo} \label{+26} \mar{+26}
Let $X$ be a paracompact topological space and $(X,\gR)$ a
local-ringed space. Let $X$ admit an open cover $\{U_i\}$ such
that the sheaf $\gR$ restricted to each $U_i$ is isomorphic to the
local-ringed space  $(\Bbb R^n, C^\infty_{R^n})$. Then $X$ is an
$n$-dimensional smooth manifold together with a natural
isomorphism of local-ringed spaces $(X,\gR)$ and $(X,C^\infty_X)$.
\end{theo}

One can think of this result as being an alternative definition of
smooth real manifolds in terms of local-ringed spaces. In
particular, there is one-to-one correspondence between smooth
manifold morphisms $X\to X'$ and the $\Bbb R$-ring morphisms
$C^\infty(X')\to C^\infty(X)$.

\begin{rem} \label{ws2} \mar{ws2}
Let $X\times X'$ be a manifold product. The ring $C^\infty(X\times
X')$ is constructed from the rings $C^\infty(X)$ and
$C^\infty(X')$ as follows. Whenever referring to a topology on the
ring $C^\infty(X)$, we will mean the topology of compact
convergence for all derivatives \cite{rob}. The $C^\infty(X)$ is a
{\it Fr\'echet ring} \index{Fr\'echet ring} with respect to this
topology, i.e., a complete  metrizable locally convex topological
vector space. There is an isomorphism of Fr\'echet rings
\mar{+55}\beq
C^\infty(X)\wh\ot C^\infty(X') \cong C^\infty(X\times X'),
\label{+55}
\eeq
where the left-hand side, called the {\it topological tensor
product}, \index{tensor product!topological} is the completion of
$C^\infty(X)\ot C^\infty(X')$ with respect to Grothendieck's
topology, defined as follows. If $E_1$ and $E_2$ are locally
convex topological vector spaces, {\it Grothendieck's topology}
\index{Grothendieck's topology} is the finest locally convex
topology on $E_1\ot E_2$  such that the canonical mapping of
$E_1\times E_2$ to $E_1\ot E_2$ is continuous \cite{rob}. It is
also called the $\pi$-topology in contrast with the coarser
$\ve$-topology on $E_1\ot E_2$ \cite{piet,trev}. Furthermore, for
any two open subsets $U\subset X$ and $U'\subset X'$, let us
consider the topological tensor product of rings
$C^\infty(U)\wh\ot C^\infty(U')$. These tensor products define a
locally ringed space $(X\times X',C^\infty_X\wh\ot
C^\infty_{X'})$. Due to the isomorphism (\ref{+55}) written for
all $U\subset X$ and $U'\subset X'$, we obtain the sheaf
isomorphism
\mar{+56}\beq
C^\infty_X\wh\ot C^\infty_{X'}=C^\infty_{X\times X'}. \label{+56}
\eeq
\end{rem}

Since a smooth manifold admits a partition of unity by smooth
functions, it follows from Proposition \ref{spr256'} that any
sheaf of $C^\infty_X$-modules on $X$ is fine and, consequently,
acyclic.

For instance, let $Y\to X$ be a smooth (finite-dimensional) vector
bundle. The germs of its sections make up a sheaf of
$C^\infty_X$-modules, called the {\it structure sheaf} $S_Y$
\index{structure sheaf!of a vector bundle} of a vector bundle
$Y\to X$. The sheaf $S_Y$ is fine.

In particular, all sheafs $\cO^k_X$, $k\in\Bbb N_+$, of germs of
exterior forms on $X$ is fine. These sheaves constitute the {\it
de Rham complex} \index{de Rham complex!of sheaves}
\mar{t67}\beq
0\to \Bbb R\ar C^\infty_X\ar^d \cO^1_X\ar^d\cdots \cO^k_X\ar^d
\cdots. \label{t67}
\eeq
The corresponding complex of structure modules of these sheaves is
the {\it de Rham complex} \index{de Rham complex!of exterior
forms}
\mar{t37}\beq
0\to \Bbb R\ar C^\infty(X)\ar^d \cO^1(X)\ar^d\cdots \cO^k(X)\ar^d
\cdots \label{t37}
\eeq
of exterior forms on a manifold $X$. Its cohomology is called the
{\it de Rham cohomology} \index{de Rham cohomology!of a manifold}
$H^*(X)$ of $X$. Due to the Poincar\'e lemma, the complex
(\ref{t67}) is exact and, thereby, is a fine resolution of the
constant sheaf $\Bbb R$ on a manifold. Then a corollary of Theorem
\ref{spr230} is the classical {\it de Rham theorem}. \index{de
Rham theorem}

\begin{theo} \label{t60} \mar{t60} There is the isomorphism
\mar{t61}\beq
H^k(X)=H^k(X;\Bbb R) \label{t61}
\eeq
of the de Rham cohomology $H^*(X)$ of a manifold $X$ to cohomology
of $X$ with coefficients in the constant sheaf $\Bbb R$.
\end{theo}

\begin{rem}
Let us consider the short exact sequence of constant sheaves
\mar{1320}\beq
0\to \Bbb Z\ar \Bbb R\ar U(1)\to 0, \label{1320}
\eeq
where $U(1)=\Bbb R/\Bbb Z$ is the circle group of complex numbers
of unit modulus. This exact sequence yields the long exact
sequence of the sheaf cohomology groups
\be
&& 0\to \Bbb Z\ar\Bbb R \ar U(1) \ar
H^1(X;\Bbb Z) \ar H^1(X;\Bbb R)\ar\cdots \\
&& \qquad  H^p(X;\Bbb Z)\ar H^p(X;\Bbb R)\ar H^p(X;U(1))\ar
H^{p+1}(X;\Bbb Z) \ar\cdots,
\ee
where
\be
H^0(X;\Bbb Z)=\Bbb Z, \qquad H^0(X;\Bbb R)=\Bbb R
\ee
and $H^0(X;U(1))=U(1)$. This exact sequence defines the
homomorphism
\mar{spr752}\beq
H^*(X;\Bbb Z)\to H^*(X;\Bbb R) \label{spr752}
\eeq
of cohomology with coefficients in the constant sheaf $\Bbb Z$ to
that with coefficients in $\Bbb R$. Combining the isomorphism
(\ref{t61}) and the homomorphism (\ref{spr752}) leads to the
cohomology homomorphism
\mar{t62}\beq
H^*(X;\Bbb Z)\to H^*(X). \label{t62}
\eeq
Its kernel contains all cyclic elements of cohomology groups
$H^k(X;\Bbb Z)$.
\end{rem}

Given a vector bundle $Y\to X$, the structure module of the sheaf
$S_Y$ coincides with the {\it structure module} \index{structure
module!of a vector bundle} $Y(X)$ of global sections of $Y\to X$.
The {\it Serre--Swan theorem}, \index{Serre--Swan theorem} shows
that these modules exhaust all projective modules of finite rank
over $C^\infty(X)$. This theorem originally has been proved in the
case of a compact manifold $X$, but it is generalized to an
arbitrary smooth manifold \cite{book05}.

\begin{theo} \label{sp60} \mar{sp60}
Let $X$ be a smooth manifold. A $C^\infty(X)$-module $P$ is
isomorphic to the structure module of a smooth vector bundle over
$X$ iff it is a projective module of finite rank.
\end{theo}

$\bullet$ The structure module $Y^*(X)$ of the dual $Y^*\to X$ of
a vector bundle $Y\to X$ is the $C^\infty(X)$-dual $Y(X)^*$ of the
structure module $Y(X)$ of $Y\to X$.

$\bullet$ Any exact sequence of vector bundles
\mar{t51}\beq
0\to Y \ar Y'\ar Y''\to 0 \label{t51}
\eeq
over the same base $X$ yields the exact sequence
\mar{t52}\beq
0\to Y(X) \ar Y'(X)\ar Y''(X)\to 0 \label{t52}
\eeq
of their structure modules, and {\it vice versa}. In accordance
with the well-known theorem \cite{book00,sard09}, the exact
sequence (\ref{t51}) is always split. Every its splitting defines
that of the exact sequence (\ref{t52}), and {\it vice versa}.

$\bullet$ For instance, the derivation module of the $\Bbb R$-ring
$C^\infty(X)$ coincides with the $C^\infty(X)$-module $\cT_1(X)$
of vector fields on $X$, i.e., with the structure module of the
tangent bundle $TX$ of $X$. Hence, it is a projective
$C^\infty(X)$-module of finite rank. It is the $C^\infty(X)$-dual
$\cT_1(X)=\cO^1(X)^*$ of the structure module $\cO^1(X)$ of the
cotangent bundle $T^*X$ of $X$ which is the module of differential
one-forms on $X$ and, conversely, $\cO^1(X)=\cT_1(X)^*$. It
follows that the Chevalley--Eilenberg differential calculus over
the $\Bbb R$-ring $C^\infty(X)$ is exactly the differential graded
algebra $(\cO^*(X),d)$ of exterior forms on $X$, where the
Chevalley--Eilenberg coboundary operator $d$ (\ref{+840})
coincides with the exterior differential. Accordingly, the de Rham
complex (\ref{t10}) of the $\Bbb R$-ring $C^\infty(X)$ is the de
Rham complex (\ref{t37}) of exterior forms on $X$. Moreover, one
can show that $(\cO^*(X),d)$ is a  minimal differential calculus,
i.e., the $C^\infty(X)$-module $\cO^1(X)$ is generated by elements
$df$, $f\in C^\infty(X)$. Indeed, using the notation in the proof
of Theorem \ref{sp60}, one can write
\mar{spr882}\beq
\cO^1(X)\ni \f=\op\sum_\xi l_\xi^2 \f= \op\sum_\xi l_\xi^2 \f_\m
dx^\m= \op\sum_\xi (l_\xi \f_\m d(l_\xi x^\m) -l_\xi \f_\m x^\m d
l_\xi), \label{spr882}
\eeq
where $(x^\m)$ are local coordinates on $U_\xi$ and $l_\xi x^\m$
and $l_\xi$ are functions on $X$.

\begin{rem}
Let us note that the above mentioned Chevalley--Eilenberg
differential calculus over the $\Bbb R$-ring $C^\infty(X)$ is a
subcomplex  of the Chevalley--Eilenberg complex of the Lie algebra
$\cT_1(X)$ with coefficients in $C^\infty(X)$. It consists of
skew-symmetric morphisms of $\cT_1(X)$ to $C^\infty(X)$ which are
not only $\Bbb R$-multilinear, but $C^\infty(X)$-multilinear. The
Chevalley--Eilenberg cohomology of smooth vector fields with
coefficients in a trivial representation and in spaces of smooth
tensor fields has been studied in detail \cite{fuks}
\end{rem}

$\bullet$ Let $Y\to X$ be a vector bundle and $Y(X)$ its structure
module. The $r$-order jet manifold $J^rY$ of $Y\to X$ consists of
the equivalence classes $j^r_xs$, $x\in X$, of sections $s$ of
$Y\to X$ which are identified by the $r+1$ terms of their Taylor
series at points $x\in X$. Since $Y\to X$ is a vector bundle, so
is the jet bundle $J^rY\to X$. Its structure module $J^rY(X)$ is
exactly the $r$-order jet module $\cJ^r(Y(X))$ of the
$C^\infty(X)$-module $Y(X)$ in Section 1.2 \cite{kras}. As a
consequence, the notion of a connection on the structure module
$Y(X)$ is equivalent to the standard geometric notion of a
connection on a vector bundle $Y\to X$ \cite{book00}. Indeed,
connection on a fibre bundle $Y\to X$ is defined as a global
section $\G$ of the affine jet bundle $J^1Y\to Y$. If $Y\to X$ is
a vector bundle, there exists the exact sequence
\mar{t50}\beq
0\to T^*X\op\ot_XY\ar J^1Y\ar Y\to 0 \label{t50}
\eeq
over $X$ which is split by $\G$. Conversely, any slitting of this
exact sequence yields a connection $Y\to X$. The exact sequence of
vector bundles (\ref{t50}) induces the exact sequence of their
structure modules
\mar{t53}\beq
0\to \cO^1(X)\op\ot Y(X)\ar J^1Y(X)\ar Y(X)\to 0. \label{t53}
\eeq
Then any connection $\G$ on a vector bundle $Y\to X$ defines a
splitting of the exact sequence (\ref{t53}) which, by Definition
\ref{+176}, is a connection on the $C^\infty(X)$-module $Y(X)$,
and {\it vice versa}.

Let now $P$ be an arbitrary $C^\infty(X)$-module. One can
reformulate Definitions \ref{+181} and \ref{1016} of a connection
on $P$ as follows.

\begin{defi} \label{t55} \mar{t55}
A connection on a $C^\infty(X)$-module $P$ is a
$C^\infty(X)$-module morphism
\mar{t56}\beq
\nabla: P\to \cO^1(X)\ot P, \label{t56}
\eeq
which satisfies the Leibniz rule
\be
\nabla(fp)=df\ot p +f\nabla(p), \qquad f\in C^\infty(X), \qquad
p\in P.
\ee
\end{defi}

\begin{defi} \label{t57} \mar{t57}
A connection on a $C^\infty(X)$-module $P$ associates to any
vector field $\tau\in\cT_1(X)$ on $X$ a first order differential
operator $\nabla_\tau$ on $P$ which obeys the Leibniz rule
\mar{t58}\beq
\nabla_\tau(fp)=(\tau\rfloor df)p +f\nabla_\tau p. \label{t58}
\eeq
\end{defi}

Since $\cO^1(X)=\cT_1(X)^*$, Definitions \ref{t55} and \ref{t57}
are equivalent.

Let us note that a connection on an arbitrary $C^\infty(X)$-module
need not exist, unless it is a projective or locally free module
(see Theorem \ref{w715} below).

The curvature of a connection $\nabla$ in Definitions \ref{t55}
and \ref{t57} is defined as the zero-order differential operator
\mar{t59}\beq
R(\tau,\tau')=[\nabla_\tau,\nabla_{\tau'}]-\nabla_{[\tau,\tau']}
\label{t59}
\eeq
on a module $P$ for all vector fields $\tau,\tau'\in\cT_1(X)$ on
$X$.

\section{Connections on local-ringed spaces}

Let $(X,\gR)$ be a local-ringed space and $\gP$ a sheaf of
$\gR$-modules on $X$.  For any open subset $U\subset X$, let us
consider the jet module $\cJ^1(\gP(U))$ of the module $\gP(U)$. It
consists of the elements of $\gR(U)\ot \gP(U)$ modulo the
pointwise relations (\ref{mos041}). Hence, there is the
restriction morphism
\be
\cJ^1(\gP(U))\to \cJ^1(\gP(V))
\ee
for any open subsets $V\subset U$, and the jet modules
$\cJ^1(\gP(U))$ constitute a presheaf. This presheaf defines the
{\it sheaf $\gj^1\gP$ of jets} \index{sheaf!of jets} of $\gP$ (or
simply the {\it jet sheaf}). \index{jet sheaf} The jet sheaf
$\gj^1\gR$ of the sheaf $\gR$ of local rings is introduced in a
similar way. Since the relations (\ref{mos041}) and (\ref{5.53})
on the ring $\gR(U)$ and modules $\gP(U)$, $\cJ^1(\gP(U))$,
$\cJ^1(\gR(U))$  are pointwise relations for any open subset
$U\subset X$, they commute with the restriction morphisms.
Therefore, the direct limits of the quotients modulo these
relations exist \cite{massey}. Then we have the sheaf $\cO^1\gR$
of one-forms over the sheaf $\gR$, the sheaf isomorphism
\be
\gj^1(\gP)=(\gR\oplus \cO^1\gR)\ot \gP,
\ee
and the exact sequences of sheaves
\mar{+213}\ben
&&0\to \cO^1\gR\ot \gP\to \gj^1(\gP)\to \gP\to 0, \label{+213}\\
&& 0\to \cO^1\gR\ot \gP\to (\gR\oplus \cO^1\gR)\ot \gP\to \gP\to 0. \label{+214}
\een
They reflect the quotient (\ref{+216}), the isomorphism
(\ref{mos071}) and the exact sequences of modules (\ref{+175}),
(\ref{+183}), respectively.

\begin{rem}
It should be emphasized that, because of the inequality
(\ref{+212}), the duality relation (\ref{5.81}) is not extended to
the sheaves $\gd \gR$ and $\cO^1\gR$ in general, unless $\gd \gR$
and $\cO^1$ are locally free sheaves of finite rank. If $\gP$ is a
locally free sheaf of finite rank, so is $\gj^1\gP$.
\end{rem}

Following Definitions \ref{+176}, \ref{+181} of a connection on
modules, we come to the following notion of a connection on
sheaves.

\begin{defi} \label{+217} \mar{+217}
Given a local-ringed space $(X,\gR)$ and a sheaf $\gP$ of
$\gR$-modules on $X$, a {\it connection} \index{connection!on a
sheaf} on a sheaf $\gP$ is defined as a splitting of the exact
sequence (\ref{+213}) or, equivalently, the exact sequence
(\ref{+214}).
\end{defi}

Theorem \ref{spr30} leads to the following compatibility of the
notion of a connection on sheaves with that of a connection on
modules.

\begin{prop} \label{+219} \mar{+219}
If there exists a connection on a sheaf $\gP$ in Definition
\ref{+217}, then there exists a connection on a module $\gP(U)$
for any open subset $U\subset X$. Conversely, if for any open
subsets $V\subset U\subset X$ there are connections on the modules
$\gP(U)$ and $\gP(V)$ related by the restriction morphism, then
the sheaf $\gP$ admits a connection.
\end{prop}

\begin{ex}
Let $Y\to X$ be a vector bundle. Every linear connection $\G$ on
$Y\to X$ defines a connection on the structure module $Y(X)$ such
that the restriction $\G|_U$ is a connection on the module $Y(U)$
for any open subset $U\subset X$. Then we have a connection on the
structure sheaf $Y_X$. Conversely, a connection on the structure
sheaf $Y_X$ defines a connection on the module $Y(X)$ and,
consequently, a connection on the vector bundle $Y\to X$.
\end{ex}

As an immediate consequence of Proposition \ref{+219}, we find
that the exact sequence of sheaves (\ref{+214}) is split iff there
exists a sheaf morphism
\mar{+3}\beq
\nabla: \gP\to \cO^1\gR\ot \gP, \label{+3}
\eeq
satisfying the Leibniz rule
\be
\nabla (fs)=df\ot s + f\nabla(s), \qquad f\in \cA(U), \qquad s\in
\gP(U),
\ee
for any open subset $U\in X$. It leads to the following equivalent
definition of a connection on sheaves in the spirit of Definition
\ref{+181}.

\begin{defi} \label{+4} \mar{+4}
The sheaf morphism (\ref{+3}) is a {\it connection} on the sheaf
$\gP$.
\end{defi}

Similarly to the case of connections on modules, the {\it
curvature} \index{curvature!of a connection on sheaves} of the
connection (\ref{+3}) on a sheaf $\gP$ is given by the expression
\mar{+105}\beq
R=\nabla^2:\gP\to \cO^2_X\ot \gP. \label{+105}
\eeq

The exact sequence (\ref{+214}) need not be split. One can obtain
the following criteria of the existence of a connection on a
sheaf.

Let $\gP$ be a locally free sheaf of $\gR$-modules. Then we have
the exact sequence of sheaves
\be
0\to \hm(\gP,\cO^1\gR\ot \gP)\to \hm(\gP,(\gR\oplus\cO^1\gR)\ot
\gP) \to \hm(\gP,\gP)\to 0
\ee
and the corresponding exact sequence (\ref{spr227}) of the
cohomology groups
\be
&& 0\to H^0(X;\hm(\gP,\cO^1\gR\ot \gP)) \to H^0(X;
\hm(\gP,(\gR\oplus\cO^1\gR)\ot \gP))\to \\
&& \qquad H^0(X;\hm(\gP,\gP))\to H^1(X;\hm(\gP,\cO^1\gR\ot \gP))\to \cdots.
\ee
The identity morphism $\id :\gP\to \gP$ belongs to
$H^0(X;\hm(\gP,\gP))$. Its image in
\be
H^1(X;\hm(\gP,\cO^1\gR\ot \gP))
\ee
is called the {\it Atiyah class}. \index{ Atiyah class} If this
class vanishes,  there exists an element of
\be
\hm(\gP,(\gR\oplus\cO^1\gR)\ot \gP))
\ee
whose image is $\id \gP$, i.e., a splitting of the exact sequence
(\ref{+214}).

In particular, let $X$ be a manifold and $\gR=C^\infty_X$ the
sheaf of smooth functions on $X$.  The sheaf $\gd C^\infty_X$ of
its derivations is isomorphic to the sheaf of vector fields on a
manifold $X$. It follows that:

$\bullet$ there is the restriction morphism $\gd(C^\infty(U))\to
\gd(C^\infty(V))$ for any open sets $V\subset U$,

$\bullet$ $\gd C^\infty_X$ is  a locally free sheaf of
$C^\infty_X$-modules of finite rank,

$\bullet$ the sheaves $\gd C^\infty_X$ and $\cO^1_X$ are mutually
dual.

Let $\gP$ be a locally free sheaf of $C^\infty_X$-modules. In this
case, $\hm(\gP,\cO^1_X\ot \gP)$ is a locally free sheaf of
$C^\infty_X$-modules. It is fine and acyclic. Its  cohomology
group
\be
H^1(X;\hm(\gP,\cO^1_X\ot \gP))
\ee
vanishes, and the exact sequence
\mar{+2}\beq
0\to \cO^1_X\ot \gP\to (C^\infty_X\oplus \cO^1_X)\ot \gP\to \gP\to
0 \label{+2}
\eeq
admits a splitting. This proves the following.

\begin{prop} \label{w715} \mar{w715}
Any locally free sheaf of $C^\infty_X$-modules on a manifold $X$
admits a connection and, in accordance with Proposition
\ref{+219}, any locally free $C^\infty(X)$-module does well.
\end{prop}

In conclusion, let us consider a sheaf $S$ of commutative
$C^\infty_X$-rings on a manifold $X$. Basing on Definition
\ref{mos088}, we come to the following notion of a connection on a
sheaf $S$ of commutative $C^\infty_X$-rings.

\begin{defi} \label{+7} \mar{+7}
Any morphism
\be
\gd C^\infty_X \ni\tau \mapsto \nabla_\tau\in \gd S,
\ee
which is a connection on $S$ as a sheaf of $C^\infty_X$-modules,
is called a connection on the sheaf $S$ of rings.
\end{defi}

Its curvature is given by the expression
\mar{+106}\beq
R(\tau,\tau')=[\nabla_\tau,\nabla_{\tau'}]-\nabla_{[\tau,\tau']},
\label{+106}
\eeq
similar to the expression (\ref{+100}) for the curvature of a
connection on modules.

\chapter{Geometry of quantum systems}

Algebraic quantum theory usually deals with Hilbert spaces. This
Chapter addresses differential geometry of Banach and Hilbert
manifolds and, in particular, Hilbert bundles and bundles of
$C^*$-algebras over a smooth manifolds $X$. For instance, this is
the case of time-dependent quantum systems (where $X=\Bbb R$)
(Section 2.5) and quantum models depending on classical parameters
(Section 2.6). Their differential geometry is similar to
differential geometry of finite-dimensional smooth manifolds and
bundles in main, and it is formulated in algebraic terms of
differential geometry of modules and, in particular,
$C^\infty(X)$-modules.

\section{Geometry of Banach manifolds}

We start with the notion of a real Banach manifold
\cite{lang95,vais73}. Banach manifolds are defined similarly to
finite-dimensional smooth manifolds, but they are modelled on
Banach spaces, not necessarily finite-dimensional.

Let us recall some particular properties of (infinite-dimensional)
real Banach spaces. Let us note that a finite-dimensional Banach
space is always provided with an Euclidean norm.

$\bullet$ Given Banach spaces $E$ and $H$, every continuous
bijective linear map of $E$ to $H$ is an isomorphism of
topological vector spaces.

$\bullet$ Given a Banach space $E$, let $F$ be its closed
subspace. One says that $F$ {\it splits} \index{split (subspace)}
in $E$ if there exists a closed complement $F'$ of $F$ such that
$E\cong F\oplus F'$. In particular, finite-dimensional and
finite-codimensional subspaces split in $E$. As a consequence, any
subspace of a finite-dimensional space splits.

$\bullet$ Let $E$ and $H$ be Banach spaces and $f:E\to H$ a
continuous injection. One says that $f$ {\it splits} \index{split
(injection)} if there exists an isomorphism
\be
g:H\to H_1\times H_2
\ee
such that $g\circ f$ yields an isomorphism of $E$ onto
$H_1\times\{0\}$.

$\bullet$ Given Banach spaces $(E,\|.\|_E)$ and $(H,\|.\|_H)$, one
can provide the set $\hm^0(E,H)$ of continuous linear morphisms of
$E$ to $H$ with the norm
\mar{w370}\beq
||f||=\op\sup_{||z||_E=1}||f(z)||_H,  \qquad f\in \hm^0(E,H).
\label{w370}
\eeq
If $E$, $H$ and $F$ are Banach spaces, the bilinear map
\be
\hm^0(E,F)\times \hm^0(F,H)\to \hm^0(E,H),
\ee
obtained by the composition $f\circ g$ of morphisms $\g\in
\hm^0(E,F)$ and $f\in \hm^0(F,H)$, is continuous. Let us note that
this assertion is false for more general spaces, e.g., the
Fr\'echet ones.

$\bullet$ Let $(E,\|.\|_E)$ and $(H,\|.\|_H)$ be real Banach
spaces. One says that a continuous map $f: E\to H$ (not
necessarily linear and isometric) is a {\it differentiable
function} \index{differentiable function!between Banach spaces}
between $E$ and $H$ if, given a point $z\in E$, there exists an
$\Bbb R$-linear continuous map
\be
df(z): E\to H
\ee
(not necessarily isometric) such that
\be
&& f(z')=f(z) +df(z)(z'-z) +o(z'-z), \\
&& \op\lim_{\|z'-z\|_E\to 0}
\frac{\|o(z'-z)\|_H}{\|z'-z\|_E}=0,
\ee
for any $z'$ in some open neighborhood $U$ of $z$. For instance,
any continuous linear morphism $f$ of $E$ to $H$ is differentiable
and $df(z)z=f(z)$. The linear map $df(z)$ is called a {\it
differential} \index{differential on a Banach space} of $f$ at a
point $z\in U$. Given an element $v\in E$, we obtain the map
\mar{spr706}\beq
E\ni z\mapsto \dr_vf(z)=df(z)v\in H, \label{spr706}
\eeq
called the {\it derivative} \index{derivative!on a Banach space}
of a function $f$ along a vector $v\in E$. One says that $f$ is
two-times differentiable if the map (\ref{spr706}) is
differentiable for any $v\in E$. Similarly, $r$-times
differentiable and infinitely differentiable (smooth) functions on
a Banach space are defined. The composition of smooth maps is a
smooth map.

The following {\it inverse mapping theorem} \index{inverse mapping
theorem} enables one to consider smooth Banach manifolds and
bundles similarly to the finite-dimensional ones.

\begin{theo} \label{w371} \mar{w371}
Let $f: E\to H$ be a smooth map such that, given a point $z\in E$,
the differential $df(z):E\to H$ is an isomorphism of topological
vector spaces. Then $f$ is a local isomorphism at $z$.
\end{theo}

Let us turn to the notion of a Banach manifold, without repeating
the statements true both for finite-dimensional and Banach
manifolds.

\begin{defi} \label{w372} \mar{w372}
A {\it Banach manifold} \index{Banach manifold} $\cB$ modelled on
a Banach space $B$ is defined as a topological space which admits
an atlas of charts $\Psi_\cB=\{(U_\iota,\f_\iota)\}$, where the
maps $\f_\iota$ are homeomorphisms of $U_\iota$ onto open subsets
of the Banach space $B$, while the transition functions
$\f_\zeta\f_\iota^{-1}$ from $\f_\iota(U_\iota\cap U_\zeta)\subset
B$ to $\f_\zeta(U_\iota\cap U_\zeta)\subset B$ are smooth. Two
atlases of a Banach manifold are said to be equivalent if their
union is also an atlas.
\end{defi}

Unless otherwise stated, Banach manifolds are assumed to be
connected paracompact Hausdorff topological spaces. A locally
compact Banach manifold is necessarily finite-dimensional.

\begin{rem} \label{w740} \mar{w740}
Let us note that a paracompact Banach manifold admits a smooth
partition of unity iff its model Banach space does. For instance,
this is the case of (real) separable Hilbert spaces. Therefore, we
restrict our consideration to Hilbert manifolds modelled on
separable Hilbert spaces.
\end{rem}

Any open subset $U$ of a Banach manifold $\cB$ is a Banach
manifold whose atlas is the restriction of an atlas of $\cB$ to
$U$.

Morphisms of Banach manifolds are defined similarly to those of
smooth finite-dimensional manifolds. However, the notion of the
immersion and submersion need a certain modification (see
Definition \ref{w374} below).

Tangent vectors to a smooth Banach manifold $\cB$ are introduced
by analogy with tangent vectors to a finite-dimensional one. Given
a point $z\in\cB$, let us consider the pair
$(v;(U_\iota,\f_\iota))$ of a vector $v\in B$ and a chart
$(U_\iota\ni z,\f_\iota)$ on a Banach manifold $\cB$. Two pairs
$(v;(U_\iota,\f_\iota))$ and $(v';(U_\zeta,\f_\zeta))$  are said
to be equivalent if
\mar{spr705}\beq
v'=d(\f_\zeta\f_\iota^{-1})(\f_\iota(z))v. \label{spr705}
\eeq
The equivalence classes of such pairs make up the {\it tangent
space} \index{tangent space!to a Banach manifold} $T_z\cB$ to a
Banach manifold $\cB$ at a point $z\in\cB$. This tangent space is
isomorphic to the topological vector space $B$. Tangent spaces to
a Banach manifold $\cB$ are assembled into the {\it tangent
bundle} $T\cB$ of $\cB$. \index{tangent bundle!of a Banach
manifold} It is a Banach manifold modelled over the Banach space
$B\oplus B$ which possesses the transition functions
\be
(\f_\zeta\f_\iota^{-1},d(\f_\zeta\f_\iota^{-1})).
\ee

Any morphism $f:\cB\to\cB'$ of Banach manifolds yields the
corresponding tangent morphism  of the tangent bundles $Tf:T\cB\to
T\cB'$.

\begin{defi} \label{w374} \mar{w374}
Let $f:\cB\to \cB'$ be a  morphism of Banach manifolds.

(i) It is called an immersion at a point $z\in\cB$ if the tangent
morphism $Tf$ at $z$ is injective and splits.

(ii) A morphism $f$ is called a submersion at a point $z\in\cB$ if
$Tf$ at $z$ is surjective and its kernel splits.
\end{defi}

In the case of finite-dimensional smooth manifolds, the split
conditions are superfluous, and Definition \ref{w374} recovers the
notion of the immersion and submersion of smooth manifolds.

The range of a surjective submersion $f$ of a Banach manifold is a
submanifold, though $f$ need not be an isomorphism onto a
submanifold, unless $f$ is an imbedding.

One can think of a surjective submersion $\pi:\cB\to\cB'$ of
Banach manifolds as a {\it fibred Banach manifold}. \index{fibred
manifold Banach} For instance, the product $\cB\times \cB'$ of
Banach manifolds is a fibred Banach manifold with respect to
$\pr_1$ and $\pr_2$.

Let $\cB$ be a Banach manifold and $E$ a Banach space. The
definition of a (locally trivial) vector bundle with the typical
fibre $E$ and the base $\cB$ is a repetition of that of
finite-dimensional smooth vector bundles. Such a vector bundle $Y$
is a Banach manifold and $Y\to \cB$ is a surjective submersion.
The above mentioned tangent bundle $T\cB$ of a Banach manifold
exemplifies a vector bundle over $\cB$.

Let {\bf Bnh} be the {\it category of Banach spaces}
\index{category {\bf Bnh}} and ${\bf Vect}(\cB)$ denotes the {\it
category of vector bundles over a Banach manifold} $\cB$
\index{category ${\bf Vect}(\cB)$} with respect to their morphisms
over $\id \cB$. Let
\mar{w375}\beq
F:{\bf Bnh}\times {\bf Bnh}\to {\bf Bnh} \label{w375}
\eeq
 be a functor of two variables which is covariant in the first
and contravariant in the second. Then there exists a functor
\be
VF: {\bf Vect}(\cB)\times {\bf Vect}(\cB)\to {\bf Vect}(\cB)
\ee
such that, if $Y_E\to\cB$ and $Y_H\to\cB$ are vector bundles with
the typical fibres $E$ and $H$, then $VF(Y_E,Y_H)$ is a vector
bundle with the typical fibre $F(E,H)$. For instance, the Whitney
sum, the tensor product, and the exterior product of vector
bundles over a Banach manifold are defined in this way. In
particular, since the topological dual $E'$ of a Banach space $E$
is a Banach space, one can associate to each vector bundle
$Y_E\to\cB$ the dual $Y^*_E=Y_{E'}$ with the typical fibre $E'$.
For instance, the dual of the tangent bundle $T\cB$ of a Banach
manifold $\cB$ is the {\it cotangent bundle} \index{cotangent
bundle!of a Banach manifold} $T^*\cB$.

Sections of the tangent bundle $T\cB\to\cB$ of a Banach manifold
are called {\it vector fields} on a Banach manifold \index{vector
field!on a Banach manifold}  $\cB$. They form a locally free
module $\cT_1(\cB)$ over the ring $C^\infty(\cB)$ of smooth real
functions on $\cB$. Every vector field $\vt$ on a Banach manifold
$\cB$ determines a derivation of the $\Bbb R$-ring $C^\infty(\cB)$
by the formula
\be
f(z)\mapsto \dr_\vt f(z)=df(z)\vt(z), \qquad  z\in \cB.
\ee
Different vector fields yield different derivations. It follows
that $\cT_1(\cB)$ possesses a structure of a real Lie algebra, and
there is its monomorphism
\mar{w376}\beq
\cT_1(\cB)\to \gd C^\infty(\cB) \label{w376}
\eeq
to the derivation module of the $\Bbb R$-ring $C^\infty(\cB)$.

Let us consider the Chevalley--Eilenberg complex of the real Lie
algebra $\cT_1(\cB)$ with coefficients in $C^\infty(\cB)$ and its
subcomplex $\cO^*[\cT_1(\cB)]$ of $C^\infty(\cB)$-multilinear
skew-symmetric maps by analogy with the complex $\cO^*[\gd\cA]$ in
Section 1.4 \cite{book05}. This subcomplex is a differential
calculus over a $\Bbb R$-ring $C^\infty(\cB)$ where the
Chevalley--Eilenberg coboundary operator $d$ (\ref{+840}) and the
product (\ref{ws103}) read
\mar{spr713,'}\ben
&& d\f(\vt_0,\ldots,\vt_r)=\op\sum^r_{i=0}(-1)^i\dr_{\vt_i}
(\f(\vt_0,\ldots,\wh{\vt_i},\ldots,\vt_r)) +\label{spr713}\\
&& \qquad \op\sum_{i<j} (-1)^{i+j}
\f([\vt_i,\vt_j],\vt_0,\ldots,
\wh \vt_i, \ldots, \wh \vt_j,\ldots,\vt_k), \nonumber\\
&& \f\w\f'(\vt_1,...,\vt_{r+s})= \label{spr713'}\\
&& \qquad \op\sum_{i_1<\cdots<i_r;j_1<\cdots<j_s}
{\rm sgn}^{i_1\cdots i_rj_1\cdots j_s}_{1\cdots r+s}
\f(\vt_{i_1},\ldots, \vt_{i_r})
\f'(\vt_{j_1},\ldots,\vt_{j_s}), \nonumber \\
&& \f\in
\cO^r[\cT_1(\cB)], \qquad \f'\in \cO^s[\cT_1(\cB)], \qquad
 \vt_i\in\cT_1(\cB). \nonumber
\een
There are the familiar relations
\be
&& \vt\rfloor df=\dr_\vt f, \qquad f\in C^\infty(\cB), \qquad \vt\in
\cT_1(\cB), \\
&& d(\f\w\f')=d\f\w\f' +(-1)^{|\f|}\f\w d\f',
\qquad \f,\f'\in \cO^*[\cT_1(\cB)].
\ee

The differential calculus $\cO^*[\cT_1(\cB)]$ contains the
following subcomplex. Let $\cO^1(\cB)$ be the
$C^\infty(\cB)$-module of global sections of the cotangent bundle
$T^*\cB$ of $\cB$. Obviously, there is its monomorphism
\mar{w377}\beq
\cO^1(\cB)\to \gd C^\infty(\cB)^* \label{w377}
\eeq
to the dual of the derivation module $\gd C^\infty(\cB)$.
Furthermore, let $\op\w^r T^*\cB$ be the $r$-degree exterior
product of the cotangent bundle $T^*\cB$ and $\cO^r(\cB)$ the
$C^\infty(\cB)$-module of its sections. Let $\cO^*(\cB)$ be the
direct sum of $C^\infty(\cB)$-modules $\cO^r(\cB)$, $r\in\Bbb N$,
where we put $\cO^0(\cB)=C^\infty(\cB)$. Elements of $\cO^*(\cB)$
are obviously $C^\infty(\cB)$-multilinear skew-symmetric maps of
$\cT_1(\cB)$ to $C^\infty(\cB)$. Therefore, the
Chevalley--Eilenberg differential $d$ (\ref{spr713}) and the
exterior product (\ref{spr713'}) of elements of $\cO^*(\cB)$ are
well defined. Moreover, one can show that $d\f$ and $\f\w\f'$,
$\f,\f'\in \cO^*(\cB)$, are also elements of $\cO^*(\cB)$. Thus,
$\cO^*(\cB)$ is a differential graded commutative algebra, called
the algebra of {\it exterior forms} \index{exterior form!on a
Banach manifold} on a Banach manifold $\cB$.

At the same time, one can consider Chevalley--Eilenberg
differential calculus $\cO^*[\gd C^\infty(\cB)]$ over the $\Bbb
R$-ring $C^\infty(\cB)$. Because of the monomorphism (\ref{w376}),
we have the homomorphism of $C^\infty(\cB)$-modules
\mar{w377'}\beq
\cO^1[\gd C^\infty(\cB)]=\gd C^\infty(\cB)^*\to \cT_1(\cB)^*
=\cO^1[\cT_1(\cB)]\lto \cO^1(\cB). \label{w377'}
\eeq
It follows that the differential calculi $\cO^*[\cT_1(\cB)]$,
$\cO^*(\cB)$ and $\cO^1[\gd C^\infty(\cB)]$ over the $\Bbb R$-ring
$C^\infty(\cB)$ are not mutually isomorphic in general. However,
it is readily observed that the minimal differential calculi in
$\cO^*[\cT_1(\cB)]$ and $\cO^*(\cB)$ coincide with the minimal
Chevalley--Eilenberg differential calculus $\cO^*C^\infty(\cB)$
over the $\Bbb R$-ring $C^\infty(\cB)$ because they
 are generated by the elements $df$,
$f\in C^\infty(\cB)$, where $d$ is the restriction (\ref{spr713})
to $\cT_1(\cB)$ of the Chevalley--Eilenberg coboundary operator
(\ref{+840}).

A {\it connection} on a Banach manifold $\cB$ \index{connection!on
a Banach manifold} is defined as a connection on the
$C^\infty(\cB)$-module $\cT_1(\cB)$ \cite{book05,vais73}. In
accordance with Definition \ref{+181}, it is an $\Bbb R$-module
morphism
\be
\nabla: \cT_1(\cB)\to \cO^1C^\infty(\cB) \ot\cT_1(\cB),
\ee
which obeys the Leibniz rule
\mar{w385}\beq
\nabla(f\vt)= df\ot\vt +f\nabla(\vt), \qquad f\in  C^\infty(\cB),
\qquad \vt\in\cT_1(\cB). \label{w385}
\eeq
In view of the inclusions,
\be
\cO^1C^\infty(\cB)\subset \cO^1(\cB)\subset \cT_1(\cB)^*, \qquad
\cT_1(\cB)\subset \cT_1(\cB)^{**}\subset \cO^1(\cB)^*,
\ee
it is however convenient to define a connection on a Banach
manifold as an $\Bbb R$-module morphism
\mar{w386}\beq
\nabla: \cT_1(\cB)\to \cO^1(\cB) \ot\cT_1(\cB), \label{w386}
\eeq
which obeys the Leibniz rule (\ref{w385}).

\section{Geometry of Hilbert manifolds}

Let us turn now to Hilbert manifolds. These are particular Banach
manifolds modelled on complex Hilbert spaces, which are assumed to
be separable (see Remark \ref{w740}).

\begin{rem} \label{w705} \mar{w705} We refer the reader to \cite{lang95}
for the theory of real Hilbert and (infinite-dimensional)
Riemannian manifolds. A real Hilbert manifold is a Banach manifold
$\cB$ modelled on a real Hilbert space $V$. It is assumed to be
connected Hausdorff and paracompact space admitting the partition
of unity by smooth functions (this is the case of a separable
$V$). In infinite-dimensional geometry, the most of local results
follow from general arguments analogous to those in the
finite-dimensional case. The global theory of real Hilbert
manifolds is more intricate.
\end{rem}

A complex Hilbert space $(E,\lng.|.\rng)$ can be seen as a real
Hilbert space
\be
E\ni v\mapsto v_\Bbb R\in E_\Bbb R, \qquad (v_\Bbb R,v'_\Bbb R)=
{\rm Re}\,\lng v|v'\rng,
\ee
equipped with the complex structure $Jv_\Bbb R=(iv)_\Bbb R$. We
have
\be
(Jv_\Bbb R,Jv'_\Bbb R)=(v_\Bbb R,v'_\Bbb R), \qquad (Jv_\Bbb
R,v'_\Bbb R)={\rm Im}\,(v'_\Bbb R,v_\Bbb R).
\ee
Let $E_\Bbb C=\Bbb C\ot E_\Bbb R$ denote the complexification of
$E_\Bbb R$ provided with the Hermitian form $\lng.|.\rng_\Bbb C$.
The complex structure $J$ on $E_\Bbb R$ is naturally extended to
$E_\Bbb C$ by letting $J\circ i= i\circ J$. Then $E_\Bbb C$ is
split into the two complex subspaces
\mar{spr720}\ben
&& E_\Bbb C=E^{1,0}\oplus E^{0,1}, \label{spr720}\\
&& E^{1,0}=\{v_\Bbb R-iJv_\Bbb R\,:\, v_\Bbb R\in
E_\Bbb R\}, \nonumber\\
&& E^{0,1}=\{v_\Bbb R+iJv_\Bbb R\,:\, v_\Bbb R\in
E_\Bbb R\}, \nonumber
\een
which are mutually orthogonal with respect to the Hermitian form
$\lng.|.\rng_\Bbb C$. Since
\be
\lng v_\Bbb R-iJv_\Bbb R)|v'_\Bbb R-iJv'_\Bbb R\rng = 2\lng
v|v'\rng, \qquad \lng v_\Bbb R+iJv_\Bbb R|v'_\Bbb R+iJv'_\Bbb
R\rng = 2\lng v'|v\rng,
\ee
there are the following linear and antilinear isometric bijections
\be
&& E\ni v\mapsto v_\Bbb R\to \frac1{\sqrt{2}}(v_\Bbb R-iJv_\Bbb R)\in E^{1,0}, \\
&& E\ni v\mapsto v_\Bbb R\to \frac1{\sqrt{2}}(v_\Bbb R+iJv_\Bbb R)\in E^{0,1}.
\ee
They make $E^{1,0}$ and $E^{0,1}$ isomorphic to the Hilbert space
$E$ and the dual Hilbert space $\ol E$, respectively. Hence, the
decomposition (\ref{spr720}) takes the form
\mar{spr722}\beq
E_\Bbb C=E\oplus \ol E. \label{spr722}
\eeq
The complex structure $J$ on the direct sum (\ref{spr722}) reads
\mar{A22}\beq
J: E\oplus \ol E\ni v+\ol u\mapsto iv-i\ol u\in E\oplus \ol E,
\label{A22}
\eeq
where $E$ and $\ol E$ are the (holomorphic and antiholomorphic)
eigenspaces of $J$ characterized by the eigenvalues $i$ and $-i$,
respectively.

Let $f$ be a function (not necessarily linear) from a Hilbert
space $E$ to a Hilbert space $H$. It is said to be {\it
differentiable} \index{differentiable function!on a Hilbert space}
if the corresponding function $f_\Bbb R$ between the real Banach
spaces $E_\Bbb R$ and $H_\Bbb R$ is differentiable. Let $df_\Bbb
R(z)$, $z\in E_\Bbb R$, be the differential (\ref{spr706}) of
$f_\Bbb R$ on $E_\Bbb R$ which is a continuous linear morphism
\be
E_\Bbb R\ni v_\Bbb R \mapsto df_\Bbb R(z)v_\Bbb R \in H_\Bbb R
\ee
between real topological vector spaces $E_\Bbb R$ and $H_\Bbb R$.
This morphism is naturally extended to the $\Bbb C$-linear
morphism
\mar{spr725}\beq
E_\Bbb C\ni v_\Bbb C \mapsto df_\Bbb R(z)v_\Bbb C \in H_\Bbb C
\label{spr725}
\eeq
between the complexifications of $E_\Bbb R$ and $H_\Bbb R$. In
view of the decomposition (\ref{spr722}), one can introduce the
$\Bbb C$-linear maps
\be
\dr f_\Bbb R(z)(v+\ol u)= df_\Bbb R(z)v, \qquad
   \ol \dr f(z)(v+\ol u)= df_\Bbb R(z)\ol u
\ee
 from $E\oplus \ol E$ to $H_\Bbb C$ such that
\be
df_\Bbb R(z)v_\Bbb C=df_\Bbb R(z)(v+\ol u)=\dr f_\Bbb R(z)v
+\ol\dr f_\Bbb R(z)\ol u.
\ee
Let us split
\be
f_\Bbb R(z)=f(z) +\ol f(z)
\ee
in accordance with the decomposition $H_\Bbb C=H\oplus \ol H$.
Then the morphism (\ref{spr725}) takes the form
\mar{spr726}\beq
df_\Bbb R(z)(v+\ol u)=\dr f(z)v +\ol\dr f(z)\ol u +\dr\ol f(z)v
+\ol\dr\, \ol f(z)\ol u, \label{spr726}
\eeq
where $\dr \ol f=\ol{\ol\dr f}$, $\ol\dr\,\ol f=\ol{\dr f}$. A
function $f:E\to H$ is said to be {\it holomorphic}
\index{holomorphic function on a Hilbert space} (resp. {\it
antiholomorphic}) \index{antiholomorphic function on a Hilbert
space} if it is differentiable and $\ol\dr f(z)=0$ (resp. $\dr
f(z)=0$) for all $z\in E$. A holomorphic function is smooth, and
is given by the Taylor series. If $f$ is a holomorphic function,
then the morphism (\ref{spr726}) is split into the sum
\be
df_\Bbb R(z)(v+\ol u)=\dr f(z)v + \ol\dr\, \ol f(z)\ol u
\ee
of morphisms $E\to H$ and $\ol E\to\ol H$.

\begin{ex} \label{spr730} \mar{spr730}
Let $f$ be a complex function on a Hilbert space $E$. Then
\be
f_\Bbb R=(\re f,\im f)
\ee
is a map of $E$ to $\Bbb R^2$. The differential $df_\Bbb R(z)$,
$z\in E$, of $f_\Bbb R$ yields the complex linear morphism
\be
E\oplus \ol E\ni v_\Bbb C\mapsto (d\re f(z)v_\Bbb C,d\im
f(z)v_\Bbb C)\mapsto d(\re f+i\im f)(z)v_\Bbb C\in \Bbb C,
\ee
which is regarded as a differential $df(z)$ of a complex function
$f$ on a Hilbert space $E$.
\end{ex}

A {\it Hilbert manifold} \index{Hilbert manifold} $\cP$ modelled
on a Hilbert space $E$ is defined as a real Banach manifold
modelled on the Banach space $E_\Bbb R$ which admits an atlas
$\{(U_\iota,\f_\iota)\}$ with holomorphic transition functions
$\f_\zeta\f_\iota^{-1}$. Let $CT\cP$ denote the {\it complexified
tangent bundle} \index{tangent bundle!complexified } of a Hilbert
manifold $\cP$. In view of the decomposition (\ref{spr722}), each
fibre $CT_z\cP$, $z\in \cP$, of $CT\cP$ is split into the direct
sum
\be
CT_z\cP=T_z\cP\oplus\ol T_z\cP
\ee
of subspaces $T_z\cP$ and $\ol T_z\cP$, which are topological
complex vector spaces isomorphic to the Hilbert space $E$ and the
dual Hilbert space $\ol E$, respectively. The spaces $CT_z\cP$,
$T_z\cP$ and $\ol T_z\cP$ are respectively called the {\it
complex}, \index{tangent space!complex} {\it holomorphic}
\index{tangent space!holomorphic} and {\it antiholomorphic tangent
spaces}  \index{tangent space!antiholomorphic} to a Hilbert
manifold $\cP$ at a point $z\in\cP$. Since transition functions of
a Hilbert manifold are holomorphic, the complex tangent bundle
$CT\cP$ is split into a sum
\be
CT\cP=T\cP\oplus \ol T\cP
\ee
of {\it holomorphic} \index{tangent bundle!holomorphic} and {\it
antiholomorphic} \index{tangent bundle!antiholomorphic}
subbundles, together with the antilinear bundle automorphism
\be
T\cP\oplus \ol T\cP\ni v+\ol u\mapsto \ol v + u\in T\cP\oplus \ol
T\cP
\ee
and the complex structure
\mar{spr700}\beq
J: T\cP\oplus \ol T\cP\ni v+\ol u\mapsto iv-i\ol u\in T\cP\oplus
\ol T\cP. \label{spr700}
\eeq

Sections of the complex tangent bundle $CT\cP\to\cP$ are called
{\it complex vector fields} \index{complex vector field!on a
Hilbert manifold} on a Hilbert manifold $\cP$. They constitute the
locally free module $C\cT_1(\cP)$ over the ring $\Bbb
C^\infty(\cP)$ of smooth complex functions on $\cP$. Every complex
vector field $\vt +\ol\up$ on $\cP$ yields a derivation
\be
f(z)\to df(z)(\vt +\up)=\dr f(z)\vt(z) +\ol\dr f(z)\upsilon(z),
\qquad f\in \Bbb C^\infty(\cP), \quad z\in \cP,
\ee
of the $\Bbb C$-ring $\Bbb C^\infty(\cP)$.

The (topological) dual of the complex tangent bundle $CT\cP$ is
the {\it complex cotangent bundle} \index{cotangent bundle!
complex} $CT^*\cP$ of $\cP$. Its fibres $CT^*_z\cP$, $z\in \cP$,
are topological complex vector spaces isomorphic to $E\oplus \ol
E$.  Since Hilbert spaces are reflexive, the complex tangent
bundle $CT\cP$ is the dual of $CT^*\cP$. The complex cotangent
bundle $CT^*\cP$ is split into the sum
\mar{spr702}\beq
CT^*\cP=T^*\cP\oplus \ol T^*\cP \label{spr702}
\eeq
of {\it holomorphic} \index{cotangent bundle!holomorphic} and {\it
antiholomorphic} \index{cotangent bundle!antiholomorphic}
subbundles, which are the annihilators of antiholomorphic and
holomorphic tangent bundles $\ol T\cP$ and $T\cP$, respectively.
Accordingly, $CT^*\cP$ is provided with the complex structure $J$
via the relation
\be
\lng v,Jw\rng=\lng Jv,w\rng, \qquad  v\in CT_z\cP, \qquad w\in
CT^*_z\cP, \qquad z\in \cP.
\ee
Sections of the complex cotangent bundle $CT^*\cP\to \cP$
constitute a locally free $\Bbb C^\infty(\cP)$-module
$\cO^1(\cP)$. It is the $\Bbb C^\infty(\cP)$-dual
\mar{w380}\beq
\cO^1(\cP)=C\cT_1(\cP)^* \label{w380}
\eeq
of the module $C\cT_1(\cP)$ of complex vector fields on $\cP$, and
{\it vice versa}.

Similarly to the case of  a Banach manifold, let us consider the
differential calculi $\cO^*[\cT_1(\cP)]$, $\cO^*(\cP)$ (further
denoted by $\cC^*(\cP)$)  and $\cO^1[\gd C^\infty(\cP)]$ over the
$\Bbb C$-ring $\Bbb C^\infty(\cP)$. Due to the isomorphism
(\ref{w380}), $\cO^*[\cT_1(\cP)]$ is isomorphic to $\cC^*(\cP)$,
whose elements are called {\it exterior form!complex}
\index{exterior form!on a Hilbert manifold} on a Hilbert manifold
$\cP$. The exterior differential $d$ on these forms is the
Chevalley--Eilenberg coboundary operator
\mar{w381}\ben
&& d\f(\vt_0,\ldots,\vt_k)=\op\sum^k_{i=0}(-1)^id
\f(\vt_0,\ldots,\wh{\vt_i},\ldots,\vt_k)\vt_i + \label{w381}\\
&& \qquad \op\sum_{i<j} (-1)^{i+j}
\f([\vt_i,\vt_j],\vt_0,\ldots, \wh \vt_i, \ldots, \wh
\vt_j,\ldots,\vt_k), \qquad \vt_i\in C\cT_1(\cP). \nonumber
\een

In view of the splitting (\ref{spr702}), the differential graded
algebra $\cC^*(\cP)$ admits the decomposition
\be
\cC^*(\cP)=\op\oplus_{p,q=0}\cC^{p,q}(\cP)
\ee
into subspaces $\cC^{p,q}(\cP)$ of {\it $p$-holomorphic}
\index{holomorphic exterior form on a Hilbert manifold} and {\it
$q$-antiholomorphic} forms. \index{exterior
form!antiholomorphic!on a Hilbert manifold} Accordingly, the
exterior differential $d$ on $\cC^*(\cP)$ is split into a sum
$d=\dr +\ol\dr$ of holomorphic and antiholomorphic differentials
\be
&& \dr: \cC^{p,q}(\cP) \to \cC^{p+1,q}(\cP), \qquad \ol\dr:
\cC^{p,q}(\cP) \to \cC^{p,q+1}(\cP),\\
&& \dr\circ\dr=0, \qquad \ol\dr\circ\ol\dr=0, \qquad \dr\circ\ol\dr
+\ol\dr\circ\dr =0.
\ee

A {\it Hermitian metric} \index{Hermitian metric!on a Hilbert
manifold} on a Hilbert manifold $\cP$ is defined as a complex
bilinear form $g$ on fibres of the complex tangent bundle $CT\cP$
which obeys the following conditions:

$\bullet$ $g$ is a smooth section of the tensor bundle $CT^*\cP\ot
CT^*\cP\to\cP$;

$\bullet$ $g(\vt_z,\vt'_z)=0$ if complex tangent vectors
$\vt_z,\vt'_z\in CT_z\cP$ are simultaneously holomorphic or
antiholomorphic;

$\bullet$ $g(\vt_z,\ol\vt_z)> 0$ for any non-vanishing complex
tangent vector $\vt_z\in CT_z\cP$;

$\bullet$ the bilinear form $g(\vt_z,\vt'_z)$, $\vt_z,\vt'_z\in
CT_z\cP$, defines a norm topology on the complex tangent space
$CT_z\cP$ which is equivalent to its Hilbert space topology.

As an immediate consequence of this definition, we obtain
\be
\ol{g(\vt_z,\vt'_z)}= g(\ol \vt_z, \ol \vt'_z),\qquad
g(J\vt_z,J\vt'_z)= g(\vt_z, \vt'_z).
\ee
A Hermitian metric exists, e.g., on paracompact Hilbert manifolds
modelled on separable Hilbert spaces.

The above mentioned properties of a Hermitian metric on a Hilbert
manifold are similar to properties of a Hermitian metric on a
finite-dimensional complex manifold. Therefore, one can think of
the pair $(\cP,g)$ as being an infinite-dimensional {\it Hermitian
manifold}. \index{Hermitian manifold!infinite-dimensional}

A Hermitian manifold $(\cP,g)$ is endowed with a non-degenerate
exterior two-form
\mar{A26}\beq
\Om(\vt_z,\vt'_z) =g(J\vt_z,\vt'_z), \qquad  \vt_z,\vt'_z\in
CT_z\cP, \qquad z\in \cP, \label{A26}
\eeq
called the {\it fundamental form} \index{fundamental form!of a
Hermitian metric on a Hilbert manifold} of the Hermitian metric
$g$. This form satisfies the relations
\be
\ol{\Om(\vt_z,\vt'_z)}= \Om(\ol \vt_z, \ol \vt'_z),\qquad
\Om(J_z\vt_z,J_z\vt'_z)= \Om(\vt_z, \vt'_z).
\ee
If $\Om$ (\ref{A26}) is a closed (i.e., symplectic) form, the
Hermitian metric $g$ is called a {\it K\"ahler metric}
\index{K\"ahler metric!on a Hilbert manifold} and $\Om$ a {\it
K\"ahler form}. \index{K\"ahler form!on a Hilbert manifold}
Accordingly, $(\cP,g, \Om)$ is said to be an infinite-dimensional
{\it K\"ahler manifold}. \index{K\"ahler
manifold!infinite-dimensional}

By analogy with the case of a Banach manifold, we modify
Definition \ref{+181} and define a {\it connection} $\nabla$
\index{connection!on a Hilbert manifold} on a Hilbert manifold
$\cP$ as a $\Bbb C$-module morphism
\be
\nabla: C\cT_1(\cP)\to \cC^1(\cP) \ot C\cT_1(\cP),
\ee
which obeys the Leibniz rule
\be
\nabla(f\vt)= df\ot\vt +f\nabla(\vt), \qquad f\in \Bbb
C^\infty(\cP), \qquad \vt\in C\cT_1(\cP).
\ee
Similarly, a connection is introduced on any $\Bbb
C^\infty(\cP)$-module, e.g., on sections of tensor bundles over a
Hilbert manifold $\cP$. Let $D$ and $\ol D$ denote the holomorphic
and antiholomorphic parts of $\nabla$, and let
$\nabla_\vt=\vt\rfloor\nabla$, $D_\vt$ and $\ol D_\vt$ be the
corresponding covariant derivatives along a complex vector field
$\vt$ on $\cP$. For any complex vector field $\vt= \nu+\ol \up$ on
$\cP$, we have the relations
\be
D_\vt=\nabla_\nu, \qquad \ol D_\vt=\ol\nabla_{\ol \up},\qquad
D_{J\vt}=iD_\vt, \qquad \ol D_{J\vt}=-i\ol D_\vt.
\ee

\begin{prop} \label{spr703} \mar{spr703}
Given a K\"ahler manifold $(\cP,g)$, there always exists a {\it
metric connection} \index{metric connection!on a Hilbert manifold}
on $\cP$ such that
\be
\nabla g=0, \qquad \nabla \Om=0, \qquad \nabla J=0,
\ee
   where $J$ is
regarded as a section of the tensor bundle $CT^*\cP\ot CT\cP$.
\end{prop}

\begin{ex} \label{+303} \mar{+303}
If $\cP=E$ is a Hilbert space, then
\be
CT\cP=E\times (E\oplus \ol E).
\ee
A Hermitian form $\lng.|.\rng$ on $E$ defines the constant
Hermitian metric
\mar{A25}\ben
&& g: (E\oplus\ol E)\times (E\oplus\ol E)\to\Bbb C, \nonumber\\
&& g(\vt,\vt')= \lng v|u'\rng + \lng v'|u\rng, \qquad
\vt=v+\ol u,\quad \vt'=v'+\ol u',  \label{A25}
\een
on $\cP=E$. The associated fundamental form (\ref{A26}) reads
\mar{A28}\beq
\Om(\vt,\vt')=i\lng v|u'\rng - i\lng v'|u\rng. \label{A28}
\eeq
It is constant on $E$. Therefore, $d\Om=0$ and $g$ (\ref{A25}) is
a K\"ahler metric. The metric connection on $E$ is trivial, i.e.,
$\nabla=d$, $D=\dr$, $\ol D=\ol\dr$.
\end{ex}

\begin{ex} \label{x1} \mar{x1}
Given a Hilbert space $E$, {\it a projective Hilbert space}
\index{projective Hilbert space} $PE$ is made up by the complex
one-dimensional subspaces (i.e., complex rays) of $E$. This is a
Hilbert manifold with the following standard atlas. For any
non-zero element $x\in E$, let us denote by $\rx$ a point of $PE$
such that $x\in\rx$. Then each normalized element $h\in E$,
$||h||=1$, defines a chart $(U_h,\f_h)$ of the projective Hilbert
space $PE$ such that
\mar{A34}\beq
U_h=\{\rx\in PE\,\,:\,\, \lng x|h\rng\neq 0\}, \qquad \f_h(\rx)=
\frac{x}{\lng x|h\rng} - h. \label{A34}
\eeq
The image of $U_h$ in the Hilbert space $E$ is the
one-codimensional closed (Hilbert) subspace
\mar{A34'}\beq
E_h=\{z\in E\,\,:\,\, \lng z|h \rng = 0\}, \label{A34'}
\eeq
where $z(\rx)+h\in \rx$. In particular, given a point $\rx\in PE$,
one can choose the centered chart $E_h$, $h\in \rx$, such that
$\f_h(\rx)=0$. Hilbert spaces $E_h$ and $E_{h'}$ associated to
different charts $U_h$ and $U_{h'}$ are isomorphic. The transition
function between them is a holomorphic function
\mar{spr811'}\beq
z'(\rx)=\frac{z(\rx)+h}{\lng z(\rx)+h|h'\rng}-h', \qquad \rx\in
U_h\cap U_{h'}, \label{spr811'}
\eeq
from $\f_h(U_h\cap U_{h'})\subset E_h$ to $\f_{h'}(U_h\cap
U_{h'})\subset E_{h'}$. The set of the charts $\{(U_h,\f_h)\}$
with the transition functions (\ref{spr811'}) provides a
holomorphic atlas of the projective Hilbert space $PE$. The
corresponding coordinate transformations for the tangent vectors
to $PE$ at $\rx\in PE$ reads \mar{spr812}\beq v'=\frac{1}{\lng
x|h'\rng}[\lng x|h\rng v-x\lng v|h\rng]. \label{spr812}
\eeq

The projective Hilbert space $PE$ is homeomorphic to the quotient
of the unitary group $U(E)$ equipped with the normed operator
topology by the stabilizer of a ray of $E$. It is connected and
simply connected \cite{cir1}.

The projective Hilbert space $PE$ admits a unique Hermitian metric
$g$ such that the corresponding distance function on $PE$ is
\mar{A35}\beq
\rho(\rx,\rx')= \sqrt 2\,{\rm arccos}(\nm{\lng x|x'\rng}),
\label{A35}
\eeq
where $x,x'$ are normalized elements of $E$. It is a K\"ahler
metric, called the {\it Fubini--Studi metric}.
\index{Fubini--Studi metric} Given a coordinate chart
$(U_h,\f_h)$, this metric reads \mar{A36}\beq g_{\rm
FS}(\vt_1,\vt_2)= \frac{\lng v_1|u_2\rng +\lng
v_2|u_1\rng}{1+\|z\|^2} - \frac{\lng z|u_2\rng\lng v_1|z\rng +\lng
z|u_1\rng\lng v_2|z\rng}{(1+\|z\|^2)^2}, \quad z\in E_h,
\label{A36}
\eeq
for any complex tangent vectors $\vt_1=v_1+\ol u_1$ and
$\vt_2=v_2+\ol u_2$ in $CT_zPE$. The corresponding K\"ahler form
is given by the expression
\mar{A37}\beq
\Om_{\rm FS}(\vt_1,\vt_2)= i\frac{\lng v_1|u_2\rng -\lng
v_2|u_1\rng}{1+\|z\|^2} - i\frac{\lng z|u_2\rng\lng v_1|z\rng
-\lng z|u_1\rng\lng v_2|z\rng}{(1+\|z\|^2)^2}. \label{A37}
\eeq
It is readily justified that the expressions (\ref{A36}) --
(\ref{A37}) are preserved under the transition functions
(\ref{spr811'}) -- (\ref{spr812}). Written in the coordinate chart
centered at a point $z(\rx)=0$, these expressions come to the
expressions (\ref{A25}) and (\ref{A28}), respectively.
\end{ex}

\section{Hilbert and $C^*$-algebra bundles}

Due to the inverse mapping Theorem \ref{w371} for Banach
manifolds, fibred Banach manifolds similarly to the
finite-dimensional smooth ones can be defined as surjective
submersions of Banach manifold onto Banach manifold. Locally
trivial fibred Banach manifolds are called (smooth) {\it Banach
vector bundle}. \index{Banach vector bundle} They are exemplified
by vector (e.g., tangent and cotangent) bundles over Banach
manifolds in Section 2.1. Here, we restrict our consideration to
particular Banach vector bundles whose fibres are $C^*$-algebras
(seen as Banach spaces) and Hilbert spaces, but the base is a
finite-dimensional smooth manifold.

Note that sections of a Banach vector bundle $\cB\to Q$ over a
smooth finite-dimensional manifold $Q$ constitute a locally free
$C^\infty(Q)$-module $\cB(Q)$. Following the proof of Proposition
\ref{+219}, one can show that it is a projective
$C^\infty(Q)$-module. In a general setting, we therefore can
consider projective locally free $C^\infty(Q)$-modules, locally
generated by a Banach space. In contrast with the case of
projective $C^\infty(X)$ modules of finite rank, such a module
need not be a module of sections of some Banach vector bundle.

Let $\cC\to Q$ be a locally trivial topological fibre bundle over
a finite-dimensional smooth real manifold $Q$ whose typical fibre
is a $C^*$-algebra $A$, regarded as a real Banach space, and whose
transition functions are smooth. Namely,
 given two trivializations charts $(U_1,\psi_1)$ and
$(U_2,\psi_2)$ of $\cC$, we have the smooth morphism of Banach
manifolds
\be
\psi_1\circ \psi_2^{-1}: U_1\cap U_2\times A\to U_1\cap U_2\times
A,
\ee
where
\be
\psi_1\circ \psi_2^{-1}|_{q\in U_1\cap U_2}
\ee
is an automorphism of $A$. We agree to call $\cC\to Q$ a {\it
bundle of $C^*$-algebras}. \index{fibre bundle!of $C^*$-algebras}
It is a Banach vector bundle. The $C^\infty(Q)$-module $\cC(Q)$ of
smooth sections of this fibre bundle is a unital involutive
algebra with respect to fibrewise operations. Let us consider a
subalgebra $A(Q)\subset \cC(Q)$ which consists of sections of
$\cC\to Q$ vanishing at infinity on $Q$. It is provided with the
norm
\mar{spr740}\beq
||\al||=\op{\rm sup}_{q\in Q}||\al(q)||<\infty, \qquad \al\in
A(Q), \label{spr740}
\eeq
but fails to be complete. Nevertheless, one extends $A(Q)$ to a
$C^*$-algebra of continuous sections of $\cC\to Q$ vanishing at
infinity on a locally compact space $Q$.

\begin{rem} \label{spr742} \mar{spr742}
Let $\cC\to Q$ be a topological bundle of $C^*$-algebras over a
locally compact topological space $Q$, and let $\cC^0(Q)$ denote
the involutive algebra of its continuous sections. This algebra
exemplifies a locally trivial continuous {\it field of
$C^*$-algebras} \index{field of $C^*$-algebras} in  \cite{dixm}.
Its subalgebra $A^0(Q)$ of sections vanishing at infinity on $Q$
is a $C^*$-algebra with respect to the norm (\ref{spr740}). It is
called a {\it $C^*$-algebra defined by a continuous field of
$C^*$-algebras}. \index{$C^*$-algebra!defined by a continuous
field of $C^*$-algebras} There are several important examples of
$C^*$-algebras of this type. For instance, any commutative
$C^*$-algebra is isomorphic to the algebra of continuous complex
functions vanishing at infinity on its spectrum.
\end{rem}

\begin{rem} \label{w430} \mar{w430}
One can consider a locally trivial bundle of $C^*$-algebras
$\cC\to Q$ as a fibre bundle with the structure topological group
Aut$(A)$ of automorphisms of $A$. If a fibre bundle $\cC$ is
smooth, this group is necessarily provided with a normed operator
topology.
\end{rem}

Hilbert bundles over a smooth manifold are similarly defined. Let
$\cE\to Q$ be a locally trivial topological fibre bundle over a
finite-dimensional smooth real manifold $Q$ whose typical fibre is
a Hilbert space $E$, regarded as a real Banach space, and whose
transition functions are smooth functions taking their values in
the unitary group $U(E)$ equipped with the normed operator
topology. We agree to call $\cE\to Q$ a {\it Hilbert bundle}.
\index{Hilbert bundle} It is a Banach vector bundle. Smooth
sections of $\cE\to Q$ constitute a $C^\infty(Q)$-module $\cE(Q)$,
called a {\it Hilbert module}. \index{Hilbert module} Continuous
sections of $\cE\to Q$ constitute a locally trivial continuous
field of Hilbert spaces \cite{dixm}.

There are the following relations between bundles of
$C^*$-algebras and Hilbert bundles.

Let $T(E)\subset B(E)$ be the $C^*$-algebra of compact (completely
continuous) operators in a Hilbert space $E$. It is called an {\it
elementary $C^*$-algebra}. \index{$C^*$-algebra!elementary} Every
automorphism $\f$ of $E$ yields the corresponding automorphism
\be
T(E)\to \f T(E)\f^{-1}
\ee
of the $C^*$-algebra $T(E)$. Therefore, given a Hilbert bundle
$\cE\to Q$ with transition functions
\be
E\to \rho_{\iota\zeta}(q) E, \qquad q\in U_\iota\cap U_\zeta,
\ee
over a cover $\{U_\iota\}$ of $Q$, we have the associated locally
trivial bundle of elementary $C^*$-algebras $T(E)$ with transition
functions
\mar{w700}\beq
T(E)\to \rho_{\al\bt}(q)T(E)(\rho_{\al\bt}(q))^{-1}, \qquad q\in
U_\al\cap U_\bt, \label{w700}
\eeq
which are proved to be continuous with respect to the normed
operator topology on $T(E)$ \cite{dixm}. The proof is based on the
following facts.

$\bullet$ The set of degenerate operators (i.e., operators of
finite rank) is dense in $T(E)$.

$\bullet$ Any operator of finite rank is a linear combination of
operators
\be
P_{\xi,eta}:\zeta \to \lng\zeta|\eta\rng\xi, \qquad \xi,\eta,\zeta
\in E,
\ee
and even the projectors $P_\xi$ onto $\xi\in E$.

$\bullet$ Let $\xi_1,\ldots,\xi_{2n}$ be variable vectors of $E$.
If $\xi_i$, $i=1,\ldots,2n$, converges to $\eta_i$ (or, more
generally, $\lng \xi_i|\xi_j\rng$ converges to
$\lng\eta_i|\eta_j\rng$ for any $i$ and $j$), then
\be
P_{\xi_1,\xi_2}+\cdots +P_{\xi_{2n-1},\xi_{2n}}
\ee
uniformly converges to
\be
P_{\eta_1,\eta_2}+\cdots +P_{\eta_{2n-1},\eta_{2n}}.
\ee

Note that, given a Hilbert bundle $\cE\to Q$, the associated
bundle of $C^*$-algebras $B(E)$ of bounded operators in $E$ fails
to be constructed in general because the transition functions
(\ref{w700}) need not be continuous.

The opposite construction however meets a topological obstruction
as follows \cite{bryl93,carey}.

Let $\cC\to Q$ be a bundle of $C^*$-algebras whose typical fibre
is an elementary $C^*$-algebra $T(E)$ of compact operators in a
Hilbert space $E$. One can think of $\cC\to Q$ as being a
topological fibre bundle with the structure group of automorphisms
of $T(E)$. This is the projective unitary group $PU(E)$. With
respect to the normed operator topology, the groups $U(E)$ and
$PU(E)$ are the Banach Lie groups \cite{harp}. Moreover, $U(E)$ is
contractible if a Hilbert space $E$ is infinite-dimensional
\cite{kuip}. Let $(U_\al,\rho_{\al\bt})$ be an atlas of the fibre
bundle $\cC\to Q$ with $PU(E)$-valued transition functions
$\rho_{\al\bt}$. These transition functions give rise to the maps
\be
\ol \rho_{\al\bt}: U_\al\cap U_\bt\to U(E),
\ee
which however fail to be transition functions of a fibre bundle
with the structure group $U(E)$ because they need not satisfy the
cocycle condition. Their failure to be so is measured by the
$U(1)$-valued cocycle
\be
e_{al\bt\g}=\ol g_{\bt\g}\ol g^{-1}_{\al\g}\ol g_{\al\bt}.
\ee
This cocycle defines a class $[e]$ in the cohomology group
$H^2(Q;U(1)_Q)$ of the manifold $Q$ with coefficients in the sheaf
$U(1)_Q$ of continuous maps of $Q$ to $U(1)$. This cohomology
class vanishes iff there exists a Hilbert bundle associated to
$\cC$. Let us consider the short exact sequence of sheaves
\be
0\to \Bbb Z\ar C^0_Q\ar^\g U(1)_Q\to 0,
\ee
where $C^0_Q$ is the sheaf of continuous real functions on $Q$ and
the morphism $\g$ reads
\be
\g: C^0_Q\ni f\mapsto \exp(2\pi if)\in U(1)_Q.
\ee
This exact sequence yields the long exact sequence (\ref{spr227})
of the sheaf cohomology groups
\be
&& 0\to \Bbb Z\ar C^0_Q \ar U(1)_Q \ar
H^1(Q;\Bbb Z) \ar\cdots \\
&& \quad  H^p(Q;\Bbb Z)\ar H^p(Q;C^0_Q)\ar H^p(Q;U(1)_Q)\ar
H^{p+1}(Q;\Bbb Z) \ar\cdots\,.
\ee
Since the sheaf $C^0_Q$ is fine and acyclic, we obtain at once
from this exact sequence the isomorphism of cohomology groups
\be
H^2(Q,U(1)_Q)=H^3(Q,\Bbb Z).
\ee
The image of $[e]$ in $H^3(Q,\Bbb Z)$ is called the {\it
Dixmier--Douady class} \index{Dixmier--Douady class} \cite{dixm}.
One can show that the negative $-[e]$ of the Dixmier--Douady class
is the obstruction class of the lift of $PU(E)$-principal bundles
to the $U(E)$-principal ones \cite{carey}.

Thus, studying Hilbert and $C^*$-algebra bundles, we come to fibre
bundles with unitary and projective unitary structure groups.

\section{Connections on Hilbert and $C^*$-algebra bundles}

There are different notions of a connection on Hilbert and
$C^*$-algebra bundles which need not be obviously equivalent,
unless bundles are finite-dimensional. These are connections on
structure modules of sections, connections as a horizontal
splitting and principal connections.

Given a bundle of $C^*$-algebras $\cC\to Q$ with a typical fibre
$A$ over a smooth real manifold $Q$, the involutive algebra
$\cC(Q)$ of its smooth sections is a $C^\infty(Q)$-algebra.
Therefore, one can introduce a connection on the fibre bundle
$\cC\to Q$ as a connection on the $C^\infty(Q)$-algebra $\cC(Q)$.
In accordance with Definition \ref{mos088}, such a {\it
connection} \index{connection!on a bundle of $C^*$-algebras}
assigns to each vector field $\tau$ on $Q$ a symmetric derivation
$\nabla_\tau$ of the involutive algebra $\cC(Q)$ which obeys the
Leibniz rule
\be
\nabla_\tau(f\al)=(\tau\rfloor df)\al +f\nabla_\tau \al, \qquad
f\in \Bbb C^\infty(Q), \qquad \al\in \cC(Q),
\ee
and the condition
\be
\nabla_\tau\al^*=(\nabla_\tau\al)^*.
\ee
Let us recall that two such connections $\nabla_\tau$ and
$\nabla'_\tau$ differ from each other in a derivation of the
$C^\infty(Q)$-algebra $\cC(Q)$. Then, given a trivialization chart
\be
\cC|_U\cong U\times A
\ee
of $\cC\to Q$, a connection on $\cC(Q)$ can be written in the form
\mar{spr754}\beq
\nabla_\tau =\tau^m(q)(\dr_m - \dl_m(q)), \qquad q\in U,
\label{spr754}
\eeq
where $(q^m)$ are local coordinates on $Q$ and $\dl_m(q)$ for all
$q\in U$ are symmetric bounded derivations of the $C^*$-algebra
$A$.

\begin{rem} \label{spr771} \mar{spr771}
A problem is that a $C^*$-algebra need not admit bounded
derivations. For instance, no commutative algebra possesses
bounded derivations. Moreover, a bounded derivation of a
$C^*$-algebra is the infinitesimal generator of a uniformly
continuous one-parameter group of automorphisms of this algebra
\cite{brat}. However, strongly continuous groups of $C^*$-algebra
automorphisms usually are considered in quantum models
\cite{book05}. Therefore, one should assume that, in these models,
the derivations $\dl_m(q)$ in the expression (\ref{spr754}) are
unbounded in general. This leads us to the notion of a {\it
generalized connection} \index{generalized connection} on bundles
of $C^*$-algebras \cite{asor}.
\end{rem}

Let $\cE\to Q$ be a Hilbert bundle with a typical fibre $E$ and
$\cE(Q)$ the $C^\infty(Q)$-module of its smooth sections. Then a
{\it connection on a Hilbert bundle} \index{connection!on a
Hilbert bundle} $\cE\to Q$ is defined as a connection $\nabla$ on
the module $\cE(Q)$. In accordance with Definition \ref{t57}, such
a connection assigns to each vector field $\tau$ on $Q$ a first
order differential operator $\nabla_\tau$ on $\cE(Q)$ which obeys
both the Leibniz rule
\be
\nabla_\tau(f\psi)=(\tau\rfloor df)\psi +f\nabla_\tau \psi, \qquad
f\in \Bbb C^\infty(Q), \qquad \psi\in \cE(Q),
\ee
and the additional condition
\mar{spr757}\beq
\lng(\nabla_\tau \psi)(q)|\psi(q)\rng
+\lng\psi(q)|(\nabla_\tau\psi)(q)\rng =\tau(q)\rfloor d\lng\psi
(q)|\psi(q)\rng. \label{spr757}
\eeq
Given a trivialization chart $\cE|_U\cong U\times E$ of $\cE\to
Q$, a connection on $\cE(Q)$ reads
\mar{spr758}\beq
\nabla_\tau= \tau^m(q)(\dr_m + i\cH_m(q)), \qquad q\in U,
\label{spr758}
\eeq
where $\cH_m(q)$ for all $q\in U$ are bounded self-adjoint
operators in the Hilbert space $E$.

In a general setting, let $\cB\to Q$ be a Banach vector bundle
over a finite-dimensional smooth manifold $Q$ and $\cB(X)$ the
locally free $C^\infty(Q)$-module of its smooth sections $s(q)$.
By virtue of Definition \ref{t57}, a connection on $\cB(Q)$
assigns to each vector field $\tau$ on $Q$ a first order
differential operator $\nabla_\tau$ on $\cB(Q)$ which obeys the
Leibniz rule
\mar{w728}\beq
\nabla_\tau(fs)=(\tau\rfloor df)s +f\nabla_\tau s, \qquad f\in
\Bbb C^\infty(Q), \qquad s\in \cB(Q). \label{w728}
\eeq
In accordance with Proposition \ref{w715}, such a connection
exists. Connections (\ref{spr754}) and (\ref{spr758}) exemplify
connections on Banach vector bundles $\cC\to Q$ and $\cE\to Q$,
but they obey additional conditions because these bundles possess
additional structures of a $C^*$-algebra bundle and a Hilbert
bundle, respectively. In particular, the connection (\ref{spr758})
is a principal connection whose second term is an element of the
Lie algebra of the unitary group $U(E)$.

In a different way, connection on a Banach vector bundle $\cB\to
Q$ can be defined as a splitting of the exact sequence
\be
0\to V\cB \to T\cB \to TQ\op\ot_Q \cB \to 0,
\ee
where $V\cB$ denotes the vertical tangent bundle of $\cB\to Q$. In
the case of finite-dimensional vector bundles, both definitions
are equivalent (Section 1.6). This equivalence is extended to the
case of Banach vector bundles over a finite-dimensional base. We
leave the proof of this fact outside the scope of our exposition
because it involves the notion of jets of Banach fibre bundles.

\section{Instantwise quantization}

This Section addresses the evolution of such quantum systems which
can be viewed as a parallel displacement along time.

It should be emphasized that, in quantum mechanics based on the
Schr\"odinger and Heisenberg equations, the physical time plays
the role of a classical parameter. Indeed, all relations between
operators in quantum mechanics are simultaneous, while computation
of mean values of operators in a quantum state  does not imply
integration over time. It follows that, at each instant $t\in\Bbb
R$, there is an instantaneous quantum system characterized by some
$C^*$-algebra $A_t$. Thus, we come to {\it instantwise
quantization}. \index{instantwise quantization} Let us suppose
that all instantaneous $C^*$-algebras $A_t$ are isomorphic to some
unital $C^*$-algebra $A$. Furthermore, let they constitute a
locally trivial smooth bundle $\cC$ of $C^*$-algebras over the
time axis $\Bbb R$. Its typical fibre is $A$. This bundle of
$C^*$-algebras is trivial, but need not admit a canonical
trivialization in general. One can think of its different
trivializations as being associated to different reference frames
\cite{book05,sard07a}.

Let us describe evolution of quantum systems in the framework of
instantwise quantization. Given a bundle of $C^*$-algebras
$\cC\to\Bbb R$, this evolution can be regarded as a parallel
displacement with respect to some connection on $\cC\to\Bbb R$
\cite{asor,sard00}. Following previous Section, we define $\nabla$
as a connection on the involutive $C^\infty(\Bbb R)$-algebra
$\cC(\Bbb R)$ of smooth sections of $\cC\to\Bbb R$. It assigns to
the standard vector field $\dr_t$ on $\Bbb R$ a symmetric
derivation $\nabla_t$ of $\cC(\Bbb R)$ which obeys the Leibniz
rule
\be
\nabla_t (f\al)=\dr_tf\al+ f\nabla_t \al, \qquad \al\in \cC(\Bbb
R), \qquad f\in C^\infty(\Bbb R),
\ee
and the condition
\be
\nabla_t\al^*=(\nabla_t\al)^*.
\ee
Given a trivialization $\cC=\Bbb R\times A$, a connection
$\nabla_t$ reads
\mar{+331}\beq
\nabla_t =\dr_t - \dl(t),\label{+331}
\eeq
where $\dl(t)$, $t\in\Bbb R$, are symmetric derivations of the
$C^*$-algebra $A$, i.e.,
\be
\dl_t(ab)= \dl_t(a)b + a\dl_t(b), \qquad \dl_t(a^*)=\dl_t(a)^*,
\qquad
  a,b\in A.
\ee

We say that a section $\al$ of the bundle of $C^*$-algebras
$\cC\to\Bbb R$ is an integral section of the connection $\nabla_t$
if
\mar{+333}\beq
\nabla_t\al(t) =[\dr_t - \dl(t)]\al(t)=0. \label{+333}
\eeq
One can think of the equation (\ref{+333}) as being the {\it
Heisenberg equation} \index{Heisenberg equation} describing
quantum evolution.

In particular, let the derivations $\dl(t)=\dl$ in the Heisenberg
equation (\ref{+333}) be the same for all $t\in\Bbb R$, and let
$\dl$ be an infinitesimal generator of a strongly continuous
one-parameter group $[G_t]$ of automorphisms of the $C^*$-algebra
$A$ \cite{book05}. The pair $(A,[G_t])$ is called the {\it
$C^*$-dynamic system}. \index{$C^*$-dynamic system} It describes
evolution of a conservative quantum system. Namely,  for any $a\in
A$, the curve $\al(t)=G_t(a)$, $t\in\Bbb R$, in $A$ is a unique
solution with the initial value $\al(0)=a$ of the Heisenberg
equation (\ref{+333}).

It should be emphasized that, if the derivation $\dl$ is
unbounded, the connection $\nabla_t$ (\ref{+331}) is not defined
everywhere on the algebra $\cC(\Bbb R)$. In this case, we deal
with a generalized connection. It is given by operators of a
parallel displacement, whose  generators however are ill defined
\cite{asor}. Moreover, it may happen that a representation $\pi$
of the $C^*$-algebra $A$ does not carry out a representation of
the automorphism group $[G_t]$. Therefore, quantum evolution
described by the conservative Heisenberg equation, whose solution
is a strongly (but not uniformly) continuous dynamic system
$(A,G_t)$, need not  be described by the Schr\"odinger equation
(see Remark \ref{spr763} below).

If $\dl$ is a bounded derivation of a $C^*$-algebra $A$, the
Heisenberg and Schr\"odinger pictures of evolution of a
conservative quantum system are equivalent. Namely, as was
mentioned above, $\dl$ is an infinitesimal generator of a
uniformly continuous one-parameter group $[G_t]$ of automorphisms
of $A$, and {\it vice versa}. For any representation $\pi$ of $A$
in a Hilbert space $E$, there exists a bounded self-adjoint
operator $\cH$ in $E$ such that
\mar{spr761}\beq
\pi(\dl(a))= -i[\cH,\pi(a)],\qquad \pi(G_t)=\exp(-it\cH), \quad
a\in A, \qquad t\in \Bbb R. \label{spr761}
\eeq
The corresponding conservative Schr\"odinger equation reads
\mar{spr760}\beq
(\dr_t +i\cH)\psi=0, \label{spr760}
\eeq
where $\psi$ is a section of the trivial Hilbert bundle $\Bbb
R\times E\to\Bbb R$. Its solution with an initial value
$\psi(0)\in E$ is
\mar{spr762}\beq
\psi(t)=\exp[-it\cH]\psi(0). \label{spr762}
\eeq

\begin{rem} \label{spr763} \mar{spr763}
If the derivation $\dl$ is unbounded, but obeys some conditions,
we also obtain the unitary representation (\ref{spr761}) of the
group $[G_t]$, but the curve $\psi(t)$ (\ref{spr762}) need not be
differentiable, and the Schr\"odinger equation (\ref{spr760}) is
ill defined \cite{book05}.
\end{rem}

Let us return to the general case of a non-conservative quantum
system characterized by a bundle of $C^*$-algebras $\cC\to\Bbb R$
with the typical fibre $A$. Let us suppose that a phase Hilbert
space of a quantum system is preserved under evolution, i.e.,
instantaneous $C^*$-algebras $A_t$ are endowed with
representations equivalent to some representation of the
$C^*$-algebra $A$ in a Hilbert space $E$. Then quantum evolution
can be described by means of the Schr\"odinger equation as
follows.

Let us consider a smooth Hilbert bundle $\cE\to \Bbb R$ with the
typical fibre $E$ and a connection $\nabla$ on the $C^\infty(\Bbb
R)$-module $\cE(\Bbb R)$ of smooth sections of $\cE\to\Bbb R$ (see
previous Section). This connection assigns to the standard vector
field $\dr_t$ on $\Bbb R$ an $\Bbb R$-module endomorphism
$\nabla_t$ of $\cE(\Bbb R)$ which obeys the Leibniz rule
\be
\nabla_t (f\psi)= \dr_tf\psi+ f\nabla_t \psi, \qquad \psi\in
\cE(\Bbb R), \qquad f\in C^\infty(\Bbb R),
\ee
and the condition
\be
\lng(\nabla_t \psi)(t)|\psi(t)\rng
+\lng\psi(t)|(\nabla_t\psi)(t)\rng =\dr_t \lng\psi
(t)|\psi(t)\rng.
\ee
Given a trivialization $\cE=\Bbb R\times E$, the connection
$\nabla_t$ reads
\mar{+346}\beq
\nabla_t\psi =(\dr_t + i\cH(t)) \psi, \label{+346}
\eeq
where $\cH(t)$ are bounded self-adjoint operators in $E$ for all
$t\in\Bbb R$. It is a $U(E)$-principal connection.

We say that a section $\psi$ of the Hilbert bundle $\Pi\to\Bbb R$
is an integral section of the connection $\nabla_t$ (\ref{+346})
if it fulfils the equation
\mar{+349}\beq
\nabla_t\psi(t) =(\dr_t + i\cH(t)) \psi(t)=0. \label{+349}
\eeq
One can think of this equation as being the {\it Schr\"odinger
equation} \index{Schr\"odinger equation} for the Hamiltonian
$\cH$. Its solution with an initial value $\psi(0)\in E$ exists
and reads
\mar{1072}\beq
\psi(t)=U(t) \psi(0), \label{1072}
\eeq
where $U(t)$ is an {\it operator of a parallel displacement}
\index{operator of a parallel displacement} with respect to the
connection (\ref{+346}). This operator is a differentiable section
of the trivial bundle
\be
\Bbb R\times U(E)\to \Bbb R
\ee
which obeys the equation
\mar{spr764}\beq
\dr_t U(t)=-i\cH(t) U(t), \qquad U_0=\bb. \label{spr764}
\eeq
The operator $U(t)$ plays the role of an {\it evolution operator}.
\index{evolution operator} It is given by the {\it time-ordered
exponential} \index{time-ordered exponential}
\mar{+351}\beq
U(t)=T\exp\left[ -i\op\int_0^t \cH(t')dt'\right], \label{+351}
\eeq
which uniformly converges in the operator norm \cite{dal}.

It should be emphasized that the evolution operator $U(t)$ has
been defined with respect to a given trivialization of a Hilbert
bundle $\cE\to\Bbb R$.

\section{Berry connection}

Let us consider a quantum system depending on a finite number of
real classical parameters given by sections of a smooth parameter
bundle $\Si\to \Bbb R$. For the sake of simplicity, we fix a
trivialization $\Si=\Bbb R\times Z$, coordinated by $(t,\si^m)$.
Although it may happen that the parameter bundle $\Si\to \Bbb R$
has no preferable trivialization.

In previous Section, we have characterized the time as a classical
parameter in quantum mechanics. This characteristic is extended to
other classical parameters. In a general setting, one assigns a
$C^*$-algebra $A_\si$ to each point $\si\in \Si$ of the parameter
bundle $\Si$, and treat $A_\si$ as a quantum system under fixed
values $(t,\si^m)$ of the parameters. However, we will simplify
repeatedly our consideration in order to single out a desired
 Berry's phase phenomenon.

Let us assume that all algebras $A_\si$ are isomorphic to the
algebra $B(E)$ of bounded operators in some Hilbert space $E$, and
consider a smooth Hilbert bundle $\cE\to \Si$ with the typical
fibre $E$. Smooth sections of  $\cE\to \Si$ constitute a module
$\cE(\Si)$ over the ring $C^\infty(\Si)$ of real functions on
$\Si$. A connection $\wt\nabla$  on $\cE(\Si)$ assigns to each
vector field $\tau$ on $\Si$ a first order differential operator
\mar{+320}\beq
\wt \nabla_\tau\in \dif_1(\cE(\Si),\cE(\Si)) \label{+320}
\eeq
which obeys the Leibniz rule
\be
\wt\nabla_\tau (fs)= (\tau\rfloor df)s+ f\wt\nabla_\tau s, \qquad
s\in \cE(\Si), \qquad f\in C^\infty(\Si).
\ee
Let $\tau$ be a vector field on $\Si$ such that $dt\rfloor\tau=1$.
Given a trivialization chart of the Hilbert bundle $\cE\to \Si$,
the operator $\wt\nabla_\tau$ (\ref{+320}) reads
\mar{+321}\beq
\wt\nabla_\tau(s) =(\dr_t + i\cH(t,\si^i)) s  +\tau^m(\dr_m - i\wh
A_m(t,\si^i))s, \label{+321}
\eeq
where $\cH(t,\si^i)$, $\wh A_m(t,\si^i)$ for each $\si\in\Si$ are
bounded self-adjoint operators in the Hilbert space $E$.

Let us consider the composite Hilbert bundle $\cE\to \Si\to\Bbb
R$. Every section $h(t)$ of the fibre bundle $\Si\to\Bbb R$
defines the subbundle $\cE_h=h^*\cE\to\Bbb R$ of the Hilbert
bundle $\cE\to \Bbb R$ whose typical fibre is $E$. Accordingly,
the connection $\wt\nabla$ (\ref{+321}) on the
$C^\infty(\Si)$-module $\cE(\Si)$ yields the pull-back connection
\mar{+322}\beq
\nabla_h(\psi) = [\dr_t - i(\wh A_m(t,h^i(t))\dr_t
h^m-\cH(t,h^i(t))]\psi \label{+322}
\eeq
on the $C^\infty(\Bbb R)$-module $\cE_h(\Bbb R)$ of sections
$\psi$ of the fibre bundle $\cE_h\to\Bbb R$.

We say that a section $\psi$ of the fibre bundle $\cE_h\to\Bbb R$
is an integral section of the connection (\ref{+322}) if
\mar{+323}\beq
\nabla_h(\psi) =[\dr_t -i(\wh A_m(t,h^i(t)) \dr_t
h^m-\cH(t,h^i(t))]\psi=0. \label{+323}
\eeq
One can think of the equation (\ref{+323}) as being the
Shr\"odinger equation for a quantum system depending on the
parameter function $h(t)$. Its solutions take the form
(\ref{1072}) where $U(t)$ is the time-ordered exponent
\mar{+356}\beq
U(t)=T\exp\left[ i\op\int_0^t (\wh A_m\dr_{t'}h^m -
\cH)dt'\right]. \label{+356}
\eeq

The term $i\wh A_m(t,h^i(t)) \dr_t h^m$ in the Shr\"odinger
equation (\ref{+323}) is responsible for the
 Berry's phase phenomenon, while $\cH$ is treated as an ordinary Hamiltonian of
a quantum system \cite{book05}.  To show the
 Berry's phase phenomenon clearly,
we will continue to simplify the system under consideration. Given
a trivialization of the fibre bundle $\cE\to\Bbb R$ and the above
mentioned trivialization $\Si=\Bbb R\times Z$ of the parameter
bundle $\Si$, let us suppose that the components $\wh A_m$ of the
connection $\wt\nabla$ (\ref{+321}) are independent of $t$ and
that the operators $\cH(\si)$ commute with the operators $\wh
A_m(\si')$ at all points of the curve $h(t)\subset \Si$. Then the
operator $U(t)$ (\ref{+356}) takes the form
\mar{+356'}\beq
U(t)=T\exp\left[ i\op\int_{h([0,t])} \wh A_m(\si^i) d\si^m
\right]\cdot T\exp\left[ -i\op\int_0^t\cH(t')dt'\right].
\label{+356'}
\eeq
One can think of the first factor in the right-hand side of the
expression (\ref{+356'}) as being the operator of a parallel
transport along the curve $h([0,t])\subset Z$ with respect to the
pull-back connection
\mar{+357}
\beq
\nabla= i^*\wt\nabla = \dr_m - i\wh A_m(t,\si^i) \label{+357}
\eeq
 on
the fibre bundle $\cE\to Z$, defined by the imbedding
\be
i: Z\hookrightarrow \{0\}\times Z\subset \Si.
\ee
Note that, since $\wh A_m$ are independent of time, one can
utilize any imbedding of $Z$ to $\{t\}\times Z$. Moreover, the
connection $\nabla$ (\ref{+357}), called the {\it Berry
connection}, \index{Berry connection} can be seen as a connection
on some principal fibre bundle $P\to Z$ for the unitary group
$U(E)$. Let the curve $h([0,t])$ be closed,  while the holonomy
group of the connection $\nabla$ at the point $h(t)=h(0)$ is not
trivial. Then the unitary operator
\mar{+359}\beq
T\exp\left[ i\op\int_{h([0,t])} \wh A_m(\si^i) d\si^m \right]
\label{+359}
\eeq
is not the identity. For example, if
\mar{+360}\beq
i\wh A_m(\si^i)=iA_m(\si^i)\Id_E \label{+360}
\eeq
is a $U(1)$-principal connection on $Z$, then the operator
(\ref{+359}) is the well-known Berry phase factor
 \be
\exp\left[ i\op\int_{h([0,t])} A_m(\si^i) d\si^m \right].
\ee
If (\ref{+360}) is a curvature-free connection,  Berry's phase is
exactly the Aharonov--Bohm effect  on the parameter space $Z$.

The following variant of the Berry's phase phenomenon leads us to
a principal bundle for familiar finite-dimensional Lie groups. Let
a Hilbert space $E$ be the Hilbert sum of $n$-dimensional
eigenspaces of the Hamiltonian $\cH(\si)$, i.e.,
\be
E= \op\bigoplus_{k=1}^\infty E_k, \qquad E_k=P_k(E),
\ee
where $P_k$ are the projection operators, i.e.,
\be
H(\si)\circ P_k = \la_k(\si) P_k
\ee
 (in the spirit of the adiabatic
hypothesis).  Let the operators $\wh A_m(z)$ be time-independent
and preserve the eigenspaces $E_k$ of the Hamiltonian $\cH$, i.e.,
\mar{+361}\beq
\wh A_m(z)= \op\sum_k \wh A_m^k(z) P_k, \label{+361}
\eeq
where $\wh A_m^k(z)$, $z\in Z$, are self-adjoint operators in
$E_k$. It follows that $\wh A_m(\si)$ commute with $\cH(\si)$ at
all points of the parameter bundle $\Si\to\Bbb R$. Then,
restricted to each subspace $E_k$, the parallel transport operator
(\ref{+359}) is a unitary operator in $E_k$. In this case, the
Berry connection (\ref{+357})  on the $U(E)$-principal bundle
$P\to Z$ can be seen as a composite connection on the composite
bundle
\be
P\to P/U(n)\to Z,
\ee
which is defined by some principal connection on the
$U(n)$-principal bundle $P\to P/U(n)$ and the trivial connection
on the fibre bundle $P/U(n)\to Z$ \cite{book05}.

Note that, since $U(E)$ is contractible, the $U(n)$-principal
bundle $U(E)\to U(E)/U(n)$ is universal and, consequently, the
typical fibre $U(E)/U(n)$ of $P/U(n)\to Z$ is exactly the
classifying space $B(U(n)$ of $U(n)$-principal bundles
\cite{book00}.
 Moreover, one can consider the parallel transport
along a curve in the bundle $P/U(n)$. In this case, a state vector
$\psi(t)$ acquires a geometric phase factor in addition to the
dynamical phase factor. In particular, if $\Si=\Bbb R$ (i.e.,
classical parameters are absent and  Berry's phase has only the
geometric origin) we come to the case of a Berry connection on the
$U(n)$-principal bundle over the classifying space $B(U(n))$
\cite{bohm}. If $n=1$, this is the variant of Berry's geometric
phase of \cite{anan}.

\chapter{Supergeometry}

Supergeometry is phrased in terms of $\Bbb Z_2$-graded modules and
sheaves over $\Bbb Z_2$-graded commutative algebras. Their
algebraic properties naturally generalize those of modules and
sheaves over commutative algebras, but this is not a particular
case of non-commutative geometry because of the peculiar
definition of graded derivations.

\section{Graded tensor calculus}

Unless otherwise stated, by a graded structure throughout this
Chapter is meant a $\Bbb Z_2$-graded structure, and the symbol
$\nw .$ stands for the $\Bbb Z_2$-graded parity. Let us recall
some basic notions of the graded tensor calculus
\cite{bart,cia,sard09a}.

An algebra $\cA$ is called {\it graded} \index{algebra!$\Bbb
Z_2$-graded} if it is endowed with a {\it grading automorphism}
\index{grading automorphism} $\g$ such that $\g^2=\id$. A graded
algebra seen as a $\Bbb Z$-module falls into the direct sum
$\cA=\cA_0\oplus \cA_1$ of two $\Bbb Z$-modules $\cA_0$ and
$\cA_1$ of {\it even} \index{even element} and {\it odd}
\index{odd element} elements such that
\be
\g(a)=(-1)^ia, \qquad a\in\cA_i, \qquad i=0,1.
\ee
One calls $\cA_0$ and $\cA_1$ the even and odd parts of $\cA$,
respectively. In particular, if $\g=\id$, then $\cA=\cA_0$. Since
\be
\g(aa')=\g(a)\g(a'),
\ee
we have
\be
[aa']=([a]+[a']){\rm mod}\,2
\ee
where $a\in \cA_{[a]}$, $a'\in \cA_{[a']}$. It follows that
$\cA_0$ is a subalgebra of $\cA$ and $\cA_1$ is an $\cA_0$-module.
If $\cA$ is a graded ring, then $[\bb]=0$.

A graded algebra $\cA$ is said to be {\it graded commutative}
\index{algebra!$\Bbb Z_2$-graded commutative} if
\be
aa'=(-1)^{[a][a']}a'a,
\ee
where $a$ and $a'$ are arbitrary {\it homogeneous elements}
\index{homogeneous element} of $\cA$, i.e., they are either even
or odd.

Given a graded algebra $\cA$, a left {\it graded $\cA$-module}
\index{graded module} $Q$ is a left $\cA$-module provided with the
grading automorphism $\g$ such that
\be
\g(aq)=\g(a)\g(q), \qquad a\in\cA,\qquad q\in Q,
\ee
i.e.,
\be
[aq]=([a]+[q]){\rm mod}\,2.
\ee
A graded module $Q$ is split into the direct sum $Q=Q_0\oplus Q_1$
of two $\cA_0$-modules $Q_0$ and $Q_1$ of even and odd elements.
Similarly, right graded modules are defined.

If $\cK$ is a graded commutative ring, a graded $\cK$-module can
be provided with a graded {\it $\cK$-bimodule}
\index{bimodule!graded} structure by letting
\be
qa= (-1)^{[a][q]}aq, \qquad a\in\cK, \qquad q\in Q.
\ee
A graded $\cK$-module is called {\it free} \index{graded
module!free} if it has a basis generated by homogeneous elements.
This basis is said to be of type $(n,m)$ if it contains $n$ even
and $m$ odd elements.

In particular, by a (real) {\it graded vector space} $B=B_0\oplus
B_1$ \index{graded vector space} is meant a graded $\Bbb
R$-module. A graded vector space is said to be $(n,m)$-dimensional
if $B_0=\Bbb R^n$  and $B_1=\Bbb R^m$.

The following are standard constructions of new graded modules
from old ones.

$\bullet$ The direct sum of graded modules over the same graded
commutative ring and a graded factor module are defined just as
those of modules over a commutative ring.

$\bullet$ The {\it tensor product} \index{tensor product!of graded
modules} $P\ot Q$ of graded $\cK$-modules $P$ and $Q$ is an
additive group generated by elements $p\ot q$, $p\in P$, $q\in Q$,
obeying the relations
\be
&& (p+p')\ot q =p\ot q + p'\ot q, \\
&& p\ot(q+q')=p\ot q+p\ot q', \\
&&  ap\ot q=(-1)^{[p][a]}pa\ot q= (-1)^{[p][a]}p\ot aq=\\
&& \qquad (-1)^{([p]+[q])[a]}p\ot qa, \qquad a\in\cK.
\ee
In particular, the tensor algebra $\ot Q$ of a graded $\cK$-module
$Q$ is defined as that of a module over a commutative ring in
Example \ref{ws40}. Its quotient $\w Q$ with respect to the ideal
generated by elements
\be
q\ot q' + (-1)^{[q][q']}q'\ot q, \qquad q,q'\in Q,
\ee
is the {\it bigraded exterior algebra} \index{exterior
algebra!bigraded} of a graded module $Q$ with respect to the {\it
graded exterior product} \index{graded exterior product}
\be
q\w q' =- (-1)^{[q][q']}q'\w q.
\ee

$\bullet$ A morphism $\Phi:P\to Q$ of graded $\cK$-modules is said
to be {\it even} \index{even morphism} (resp. {\it odd})
\index{odd morphism} if $\Phi$ preserves (resp. change) the graded
parity of all elements $P$. It obeys the relations
\be
\Phi(ap)=(-1)^{[\Phi][a]}\Phi(p), \qquad p\in P, \qquad a\in\cK.
\ee
The set $\hm_\cK(P,Q)$ of graded morphisms of a graded
$\cK$-module $P$ to a graded $\cK$-module $Q$ is naturally a
graded $\cK$-module. The graded $\cK$-module $P^*=\hm_\cK(P,\cK)$
is called the {\it dual} \index{graded module!dual} of a graded
$\cK$-module $P$.

A {\it graded commutative $\cK$-ring} \index{graded commutative
ring} $A$ is a graded commutative ring which is also a graded
$\cK$-module. A graded commutative $\Bbb R$-ring is said to be of
rank $N$ if it is a free algebra generated by the unit $\bb$ and
$N$ odd elements. A {\it graded commutative Banach ring}
\index{graded commutative ring!Banach} $A$ is a graded commutative
$\Bbb R$-ring which is a real Banach algebra whose norm obeys the
additional condition
\be
\|a_0 + a_1\|=\|a_0\| + \|a_1\|, \qquad a_0\in A_0, \quad a_1\in
A_1.
\ee

Let $V$ be a real vector space. Let $\La=\w V$ be its ($\Bbb
N$-graded) exterior algebra provided  with the $\Bbb Z_2$-graded
structure
\mar{+66}\beq
\La=\La_0\oplus \La_1, \qquad \La_0=\Bbb R\op\bigoplus_{k=1}
\op\w^{2k} V, \qquad \La_1=\op\bigoplus_{k=1} \op\w^{2k-1} V.
\label{+66}
\eeq
It is a graded commutative $\Bbb R$-ring, called the {\it
Grassmann algebra}. \index{Grassmann algebra} A Grassmann algebra,
seen as an additive group, admits the decomposition
\mar{+11}\beq
\La=\Bbb R\oplus R =\Bbb R\oplus R_0\oplus R_1=\Bbb R \oplus
(\La_1)^2 \oplus \La_1, \label{+11}
\eeq
where $R$ is the {\it ideal of nilpotents} \index{ideal!of
nilpotents} of $\La$. The corresponding projections
$\si:\La\to\Bbb R$ and $s:\La\to R$ are called the {\it body}
\index{body map} and {\it soul} \index{soul map} maps,
respectively.

Hereafter, we restrict our consideration to Grassmann algebras of
finite rank. Given a basis $\{c^i\}$ for the vector space $V$, the
elements of the Grassmann algebra $\La$ (\ref{+66}) take the form
\mar{z784}\beq
a=\op\sum_{k=0} \op\sum_{(i_1\cdots i_k)}a_{i_1\cdots
i_k}c^{i_1}\cdots c^{i_k}, \label{z784}
\eeq
where the second sum runs through all the tuples $(i_1\cdots i_k)$
such that no two of them are permutations of each other. The
Grassmann algebra $\La$ becomes a graded commutative Banach ring
if its elements (\ref{z784}) are endowed with the norm
\be
\|a\|=\op\sum_{k=0} \op\sum_{(i_1\cdots i_k)}\nm{a_{i_1\cdots
i_k}}.
\ee

Let $B$ be a graded vector space. Given a Grassmann algebra $\La$
of rank $N$, it can be brought into a graded $\La$-module
\be
\La B=(\La B)_0\oplus (\La B)_1=(\La_0\ot B_0\oplus \La_1\ot
B_1)\oplus (\La_1\ot B_0\oplus \La_0\ot B_1),
\ee
called a {\it superspace}. \index{graded envelope} The superspace
\mar{+70}\beq
B^{n\mid m}=[(\op\oplus^n\La_0) \oplus (\op\oplus^m\La_1)]\oplus
[(\op\oplus^n\La_1)\oplus (\op\oplus^m\La_0)] \label{+70}
\eeq
is said to be $(n,m)$-dimensional. The graded $\La_0$-module
\be
B^{n,m}= (\op\oplus^n\La_0) \oplus (\op\oplus^m\La_1)
\ee
is called an $(n,m)$-dimensional {\it supervector space}.
\index{supervector space}

Whenever referring to a topology on a supervector space
$B^{n,m}$, we will mean the Euclidean topology on a
$2^{N-1}[n+m]$-dimensional real vector space.

Given a superspace $B^{n\mid m}$  over a Grassmann algebra $\La$,
a $\La$-module endomorphism of $B^{n\mid m}$ is represented by an
$(n+ m)\times (n+m)$ matrix
\mar{+200}\beq
L=\left(
\begin{array}{cc}
L_1 & L_2 \\
L_3 & L_4
\end{array}
\right) \label{+200}
\eeq
with entries in $\La$. It is called a {\it supermatrix}.
\index{supermatrix}  One says that  a supermatrix $L$ is

$\bullet$ {\it even} \index{supermatrix!even} if $L_1$ and $L_4$
have even entries, while $L_2$ and $L_3$ have the odd ones;

$\bullet$ {\it odd} \index{supermatrix!odd} if $L_1$ and $L_4$
have odd entries, while $L_2$ and $L_3$ have the even ones.

Endowed with this gradation, the set of supermatrices (\ref{+200})
is a graded $\La$-ring.

 Invertible supermatrices constitute a
group $GL(n|m;\La)$, called the {\it general linear graded group}.
\index{general linear graded group}

Let $\cK$ be a graded commutative ring. A graded commutative
(non-associative) $\cK$-algebra $\cG$ is called a {\it Lie
$\cK$-superalgebra} \index{Lie superalgebra} if its product,
called the {\it superbracket} \index{superbracket} and denoted by
$[.,.]$, obeys the relations
\be
&& [\ve,\ve']=-(-1)^{[\ve][\ve']}[\ve',\ve],\\
&& (-1)^{[\ve][\ve'']}[\ve,[\ve',\ve'']]
+(-1)^{[\ve'][\ve]}[\ve',[\ve'',\ve]] +
(-1)^{[\ve''][\ve']}[\ve'',[\ve,\ve']] =0.
\ee
Obviously, the even part $\cG_0$ of a Lie $\cK$-superalgebra $\cG$
is a Lie $\cK_0$-algebra. A graded $\cK$-module $P$ is called a
$\cG$-module if it is provided with a $\cK$-bilinear map
\be
&& \cG\times P\ni (\ve,p)\mapsto \ve p\in P, \\
&& [\ve
p]=([\ve]+[p]){\rm mod}\,2,\\
&& [\ve,\ve']p=(\ve\circ\ve'-(-1)^{[\ve][\ve']}\ve'\circ\ve)p.
\ee

\section{Graded differential calculus and connections}

Linear differential operators and connections on graded modules
over graded commutative rings are defined similarly to those in
commutative geometry \cite{book05,sard09a}.

Let $\cK$ be a graded commutative ring and $\cA$ a graded
commutative $\cK$-ring. Let $P$ and $Q$ be graded $\cA$-modules.
The graded $\cK$-module $\hm_\cK (P,Q)$ of graded $\cK$-module
homomorphisms $\Phi:P\to Q$ can be endowed with the two graded
$\cA$-module structures
\mar{ws11}\beq
(a\Phi)(p)= a\Phi(p),  \qquad  (\Phi\bll a)(p) = \Phi (a p),\qquad
a\in \cA, \quad p\in P, \label{ws11}
\eeq
called $\cA$- and $\cA^\bll$-module structures, respectively. Let
us put
\mar{ws12}\beq
\dl_a\Phi= a\Phi -(-1)^{[a][\Phi]}\Phi\bll a, \qquad a\in\cA.
\label{ws12}
\eeq
An element $\Delta\in\hm_\cK(P,Q)$ is said to be a $Q$-valued {\it
graded differential operator} \index{graded differential operator}
of order $s$ on $P$ if
\be
\dl_{a_0}\circ\cdots\circ\dl_{a_s}\Delta=0
\ee
for any tuple of $s+1$ elements $a_0,\ldots,a_s$ of $\cA$. The set
$\dif_s(P,Q)$ of these operators inherits the graded module
structures (\ref{ws11}).

In particular, zero order graded differential operators obey the
condition
\be
\dl_a \Delta(p)=a\Delta(p)-(-1)^{[a][\Delta]}\Delta(ap)=0, \qquad
a\in\cA, \qquad p\in P,
\ee
i.e., they coincide with graded $\cA$-module morphisms $P\to Q$. A
first order graded differential operator $\Delta$ satisfies the
condition
\be
&& \dl_a\circ\dl_b\,\Delta(p)=
ab\Delta(p)- (-1)^{([b]+[\Delta])[a]}b\Delta(ap)-
(-1)^{[b][\Delta]}a\Delta(bp)+\\
&& \qquad (-1)^{[b][\Delta]+([\Delta]+[b])[a]}
=0, \qquad a,b\in\cA, \quad p\in P.
\ee

For instance, let $P=\cA$. Any zero order $Q$-valued graded
differential operator $\Delta$ on $\cA$ is defined by its value
$\Delta(\bb)$. Then there is a graded $\cA$-module isomorphism
\be
\dif_0(\cA,Q)=Q
\ee
via the association
\be
Q\ni q\mapsto \Delta_q\in \dif_0(\cA,Q),
\ee
where $\Delta_q$ is given by the equality $\Delta_q(\bb)=q$. A
first order $Q$-valued graded differential operator $\Delta$ on
$\cA$ fulfils the condition
\be
\Delta(ab)= \Delta(a)b+ (-1)^{[a][\Delta]}a\Delta(b)
-(-1)^{([b]+[a])[\Delta]} ab \Delta(\bb), \qquad  a,b\in\cA.
\ee
It is called a $Q$-valued {\it graded derivation} \index{graded
derivation} of $\cA$ if $\Delta(\bb)=0$, i.e., the graded Leibniz
rule
\mar{ws10}\beq
\Delta(ab) = \Delta(a)b + (-1)^{[a][\Delta]}a\Delta(b), \qquad
a,b\in \cA, \label{ws10}
\eeq
holds (cf. (\ref{ws100})). One obtains at once that any first
order graded differential operator on $\cA$ falls into the sum
\be
\Delta(a)= \Delta(\bb)a +[\Delta(a)-\Delta(\bb)a]
\ee
of a zero order graded differential operator $\Delta(\bb)a$ and a
graded derivation $\Delta(a)-\Delta(\bb)a$. If $\dr$ is a graded
derivation of $\cA$, then $a\dr$ is so for any $a\in \cA$. Hence,
graded derivations of $\cA$ constitute a graded $\cA$-module
$\gd(\cA,Q)$, called the {\it graded derivation module}.
\index{graded derivation module}

If $Q=\cA$, the graded derivation module $\gd\cA$ is also a Lie
superalgebra over the graded commutative ring $\cK$ with respect
to the superbracket
\mar{ws14}\beq
[u,u']=u\circ u' - (-1)^{[u][u']}u'\circ u, \qquad u,u'\in \cA.
\label{ws14}
\eeq
We have the graded $\cA$-module decomposition
\mar{ws15}\beq
\dif_1(\cA) = \cA \oplus\gd\cA. \label{ws15}
\eeq

Let us turn now to jets of graded modules. Given a graded
$\cA$-module $P$, let us consider the tensor product
$\cA\otimes_\cK P$ of graded $\cK$-modules $\cA$ and $P$. We put
\mar{ws16}\beq
\dl^b(a\otimes p)= (ba)\otimes p - (-1)^{[a][b]}a\otimes (bp),
\qquad p\in P, \quad a,b\in\cA.  \label{ws16}
\eeq
The {\it $k$-order graded jet module} $\cJ^k(P)$ \index{graded jet
module} of the module $P$ is defined as the quotient of the graded
$\cK$-module $\cA\otimes_\cK P$ by its submodule generated by
elements of type
\be
\dl^{b_0}\circ \cdots \circ\dl^{b_k}(a\otimes p).
\ee

In particular, the first order graded jet module $\cJ^1(P)$
consists of elements $a\ot_1 p$ modulo the relations
\mar{ws20}\beq
ab\otimes_1 p -(-1)^{[a][b]}b\otimes_1 (ap)
  -a\otimes_1(bp)  + \bb\ot_1(abp) =0. \label{ws20}
\eeq

For any $h\in\hm_\cA (\cA\ot P,Q)$, the equality
\be
\dl_b(h(a\ot p))=(-1)^{[h][b]}h(\dl^b(a\ot p))
\ee
holds. By analogous with Theorem \ref{t6}, one then can show that
any $Q$-valued graded differential operator $\Delta$ of order $k$
on a graded $\cA$-module $P$ factorizes uniquely
\be
\Delta: P\ar^{J^k} \cJ^k(P)\ar Q
\ee
through the morphism
\be
J^k:p\ni p\mapsto \bb\ot_k p\in \cJ^k(P)
\ee
and some homomorphism ${\got f}^\Delta:\cJ^k(P)\to Q$.
Accordingly, the assignment $\Delta\mapsto {\got f}^\Delta$
defines an isomorphism
\mar{ws21}\beq
\dif_s(P,Q)=\hm_{\cA}(\cJ^s(P),Q). \label{ws21}
\eeq

Let us focus on the first order graded jet module $\cJ^1$ of $\cA$
consisting of the elements $a\otimes_1 b$, $a,b\in\cA$, subject to
the relations
\mar{ws22}\beq
ab\otimes_1 \bb -(-1)^{[a][b]}b\otimes_1 a
  -a\otimes_1b  + \bb\ot_1(ab) =0. \label{ws22}
\eeq
It is endowed with the $\cA$- and $\cA^\bll$-module structures
\be
c(a\ot_1 b)=(ca)\ot_1 b,\qquad c\bll(a\ot_1 b)= a\ot_1(cb).
\ee
There are canonical $\cA$- and $\cA^\bll$-module monomorphisms
\be
&& i_1: \cA \ni a  \mapsto a\otimes_1 \bb\in \cJ^1,\\
&& J^1: \cA\ni a\mapsto \bb\otimes_1 a\in \cJ^1,
\ee
such that $\cJ^1$, seen as a graded $\cA$-module, is generated by
the elements $J^1a$, $a\in \cA$. With these monomorphisms, we have
the canonical $\cA$-module splitting
\mar{ws23}\ben
&& \cJ^1=i_1(\cA)\oplus \cO^1, \label{ws23} \\
&&
aJ^1(b)= a\ot_1 b=ab\ot_1\bb + a(\bb\ot_1 b- b\ot_1\bb), \nonumber
\een
where the graded $\cA$-module $\cO^1$ is generated by the elements
\be
\bb\ot_1 b-b\ot_1 \bb, \qquad b\in\cA.
\ee
Let us consider the corresponding $\cA$-module epimorphism
\mar{ws30}\beq
h^1:\cJ^1\ni \bb\ot_1 b\mapsto \bb\ot_1 b-b\ot_1 \bb\in \cO^1
\label{ws30}
\eeq
and the composition
\mar{ws31}\beq
d=h^1\circ J_1: \cA \ni b \mapsto \bb\ot_1 b- b\ot_1\bb \in \cO^1.
\label{ws31}
\eeq
The equality
\be
d(ab)= a\otimes_1 b  +b\otimes_1 a- ab\otimes_1 \bb-ba\otimes_1
\bb  =(-1)^{[a][b]}bda + adb
\ee
shows that $d$ (\ref{ws31}) is an even $\cO^1$-valued derivation
of $\cA$. Seen as a graded $\cA$-module, $\cO^1$ is generated by
the elements $da$ for all $a\in\cA$.

In view of the splittings (\ref{ws15}) and (\ref{ws23}), the
isomorphism (\ref{ws21}) reduces to the isomorphism
\mar{ws47}\beq
\gd\cA=\cO^{1*}=\hm_{\cA}(\cO^1,\cA) \label{ws47}
\eeq
of $\gd\cA$ to the dual $\cO^{1*}$ of the graded $\cA$-module
$\cO^1$. It is given by the duality relations
\mar{ws41}\beq
\gd\cA\ni u\leftrightarrow
   \f_u\in \cO^{1*}, \qquad \f_u(da)=u(a), \qquad  a\in \cA.
   \label{ws41}
\eeq
Using this fact, let us construct a differential calculus over a
graded commutative $\cK$-ring $\cA$.

Let us consider the bigraded exterior algebra $\cO^*$ of a graded
module $\cO^1$. It consists of finite linear combinations of
monomials of the form
\mar{ws42}\beq
\f=a_0 da_1\w\cdots\w da_k, \qquad a_i\in \cA, \label{ws42}
\eeq
whose product obeys the juxtaposition rule
\be
(a_0d a_1)\w (b_0d b_1)=a_0d (a_1b_0)\w db_1- a_0a_1d b_0\w d b_1
\ee
and the bigraded commutative relations
\mar{ws45}\beq
\f\w \f'=(-1)^{|\f||\f'|+[\f][\f']}\f'\w\f. \label{ws45}
\eeq
In order to make $\cO^*$ to a differential algebra, let us define
the coboundary operator $d:\cO^1\to\cO^2$ by the rule
\be
d\f(u,u')=-u'(u(\f)) +(-1)^{[u][u']}u(u'(\f)) +[u',u](\f),
\ee
where $u,u'\in \gd\cA$, $\f\in\cO^1$, and $u$, $u'$ are both
graded derivatives of $\cA$ and $\cA$-valued forms on $\cO^1$. It
is readily observed that, by virtue of the relation (\ref{ws41}),
$(d\circ d)(a)=0$ for all $a\in \cA$. Then $d$ is extended to the
bigraded exterior algebra $\cO^*$ if its action on monomials
(\ref{ws42}) is defined as
\be
d(a_0 da_1\w\cdots\w da_k)=da_0\w da_1\w\cdots\w da_k.
\ee
This operator is nilpotent and fulfils the familiar relations
\mar{ws44}\beq
d(\f\w\f')= d\f\w\f' +(-1)^{|\f||\f'|}\f\w d\f'. \label{ws44}
\eeq
It makes $\cO^*$ into a {\it differential bigraded algebra},
\index{differential bigraded algebra} called a {\it graded
differential calculus} \index{graded differential calculus} over a
graded commutative $\cK$-ring $\cA$.

Furthermore, one can extend the duality relation (\ref{ws41}) to
the {\it graded interior product} \index{interior product!graded}
of $u\in\gd\cA$ with any monomial $\f$ (\ref{ws42}) by the rules
\mar{ws46}\ben
&& u\rfloor(bda) =(-1)^{[u][b]}u(a), \nonumber\\
&& u\rfloor(\f\w\f')=
(u\rfloor\f)\w\f'+(-1)^{|\f|+[\f][u]}\f\w(u\rfloor\f').
\label{ws46}
\een
As a consequence, any graded derivation $u\in\gd\cA$ of $\cA$
yields a derivation
\mar{+117}\ben
&& \bL_u\f= u\rfloor d\f + d(u\rfloor\f), \qquad \f\in\cO^*, \qquad
u\in\gd\cA, \label{+117} \\
&& \bL_u(\f\w\f')=\bL_u(\f)\w\f' + (-1)^{[u][\f]}\f\w\bL_u(\f'), \nonumber
\een
of the bigraded algebra $\cO^*$ called the {\it graded Lie
derivative} \index{Lie derivative!of a bigraded algebra} of
$\cO^*$.

\begin{rem}
Since $\gd\cA$ is a Lie $\cK$-superalgebra, let us consider the
Chevalley--Eilenberg complex $C^*[\gd\cA;\cA]$ where the graded
commutative ring $\cA$ is a regarded as a $\gd\cA$-module
\cite{fuks}. It is the complex
\mar{ws85}\beq
0\to \cA\ar^{\dl^0}C^1[\gd\cA;\cA]\ar^{\dl^1}\cdots
C^k[\gd\cA;\cA]\ar^{\dl^k}\cdots \label{ws85}
\eeq
where
\be
C^k[\gd\cA;\cA]=\hm_\cK(\op\w^k \gd\cA,\cA)
\ee
are $\gd\cA$-modules of $\cK$-linear graded morphisms of the
graded exterior products $\op\w^k \gd\cA$ of the $\cK$-module
$\gd\cA$ to $\cA$. Let us bring homogeneous elements of $\op\w^k
\gd\cA$ into the form
\be
\ve_1\w\cdots\ve_r\w\e_{r+1}\w\cdots\w \e_k, \qquad
\ve_i\in\gd\cA_0, \quad \e_j\in\gd\cA_1.
\ee
Then the coboundary operators of the complex (\ref{ws85}) are
given by the expression
\mar{ws86}\ben
&& \dl^{r+s-1}c(\ve_1\w\cdots\w\ve_r\w\e_1\w\cdots\w\e_s)=
\label{ws86}\\
&&\op\sum_{i=1}^r (-1)^{i-1}\ve_i
c(\ve_1\w\cdots\wh\ve_i\cdots\w\ve_r\w\e_1\w\cdots\e_s)+
\nonumber \\
&& \op\sum_{j=1}^s (-1)^r\ve_i
c(\ve_1\w\cdots\w\ve_r\w\e_1\w\cdots\wh\e_j\cdots\w\e_s)
+\nonumber\\
&& \op\sum_{1\leq i<j\leq r} (-1)^{i+j}
c([\ve_i,\ve_j]\w\ve_1\w\cdots\wh\ve_i\cdots\wh\ve_j
\cdots\w\ve_r\w\e_1\w\cdots\w\e_s)+\nonumber\\
&&\op\sum_{1\leq i<j\leq s} c([\e_i,\e_j]\w\ve_1\w\cdots\w
\ve_r\w\e_1\w\cdots
\wh\e_i\cdots\wh\e_j\cdots\w\e_s)+\nonumber\\
&& \op\sum_{1\leq i<r,1\leq j\leq s} (-1)^{i+r+1}
c([\ve_i,\e_j]\w\ve_1\w\cdots\wh\ve_i\cdots\w\ve_r\w
\e_1\w\cdots\wh\e_j\cdots\w\e_s).\nonumber
\een
The subcomplex $\cO^*[\gd\cA]$ of the complex (\ref{ws85}) of
$\cA$-linear morphisms is the graded Chevalley--Eilenberg
differential calculus over a graded commutative $\cK$-ring $\cA$.
Then one can show that the above mentioned graded differential
calculus $\cO^*$ is a subcomplex of the Chevalley--Eilenberg one
$\cO^*[\gd\cA]$.
\end{rem}

Following the construction of a connection in commutative geometry
in Section 1.3, one comes to the notion of a connection on modules
over a graded commutative $\Bbb R$-ring $\cA$. The following is a
straightforward counterpart of Definitions \ref{1016} and
\ref{mos088}.

\begin{defi} \label{ws33} \mar{ws33}
A {\it connection} \index{connection!on a graded module} on a
graded $\cA$-module $P$ is an $\cA$-module morphism
\mar{ws34}\beq
\gd\cA\ni u\mapsto \nabla_u\in \dif_1(P,P) \label{ws34}
\eeq
such that the first order differential operators $\nabla_u$ obey
the Leibniz rule
\mar{ws35}\beq
\nabla_u (ap)= u(a)p+ (-1)^{[a][u]}a\nabla_u(p), \quad a\in \cA,
\quad p\in P. \label{ws35}
\eeq
\end{defi}

\begin{defi} \label{ws36} \mar{ws36}
Let $P$ in Definition \ref{ws33} be a graded commutative
$\cA$-ring and $\gd P$ the derivation module of $P$ as a graded
commutative $\cK$-ring. A {\it connection} \index{connection!on a
graded commutative ring} on a graded commutative $\cA$-ring $P$ is
a $\cA$-module morphism
\mar{ws37}\beq
\gd\cA\ni u\mapsto \nabla_u\in \gd P, \label{ws37}
\eeq
which is a connection on $P$ as an $\cA$-module, i.e., obeys the
Leibniz rule (\ref{ws35}).
\end{defi}

\section{Geometry of graded manifolds}

By a {\it graded manifold} \index{graded manifold} of dimension
$(n,m)$ is meant a local-ringed space $(Z,\gA)$ where $Z$ is an
$n$-dimensional smooth manifold $Z$ and $\gA=\gA_0\oplus\gA_1$ is
a sheaf of graded commutative algebras of rank $m$ such that
\cite{bart}:

(i) there is the exact sequence of sheaves
\mar{cmp140}\beq
0\to \cR \to\gA \op\to^\si C^\infty_Z\to 0, \qquad
\cR=\gA_1+(\gA_1)^2,\label{cmp140}
\eeq
where $C^\infty_Z$ is the sheaf of smooth real functions on $Z$;

(ii) $\cR/\cR^2$ is a locally free sheaf of $C^\infty_Z$-modules
of finite rank (with respect to pointwise operations), and the
sheaf $\gA$ is locally isomorphic to the the exterior product
$\w_{C^\infty_Z}(\cR/\cR^2)$.

The sheaf $\gA$ is called a {\it structure sheaf} \index{structure
sheaf!of a graded manifold} of the graded manifold $(Z,\gA)$,
while the manifold $Z$ is said to be a {\it body} \index{body!of a
graded manifold} of $(Z,\gA)$. Sections of the sheaf $\gA$ are
called {\it graded functions}. \index{graded function} They make
up a graded commutative $C^\infty(Z)$-ring $\gA(Z)$.

A graded manifold $(Z,\gA)$ has the following local structure.
Given a point $z\in Z$, there exists its open neighborhood $U$,
called a {\it splitting domain}, \index{splitting domain} such
that
\mar{+54}\beq
\gA(U)\cong C^\infty(U)\ot\w\Bbb R^m. \label{+54}
\eeq
It means that the restriction $\gA|_U$ of the structure sheaf
$\gA$ to $U$ is isomorphic to the sheaf $C^\infty_U\ot\w\Bbb R^m$
of sections of some exterior bundle
\be
\w E^*_U= U\times \w\Bbb R^m\to U.
\ee

The well-known {\it Batchelor's theorem} \index{Batchelor's
theorem} \cite{bart,batch1,book05} states that such a structure of
graded manifolds is global.

\begin{theo} \label{lmp1} \mar{lmp1}
Let $(Z,\gA)$ be a graded manifold. There exists a vector bundle
$E\to Z$ with an $m$-dimensional typical fibre $V$ such that the
structure sheaf $\gA$ of $(Z,\gA)$ is isomorphic to the structure
sheaf $\gA_E$ of sections of the exterior bundle $\w E^*$, whose
typical fibre is the Grassmann algebra $\w V^*$.
\end{theo}

It should be emphasized that Batchelor's isomorphism in Theorem
\ref{lmp1} fails to be canonical. At the same time, there are many
physical models where a vector bundle $E$ is introduced from the
beginning. In this case, it suffices to consider the structure
sheaf $\gA_E$ of the exterior bundle $\w E^*$. We agree to call
the pair $(Z,\gA_E)$ a {\it simple graded manifold}. \index{graded
manifold!simple} Its automorphisms are restricted to those induced
by automorphisms of the vector bundle $E\to Z$, called the {\it
characteristic vector bundle} \index{characteristic vector bundle}
of the simple graded manifold $(Z,\gA_E)$. Accordingly, the
structure module
\be
\gA_E(Z)=\w E^*(Z)
\ee
of the sheaf $\gA_E$ (and of the exterior bundle $\w E^*$) is said
to be the {\it structure module of the simple graded manifold}
$(Z,\gA_E)$. \index{structure module!of a simple graded manifold}

Given a simple graded manifold $(Z,\gA_E)$, every trivialization
chart $(U; z^A,y^a)$ of the vector bundle $E\to Z$ is a splitting
domain of $(Z,\gA_E)$. Graded functions on such a chart are
$\La$-valued functions
\mar{z785}\beq
f=\op\sum_{k=0}^m \frac1{k!}f_{a_1\ldots a_k}(z)c^{a_1}\cdots
c^{a_k}, \label{z785}
\eeq
where $f_{a_1\cdots a_k}(z)$ are smooth functions on $U$ and
$\{c^a\}$ is the  fibre basis for $E^*$. In particular, the sheaf
epimorphism $\si$ in (\ref{cmp140}) is induced by the body map of
$\La$. We agree to call $\{z^A,c^a\}$ the local {\it basis}
\index{basis!for a graded manifold} for the graded manifold
$(Z,\gA_E)$. Transition functions
\be
y'^a=\rho^a_b(z^A)y^b
\ee
of bundle coordinates on $E\to Z$ induce the corresponding
transformation
\mar{+6}\beq
c'^a=\rho^a_b(z^A)c^b \label{+6}
\eeq
of the associated local basis for the graded manifold $(Z,\gA_E)$
and the according coordinate transformation law of graded
functions (\ref{z785}).

\begin{rem}
Although graded functions are locally represented by $\La$-valued
functions (\ref{z785}), they are not $\La$-valued functions on a
manifold $Z$ because of the transformation law (\ref{+6}).
\end{rem}

\begin{rem}
Let us note that general automorphisms of a graded manifold take
the form
\mar{+95}\beq
c'^a=\rho^a(z^A,c^b), \label{+95}
\eeq
where $\rho^a(z^A,c^b)$ are local graded functions. Considering
simple graded manifolds, we actually restrict the class of graded
manifold transformations (\ref{+95}) to the linear ones
(\ref{+6}), compatible with given Batchelor's isomorphism.
\end{rem}

Let $E\to Z$ and $E'\to Z$ be vector bundles and $\Phi: E\to E'$
their bundle morphism over a morphism $\zeta: Z\to Z'$. Then every
section $s^*$ of the dual bundle $E'^*\to Z'$ defines the
pull-back section $\Phi^*s^*$ of the dual bundle $E^*\to Z$ by the
law
\be
v_z\rfloor \Phi^*s^*(z)=\Phi(v_z)\rfloor s^*(\zeta(z)), \qquad
v_z\in E_z.
\ee
 It follows that a linear bundle morphism
$\Phi$ yields a morphism
\mar{w901}\beq
S\Phi: (Z,\gA_E) \to (Z',\gA_{E'}) \label{w901}
\eeq
 of simple graded manifolds seen as
local-ringed spaces. This is the pair $(\zeta,\zeta_*\circ\Phi^*)$
of the morphism $\zeta$ of the body manifolds and the composition
of the pull-back
\be
 \gA_{E'}\ni
f\mapsto \Phi^*f\in\gA_E
\ee
 of graded functions and the direct
image $\zeta_*$ of the sheaf $\gA_E$ onto $Z'$. Relative to local
bases $(z^A,c^a)$ and $(z'^A,c'^a)$ for $(Z,\gA_E)$ and
$(Z',\gA_{E'})$ respectively, the morphism (\ref{w901}) reads
\be
S\Phi(z)=\zeta(z), \qquad S\Phi(c'^a)=\Phi^a_b(z)c^b.
\ee
Accordingly, the pull-back onto $Z$ of graded exterior forms on
$Z'$ is defined.

Given a graded manifold $(Z,\gA)$, by the sheaf $\gd\gA$ of graded
derivations of $\gA$ is meant a subsheaf of endomorphisms of the
structure sheaf $\gA$ such that any section $u$ of $\gd\gA$ over
an open subset $U\subset Z$ is a {\it graded derivation}
\index{derivation!graded} of the graded  algebra $\gA(U)$.
Conversely, one can show that, given open sets $U'\subset U$,
there is a surjection of the derivation modules
\be
\gd(\gA(U))\to \gd(\gA(U'))
\ee
\cite{bart}. It follows that any graded derivation of the local
graded algebra $\gA(U)$ is also a local section over $U$ of the
sheaf $\gd\gA$. Sections of $\gd\gA$ are  called {\it graded
vector fields} \index{graded vector field} on the graded manifold
$(Z,\gA)$. They make up the graded derivation module $\gd\gA(Z)$
of the graded commutative $\Bbb R$-ring $\gA(Z)$. This module is a
Lie superalgebra with respect to the superbracket (\ref{ws14}).

In comparison with general theory of graded manifolds, an
essential simplification is that graded vector fields on a simple
graded manifold $(Z,\gA_E)$ can be seen as sections of a vector
bundle as follows.

Due to the vertical splitting
\be
VE\cong E\times E,
\ee
the vertical tangent bundle $VE$ of $E\to Z$ can be provided with
the fibre bases $\{\dr/\dr c^a\}$, which are the duals of the
bases $\{c^a\}$. These are the fibre bases for
\be
\pr_2VE\cong E.
\ee
Then graded vector fields on a trivialization chart $(U;z^A,y^a)$
of $E$ read
\mar{hn14}\beq
u= u^A\dr_A + u^a\frac{\dr}{\dr c^a}, \label{hn14}
\eeq
where $u^\la, u^a$ are local graded functions on $U$. In
particular,
\be
\frac{\dr}{\dr c^a}\circ\frac{\dr}{\dr c^b} =-\frac{\dr}{\dr
c^b}\circ\frac{\dr}{\dr c^a}, \qquad \dr_A\circ\frac{\dr}{\dr
c^a}=\frac{\dr}{\dr c^a}\circ \dr_A.
\ee
The derivations (\ref{hn14}) act on graded functions
$f\in\gA_E(U)$ (\ref{z785}) by the rule
\mar{cmp50a}\beq
u(f_{a\ldots b}c^a\cdots c^b)=u^A\dr_A(f_{a\ldots b})c^a\cdots c^b
+u^k f_{a\ldots b}\frac{\dr}{\dr c^k}\rfloor (c^a\cdots c^b).
\label{cmp50a}
\eeq
This rule implies the corresponding coordinate transformation law
\be
u'^A =u^A, \qquad u'^a=\rho^a_ju^j +u^A\dr_A(\rho^a_j)c^j
\ee
of graded vector fields. It follows that graded vector fields
(\ref{hn14}) can be represented by sections of the vector bundle
$\cV_E\to Z$ which is locally isomorphic to the vector bundle
\mar{+243}\beq
\cV_E|_U\approx\w E^*\op\ot_Z(E\op\oplus_Z TZ)|_U, \label{+243}
\eeq
and is characterized by an atlas of bundle coordinates
\be
(z^A,z^A_{a_1\ldots a_k},v^i_{b_1\ldots b_k}), \qquad
k=0,\ldots,m,
\ee
possessing the transition functions
\be
 && z'^A_{i_1\ldots
i_k}=\rho^{-1}{}_{i_1}^{a_1}\cdots
\rho^{-1}{}_{i_k}^{a_k} z^A_{a_1\ldots a_k}, \\
&& v'^i_{j_1\ldots j_k}=\rho^{-1}{}_{j_1}^{b_1}\cdots
\rho^{-1}{}_{j_k}^{b_k}\left[\rho^i_jv^j_{b_1\ldots b_k}+
\frac{k!}{(k-1)!} z^A_{b_1\ldots b_{k-1}}\dr_A\rho^i_{b_k}\right],
\ee
which fulfil the cocycle condition. Thus, the derivation module
$\gd\gA_E(Z)$ is isomorphic to the structure module $\cV_E(Z)$ of
global sections of the vector bundle $\cV_E\to Z$.

There is the exact sequence
\mar{1030}\beq
0\to \w E^*\op\ot_Z E\to\cV_E\to \w E^*\op\ot_Z TZ\to 0
\label{1030}
\eeq
of vector bundles over $Z$. Its splitting
\mar{cmp70}\beq
\wt\g:\dot z^A\dr_A \mapsto \dot z^A(\dr_A
+\wt\g_A^a\frac{\dr}{\dr c^a}) \label{cmp70}
\eeq
transforms every vector field $\tau$ on $Z$ into the graded vector
field
\mar{ijmp10}\beq
\tau=\tau^A\dr_A\mapsto \nabla_\tau=\tau^A(\dr_A
+\wt\g_A^a\frac{\dr}{\dr c^a}), \label{ijmp10}
\eeq
which is a graded derivation of the graded commutative $\Bbb
R$-ring $\gA_E(Z)$ satisfying the Leibniz rule
\be
\nabla_\tau(sf)=(\tau\rfloor ds)f +s\nabla_\tau(f), \quad
f\in\gA_E(Z), \quad s\in C^\infty(Z).
\ee
It follows that the splitting (\ref{cmp70}) of the exact sequence
(\ref{1030}) yields a connection on the graded commutative
$C^\infty(Z)$-ring $\gA_E(Z)$ in accordance with Definition
\ref{ws36}. It is called a {\it graded connection} \index{graded
connection} on the simple graded manifold $(Z,\gA_E)$. In
particular, this connection provides the corresponding horizontal
splitting
\be
u= u^A\dr_A + u^a\frac{\dr}{\dr c^a}=u^A(\dr_A
+\wt\g_A^a\frac{\dr}{\dr c^a}) + (u^a- u^A\wt\g_A^a)\frac{\dr}{\dr
c^a}
\ee
of graded vector fields. A graded connection (\ref{cmp70}) always
exists \cite{book05}.

\begin{rem} \label{+94} \mar{+94}
By virtue of the isomorphism (\ref{+54}), any connection $\wt \g$
on a graded manifold $(Z,\gA)$, restricted to a splitting domain
$U$, takes the form (\ref{cmp70}). Given two splitting domains $U$
and $U'$ of $(Z,\gA)$ with the transition functions (\ref{+95}),
the connection components $\wt\g^a_A$ obey the transformation law
\mar{+96}\beq
\wt\g'^a_A= \wt\g^b_A\frac{\dr}{\dr c^b}\rho^a +\dr_A\rho^a.
\label{+96}
\eeq
If $U$ and $U'$ are the trivialization charts of the same vector
bundle $E$ in Theorem \ref{lmp1} together with the transition
functions (\ref{+6}), the transformation law (\ref{+96}) takes the
form
\mar{+97}\beq
\wt\g'^a_A= \rho^a_b(z)\wt\g^b_A +\dr_A\rho^a_b(z)c^b. \label{+97}
\eeq
\end{rem}

\begin{rem}
It should be emphasized that the above notion of a graded
connection is a connection on the graded commutative ring
$\gA_E(Z)$ seen as a $C^\infty(Z)$-module. It differs from that of
a  connection on a graded fibre bundle $(Z,\gA)\to (X,\cB)$ in
\cite{alm}. The latter is a connection on a graded $\cB(X)$-module
represented by a section of the jet graded bundle $J^1(Z/X)\to
(Z,\gA)$ of sections of the graded fibre bundle
   $(Z,\gA)\to (X,\cB)$ \cite{rup}.
\end{rem}

\begin{ex} \label{1031} \mar{1031}
Every linear connection
\be
\g=dz^A\ot (\dr_A +\g_A{}^a{}_by^b \dr_a)
\ee
on the vector bundle $E\to Z$ yields the graded connection
\mar{cmp73}\beq
\g_S=dz^A\ot (\dr_A +\g_A{}^a{}_bc^b\frac{\dr}{\dr c^a})
\label{cmp73}
\eeq
on the simple graded manifold $(Z,\gA_E)$. In view of Remark
\ref{+94}, $\g_S$ is also a graded connection on the graded
manifold
\be
(Z,\gA)\cong (Z,\gA_E),
\ee
but its linear form (\ref{cmp73}) is not maintained under the
transformation law (\ref{+96}).
\end{ex}

The {\it curvature} \index{curvature!of a graded connection} of
the graded connection $\nabla_\tau$ (\ref{ijmp10}) is defined by
the expression (\ref{t59}):
\mar{+110}\ben
&&
R(\tau,\tau')=[\nabla_\tau,\nabla_{\tau'}]-\nabla_{[\tau,\tau']},\nonumber\\
&& R(\tau,\tau') =\tau^A\tau'^B R^a_{AB}\frac{\dr}{\dr c^a}:
\gA_E(Z)\to \gA_E(Z),
\nonumber\\
&&R^a_{AB} =\dr_A\wt\g^a_B-\dr_B\wt\g^a_A
+\wt\g^k_A\frac{\dr}{\dr c^k}\wt\g^a_B -
   \wt\g^k_B\frac{\dr}{\dr c^k} \wt\g^a_A.\label{+110}
\een
It can also be written in the form
\mar{+111}\ben
&&R:\gA_E(Z)\to \cO^2(Z)\ot \gA_E(Z), \nonumber \\
&& R =\frac12 R_{AB}^a dz^A\w dz^B\ot\frac{\dr}{\dr c^a}. \label{+111}
\een

Let now $\cV^*_E\to  Z$ be a vector bundle which is the pointwise
$\w E^*$-dual of the vector bundle $\cV_E\to Z$. It is locally
isomorphic to the vector bundle
\mar{+244}\beq
\cV^*_E|_U\approx \w E^*\op\ot_Z(E^*\op\oplus_Z T^*Z)|_U.
\label{+244}
\eeq
With respect to the dual bases $\{dz^A\}$ for $T^*Z$ and
$\{dc^b\}$ for
\be
\pr_2V^*E\cong E^*,
\ee
sections of the vector bundle $\cV^*_E$ take the coordinate form
\be
\f=\f_A dz^A + \f_adc^a,
\ee
together with transition functions
\be
\f'_a=\rho^{-1}{}_a^b\f_b, \qquad \f'_A=\f_A
+\rho^{-1}{}_a^b\dr_A(\rho^a_j)\f_bc^j.
\ee
They are regarded as {\it graded exterior one-forms} \index{graded
exterior form} on the graded manifold $(Z,\gA_E)$, and make up the
$\gA_E(Z)$-dual
\be
\cC^1_E=\gd\gA_E(Z)^*
\ee
of the derivation module
\be
\gd\gA_E(Z)=\cV_E(Z).
\ee
Conversely,
\be
\gd\gA_E(Z)=(\cC^1_E)^*.
\ee
The duality morphism is given by the graded interior product
\mar{cmp65}\beq
u\rfloor \f=u^A\f_A + (-1)^{\nw{\f_a}}u^a\f_a. \label{cmp65}
\eeq
In particular, the dual of the exact sequence (\ref{1030}) is the
exact sequence
\mar{cmp72}\beq
0\to \w E^*\op\ot_ZT^*Z\to\cV^*_E\to \w E^*\op\ot_Z E^*\to 0.
\label{cmp72}
\eeq
Any graded connection $\wt\g$ (\ref{cmp70}) yields the splitting
of the exact sequence (\ref{cmp72}), and determines the
corresponding decomposition of graded one-forms
\be
\f=\f_A dz^A + \f_adc^a =(\f_A+\f_a\wt\g_A^a)dz^A +\f_a(dc^a
-\wt\g_A^adz^A).
\ee

Higher degree graded exterior forms  are defined as sections of
the exterior bundle $\op\w^k_Z\cV^*_E$. They make up a bigraded
algebra $\cC^*_E$ which is isomorphic to the bigraded exterior
algebra of the graded module $\cC^1_E$ over $\cC^0_E=\gA(Z)$. This
algebra is locally generated by graded forms $dz^A$, $dc^i$ such
that
\mar{+113'}\beq
dz^A\w dc^i=-dc^i\w dz^A, \qquad dc^i\w dc^j= dc^j\w dc^i.
\label{+113'}
\eeq

The {\it graded exterior differential} \index{graded exterior
differential} $d$ of graded functions is introduced by the
condition $u\rfloor df=u(f)$ for an arbitrary graded vector field
$u$, and is extended uniquely to graded exterior forms by the rule
(\ref{ws44}). It is given by the coordinate expression
\be
d\f= dz^A \w \dr_A\f +dc^a\w \frac{\dr}{\dr c^a}\f,
\ee
where the derivatives $\dr_\la$, $\dr/\dr c^a$ act on coefficients
of graded exterior forms by the formula (\ref{cmp50a}), and they
are graded commutative with the graded forms $dz^A$ and $dc^a$.
The formulae (\ref{ws45}) -- (\ref{+117}) hold.

The graded exterior differential $d$ makes $\cC^*_E$ into a
bigraded differential algebra whose de Rham complex reads
\mar{+137}\beq
0\to\Bbb R\to \gA_E(Z)\ar^d \cC^1_E \ar^d \cdots \cC^k_E \ar^d
\cdots. \label{+137}
\eeq
Its cohomology $H^*_{GR}(Z)$  is called the {\it graded de Rham
cohomology} \index{graded de Rham cohomology} of the graded
manifold $(Z,\gA_E)$. One can compute this cohomology with the aid
of the abstract de Rham theorem. Let $\gO^k\gA_E$ denote the sheaf
of germs of graded $k$-forms on $(Z,\gA_E)$. Its structure module
is $\cC^k_E$. These sheaves make up the complex
\mar{1033}\beq
0\to\Bbb R\ar \gA_E \ar^d \gO^1\gA_E\ar^d\cdots
\gO^k\gA_E\ar^d\cdots. \label{1033}
\eeq
Its members $\gO^k\gA_E$ are sheaves of $C^\infty_Z$-modules on
$Z$ and, consequently, are fine and acyclic. Furthermore, the
Poincar\'e lemma for graded exterior forms holds \cite{bart}. It
follows that the complex (\ref{1033}) is a fine resolution of the
constant sheaf $\Bbb R$ on the manifold $Z$.  Then, by virtue of
Theorem \ref{spr230}, there is
   an isomorphism
\mar{+136}\beq
H^*_{GR}(Z)=H^*(Z;\Bbb R)=H^*(Z) \label{+136}
\eeq
of the graded de Rham cohomology $H^*_{GR}(Z)$ to the de Rham
cohomology $H^*(Z)$ of the smooth manifold $Z$. Moreover, the
cohomology isomorphism (\ref{+136}) accompanies the cochain
monomorphism $\cO^*(Z)\to \cC^*_E$ of the de Rham complex
$\cO^*(Z)$ (\ref{t37}) of smooth exterior forms on $Z$ to the
graded de Rham complex (\ref{+137}). Hence, any closed graded
exterior form is split into a sum $\f=d\si +\vf$ of an exact
graded exterior form $d\si\in \cO^*\gA_E$ and a closed exterior
form $\vf\in \cO^*(Z)$ on $Z$.

\section{Supermanifolds}

There are different types of supermanifolds. These are
$H^\infty$-, $G^\infty$-, $GH^\infty$-, $G$-, and DeWitt
supermanifolds \cite{bart,book05,sard09a}. By analogy with smooth
manifolds, supermanifolds are constructed by gluing together of
open subsets of supervector spaces $B^{n,m}$ with the aid  of
transition superfunctions. Therefore, let us start with the notion
of a superfunction.

Though there are different classes of superfunctions, they can be
introduced in the same manner as follows.

Let
\be
B^{n,m}=\La^n_0\oplus \La^m_1,\qquad n,m\geq 0,
\ee
be a supervector space, where $\La$ is a Grassmann algebra of rank
$0<N\geq m$. Let
\be
\si^{n,m}: B^{n,m}\to \Bbb R^n, \qquad s: B^{n,m}\to R^{n,m}=R^n_0
\oplus R_1^m
\ee
be the corresponding body and soul maps (see the decomposition
(\ref{+11})). Then any element $q\in B^{n,m}$ is uniquely split as
\mar{+10}\beq
q=(x,y)=(\si(x^i) + s(x^i))e^0_i + y^je^1_j, \label{+10}
\eeq
where $\{e^0_i,e^1_j\}$ is a basis for $B^{n,m}$ and $\si(x^i)\in
\Bbb R$, $s(x^i)\in R_0$, $y^j\in R_1$.

Let $\La'$ be another Grassmann algebra of rank $0\leq N'\leq N$
which is treated as a subalgebra of $\La$, i.e., the basis
$\{c^a\}$, $a=1,\ldots,N'$, for $\La'$ is a subset of the basis
$\{c^i\}$, $i=1,\ldots,N$, for $\La$. Given an open subset
$U\subset \Bbb R^n$, let us consider a $\La'$-valued function
\mar{+12}\beq
f(z)=\op\sum_{k=0}^{N'} \frac1{k!}f_{a_1\ldots
a_k}(z)c^{a_1}\cdots c^{a_k} \label{+12}
\eeq
on $U$ with smooth coefficients $f_{a_1\cdots a_k}(z)$, $z\in U$.
It is a graded function on $U$. Its prolongation to
$(\si^{n,0})^{-1}(U)\subset B^{n,0}$ is defined as the formal
Taylor series
\mar{+14}\beq
f(x)= \op\sum_{k=0}^{N'} \frac1{k!}\left[
   \op\sum_{p=0}^N\frac{1}{p!}\frac{\dr^pf_{a_1\ldots a_k}}{\dr
z^{i_1}\cdots \dr z^{i_p}}(\si(x))s(x^{i_1})\cdots
s(x^{i_p})\right]c^{a_1}\cdots c^{a_k}. \label{+14}
\eeq
Then a {\it superfunction} \index{superfunction} $F(q)$ on
\be
(\si^{n,m})^{-1}(U)\subset B^{n,m}
\ee
is given by a sum
\mar{+13}\beq
F(q)=F(x,y)= \op\sum_{r=0}^m \frac1{r!} f_{j_1\ldots
j_r}(x)y^{j_1}\cdots y^{j_r}, \label{+13}
\eeq
where $f_{j_1\ldots j_r}(x)$ are functions (\ref{+14}). However,
the representation of a superfunction $F(x,y)$ by the sum
(\ref{+13}) need not be unique.

The germs of superfunctions (\ref{+13}) constitute the sheaf
$\gS_{N'}$ of graded commutative $\La'$-algebras on $B^{n,m}$, but
it is not a sheaf of $C^\infty_{B^{n,m}}$-modules since
superfunctions are expressed in Taylor series.

Using the representation (\ref{+13}), one can define derivatives
of superfunctions as follows. Let $f(x)$ be a superfunction  on
$B^{n,0}$. Since $f$, by definition, is the Taylor series
   (\ref{+14}), its
partial derivative along an even coordinate $x^i$ is defined in a
natural way as
\mar{+16}\ben
&& \dr_if(x)=(\dr_if)(\si(x),s(x))=
\label{+16}\\ && \qquad \op\sum_{k=0}^{N'} \frac1{k!}\left[
   \op\sum_{p=0}^N\frac{1}{p!}\frac{\dr^{p+1}f_{a_1\ldots a_k}}{\dr
z^i\dr z^{i_1}\cdots \dr z^{i_p}}(\si(x))s(x^{i_1})\cdots
s(x^{i_p})\right]c^{a_1}\cdots c^{a_k}. \nonumber
\een
This even derivative is extended to superfunctions $F$ on
$B^{n,m}$ in spite of the fact that the representation (\ref{+13})
is not necessarily unique. However, the definition of odd
derivatives of superfunctions is more intricate.

Let $\gS_{N'}^0\subset \gS_{N'}$ be the subsheaf of superfunctions
$F(x,y)=f(x)$ (\ref{+14}) independent of the odd arguments $y^j$.
Let $\w\Bbb R^m$ be a Grassmann algebra generated by
$(a^1,\ldots,a^m)$. The expression (\ref{+13}) implies that, for
any open subset $U\subset B^{n,m}$, there exists the sheaf
morphism
\mar{+15,7}\ben
&&\la: \gS^0_{N'}\ot\w\Bbb R^m \to \gS_{N'}, \label{+15}\\
&& \la(x,y):\op\sum_{r=0}^m \frac1{r!} f_{j_1\ldots
j_r}(x)\ot(a^{j_1}\cdots a^{j_r})\to \label{+17}\\
&& \qquad \op\sum_{r=0}^m \frac1{r!} f_{j_1\ldots
j_r}(x)y^{j_1}\cdots y^{j_r},\nonumber
\een
over $B^{n,m}$. Clearly, the morphism $\la$ (\ref{+15}) is an
epimorphism. One can show that this epimorphism is injective and,
consequently, is an isomorphism iff
\mar{+20}\beq
N-N'\geq m \label{+20}
\eeq
\cite{bart}. Roughly speaking, in this case, there exists a tuple
of elements $y^{j_1},\ldots,y^{j_r}\in \La$ for each superfunction
$f$ such that
\be
\la(f\ot(a^{j_1}\cdots a^{j_r}))\neq 0
\ee
at the point $(x,y^{j_1},\ldots,y^{j_m})$ of $B^{n,m}$.

If the condition (\ref{+20}) holds, the representation of each
superfunction $F(x,y)$ by the sum (\ref{+13}) is unique, and it is
an image of some section $f\ot a$ of the sheaf
$\gS^0_{N'}\ot\w\Bbb R^m$ with respect to the morphism $\la$
(\ref{+17}). Then an odd derivative of $F$ is defined as
\be
\frac{\dr}{\dr y^j}(\la(f\ot y))=\la (f\ot \frac{\dr}{\dr
a^j}(a)).
\ee
This definition is consistent only if $\la$ is an isomorphism,
i.e., the relation (\ref{+20}) holds. If otherwise, there exists a
non-vanishing element $f\ot a$ such that
\be
\la(f\ot a)=0,
\ee
whereas
\be
\la (f\ot \dr_j(a))\neq 0.
\ee
For instance, if
\be
N-N'=m-1,
\ee
such an element is
\be
f\ot a=c^1\cdots c^{N'}\ot(a^1\cdots a^m).
\ee

In order to classify superfunctions, we follow the terminology of
\cite{bart,rog86,rog07}.

$\bullet$ If $N'=N$, one deals with {\it
$G^\infty$-superfunctions}, \index{$G^\infty$-superfunctions}
introduced in \cite{rog80}. In this case, the inequality
(\ref{+20}) is not satisfied, unless $m=0$.

$\bullet$ If the condition (\ref{+20}) holds,
$\gS_{N'}=\ccG\cH_{N'}$ is the sheaf of {\it
$GH^\infty$-superfunctions}. \index{$GH^\infty$-superfunctions}

$\bullet$ In particular, if $N'=0$, the condition (\ref{+20}) is
satisfied, and $\gS_{N'}=\cH^\infty$ is the sheaf of {\it
$H^\infty$-superfunctions} \index{$H^\infty$-superfunctions}
\mar{+41}\beq
F(x,y)= \op\sum_{r=0}^m \frac1{r!}\left[
   \op\sum_{p=0}^N\frac{1}{p!}\frac{\dr^pf_{j_1\ldots j_r}}{\dr
z^{i_1}\cdots \dr z^{i_p}}(\si(x))s(x^{i_1})\cdots
s(x^{i_p})\right]y^{j_1}\cdots y^{j_r}, \label{+41}
\eeq
where $f_{j_1\ldots j_r}$ are real functions \cite{batch2,dewt}.

Superfunctions of the above three types are called {\it smooth
superfunctions}. \index{superfunction!smooth} The fourth type of
superfunctions is the following.

Given the sheaf $\ccG\cH_{N'}$ of $GH^\infty$-superfunctions on a
supervector space $B^{n,m}$, let us define the sheaf of graded
commutative $\La$-algebras
\mar{+21}\beq
\ccG_{N'}=\ccG\cH_{N'}\op\ot_{\La'} \La, \label{+21}
\eeq
where $\La$ is regarded as a graded commutative $\La'$-algebra.
The sheaf $\ccG_{N'}$ (\ref{+21}) possesses the following
important properties \cite{bart}.

$\bullet$ There is the {\it evaluation morphism} \index{evaluation
morphism}
\mar{+22}\beq
\dl:\ccG_{N'}\ni F\ot a\mapsto Fa \in C^\La_{B^{n,m}}, \label{+22}\\
\eeq
where
\be
C^\La_{B^{n,m}}\cong C^0_{B^{n,m}}\ot\La
\ee
is the sheaf of continuous $\La$-valued functions on $B^{n,m}$.
Its image is isomorphic to the sheaf $\ccG^\infty$ of
$G^\infty$-superfunctions on $B^{n,m}$.

$\bullet$ For any two integers $N'$ and $N''$ satisfying the
condition (\ref{+20}), there exists the canonical isomorphism
between the sheaves $\ccG_{N'}$ and $\ccG_{N''}$. Therefore, one
can define the canonical sheaf $\ccG_{n,m}$ of graded commutative
$\La$-algebras on the supervector space $B^{n,m}$ whose sections
can be seen as tensor products $F\ot a$ of
$H^\infty$-superfunctions $F$ (\ref{+41}) and elements $a\in\La$.
They are called {\it $G$-superfunctions}.
\index{$G$-superfunction}

$\bullet$ The sheaf $\gd\ccG_{n,m}$ of graded derivations of the
sheaf $\ccG_{n,m}$ is a locally free sheaf of $\ccG_{n,m}$-modules
of rank $(n,m)$. On any open set $U\subset B^{n,m}$, the
$\ccG_{n,m}(U)$-module $\gd\ccG_{n,m}(U)$ is generated by the
derivations $\dr/\dr x^i$, $\dr/\dr y^j$ which act on
$\ccG_{n,m}(U)$ by the rule
\mar{+83}\beq
\frac{\dr}{\dr x^i}(F\ot a)=\frac{\dr F}{\dr x^i}\ot a, \qquad
\frac{\dr}{\dr y^j}(F\ot a)=\frac{\dr F}{\dr y^j}\ot a.
\label{+83}
\eeq

These properties of $G$-superfunctions make $G$-supermanifolds
most suitable for differential geometric constructions.

A paracompact topological space $M$ is said to be an
$(n,m)$-dimensional {\it smooth supermanifold} \index{smooth
supermanifold} if it admits an atlas
\be
\Psi=\{U_\zeta,\f_\zeta\}, \qquad \f_\zeta: U_\zeta\to B^{n,m},
\ee
such that the transition functions $\f_\zeta\circ\f_\xi^{-1}$ are
supersmooth. Obviously, a smooth supermanifold of dimension
$(n,m)$ is also a real smooth manifold of dimension
$2^{N-1}(n+m)$. If transition superfunctions are $H^\infty$-,
$G^\infty$- or $GH^\infty$-superfunctions, one deals with
$H^\infty$-, $G^\infty$- or $GH^\infty$-supermanifolds,
respectively.

By virtue of Theorem \ref{+26} extended to graded local-ringed
spaces, this definition is equivalent to the following one.

\begin{defi} \label{+27} \mar{+27}
A smooth supermanifold is a graded local-ringed space $(M,\gS)$
which is locally isomorphic to $(B^{n,m},\cS)$, where $\cS$ is one
of the sheaves of smooth superfunctions on $B^{n,m}$. The sheaf
$\cS$ is called the {\it structure sheaf} \index{structure
sheaf!of a supermanifold} of a smooth supermanifold.
\end{defi}

In accordance with Definition \ref{+27}, by a morphism of smooth
supermanifolds is meant their morphism $(\vf,\Phi)$ as graded
local-ringed spaces, where $\Phi$ is an even graded morphism. In
particular, every morphism $\vf: M\to M'$ yields the smooth
supermanifold morphism $(\vf,\Phi=\vf^*)$.

Smooth supermanifolds however are effected by serious
inconsistencies as follows. Since odd derivatives of
$G^\infty$-superfunctions are ill defined, the sheaf of
derivations of the sheaf of $G^\infty$-superfunctions is not
locally free. Nevertheless, any $G$-supermanifold has an
underlying $G^\infty$-supermanifold.

In the case of $GH^\infty$-supermanifolds (including
$H^\infty$-ones), spaces of values of $GH^\infty$-superfunctions
at different points are not mutually isomorphic because the
Grassmann algebra $\La$ is not a free module with respect to its
subalgebra $\La'$. By these reasons, we turn to
$G$-supermanifolds. Their definition repeats Definition \ref{+27}.

\begin{defi} \label{+30} \mar{+30}
An $(n,m)$-dimensional {\it $G$-supermanifold}
\index{$G$-supermanifold} is a graded local-ringed space
$(M,G_M)$, satisfying the following conditions:

$\bullet$ $M$ is a paracompact topological space;

$\bullet$ $(M,G_M)$ is locally isomorphic to
$(B^{n,m},\ccG_{n,m})$;

$\bullet$ there exists a morphism of sheaves of graded commutative
$\La$-algebras $\dl:G_M\to C^\La_M$, where
\be
C^\La_M\cong C^0_M\ot\La
\ee
is sheaf of continuous $\La$-valued functions on $M$,  and $\dl$
is locally isomorphic to the evaluation morphism (\ref{+22}).
\end{defi}

\begin{ex} \label{+50} \mar{+50}
The triple $(B^{n,m},\ccG_{n,m},\dl)$, where $\dl$ is the
evaluation morphism (\ref{+22}),  is called the {\it standard
$G$-supermanifold}. \index{$G$-supermanifold standard} For any
open subset $U\subset B^{n,m}$, the space $\ccG_{n,m}(U)$ can be
provided with the topology which makes it into a graded Fr\'echet
algebra. Then there are  isometrical isomorphisms
\mar{+65}\ben
&& \ccG_{n,m}(U)\cong \cH^\infty(U)\ot \La\cong
C^\infty(\si^{n,m}(U))\ot\La\ot\w
\Bbb R^m \cong \label{+65} \\
&&\qquad C^\infty(\si^{n,m}(U))\ot\w\Bbb R^{N+m}. \nonumber
\een
\end{ex}

\begin{rem} \label{+81} \mar{+81}
Any $GH^\infty$-supermanifold $(M,GH^\infty_M)$ with the structure
sheaf $GH^\infty_M$ is naturally extended to the $G$-supermanifold
$(M,GH^\infty_M\ot \La)$. Every $G$-supermanifold defines an {\it
underlying $G^\infty$-supermanifold} \index{underlying
$G^\infty$-supermanifold} $(M,\dl(G_M))$, where
$\dl(G_M)=G^\infty_M$ is the sheaf of $G^\infty$-superfunctions on
$M$.
\end{rem}

As in the case of smooth supermanifolds, the underlying space $M$
of a $G$-supermanifold $(M,G_M)$ is provided with the structure of
a real smooth manifold of dimension $2^{N-1}(n+m)$, and morphisms
of $G$-supermanifolds are smooth morphisms of the underlying
smooth manifolds. However, it may happen that non-isomorphic
$G$-supermanifold have isomorphic underlying smooth manifolds.

Let $(M,G_M)$ be a $G$-supermanifold. Sections $u$ of the sheaf
$\gd G_M$ of graded derivations are called {\it supervector
fields} \index{supervector field} on the $G$-supermanifold
$(M,G_M)$, while sections $\f$ of the dual sheaf $\gd G_M^*$ are
{\it one-superforms} \index{superform} on $(M,G_M)$. Given a
coordinate chart $(q^i)=(x^i,y^j)$ on $U\subset M$, supervector
fields and one-superforms read
\be
u=u^i\dr_i, \qquad \f=\f_idq^i,
\ee
where coefficients $u^i$ and $\f_i$ are $G$-superfunctions on
$U$. The graded differential calculus in supervector fields and
superforms obeys the standard formulae (\ref{ws14}), (\ref{ws45}),
(\ref{ws44}) and (\ref{ws46}).

Let us consider cohomology of $G$-supermanifolds. Given a
$G$-supermanifold $(M,G_M)$, let
\be
\gO^k_{\La M}=\gO^k_M\ot\La
\ee
be the sheaves of smooth $\La$-valued exterior forms on $M$. These
sheaves are fine, and they constitute the fine resolution
\be
0\to\La\to C^\infty_M\ot\La \to \gO^1_M\ot\La\to\cdots
\ee
of the constant sheaf $\La$ on $M$. We have the corresponding de
Rham complex
\be
0\to\La\to C^\infty_\La(M) \to \cO^1_\La(M)\to\cdots
\ee
of $\La$-valued exterior forms on $M$. By virtue of Theorem
\ref{spr230}, the cohomology $H^*_\La(M)$ of this complex is
isomorphic to the sheaf cohomology $H^*(M;\La)$ of $M$ with
coefficients in the constant sheaf $\La$ and, consequently, is
related to the de Rham cohomology as follows:
\mar{+146}\beq
H^*_\La(M)=H^*(M;\La)=H^*(M)\ot\La. \label{+146}
\eeq
Thus, the cohomology groups of $\La$-valued exterior forms do not
provide us with information on the $G$-supermanifold structure of
$M$.

Let us turn to cohomology of superforms on a $G$-supermanifold
$(M,G_M)$. The sheaves $\op\w^k\gd G_M^*$ of superforms constitute
the complex
\mar{+141}\beq
0\to\La\to G_M\to \gd^*G_M\to \cdots. \label{+141}
\eeq
The Poincar\'e lemma for superforms is proved to hold
\cite{bart,bruz88}, and this complex is exact. However, the
structure sheaf $G_M$ need not be acyclic, and the exact sequence
(\ref{+141}) fails to be a resolution of the constant sheaf $\La$
on $M$ in general. Therefore, the cohomology $H^*_S(M)$ of the de
Rham complex of superforms are not equal to cohomology
$H^*(M;\La)$ of $M$ with coefficients in the constant sheaf $\La$,
and need not be related to the de Rham cohomology $H^*(M)$ of the
smooth manifold $M$. In particular, cohomology $H^*_S(M)$ is not a
topological invariant, but it is invariant under $G$-isomorphisms
of $G$-supermanifolds.

\begin{prop} \label{+140} \mar{+140}  The structure sheaf
$\ccG_{n,m}$ of the standard $G$-supermanifold
$(B^{n,m},\ccG_{n,m})$ is acyclic, i.e.,
\be
H^{k>0}(B^{n,m};\ccG_{n,m})=0.
\ee
\end{prop}

The proof is based on the isomorphism (\ref{+65}) and some
cohomological constructions \cite{bart,bruz99}.

\section{Supervector bundles}

As was mentioned above, supervector bundles are considered in the
category of $G$-supermanifolds. We start with the definition of
the product of two $G$-supermanifolds seen as a trivial
supervector bundle.

Let $(B^{n,m},\ccG_{n,m})$ and $(B^{r,s},\ccG_{r,s})$ be two
standard $G$-supermanifolds in Example \ref{+30}. Given open sets
$U\subset B^{n,m}$ and $V\subset B^{r,s}$, we consider the
presheaf
\mar{+67}\beq
U\times V\to \ccG_{n,m}(U)\wh\ot \ccG_{r,s}(V), \label{+67}
\eeq
where $\wh\ot$ denotes the tensor product of modules completed in
Grothendieck's topology (see Remark \ref{ws2}). Due to the
isomorphism (\ref{+65}), it is readily observed that the structure
sheaf $\ccG_{n+r,m+s}$ of the standard $G$-supermanifold on
$B^{n+r,m+s}$ is isomorphic to that, defined by the presheaf
(\ref{+67}). This construction is generalized to arbitrary
$G$-supermanifolds as follows.

Let  $(M,G_M)$ and $(M',G_{M'})$ be two $G$-supermanifolds of
dimensions $(n,m)$ and $(r,s)$, respectively. Their {\it product}
\index{product!of $G$-supermanifolds}
\be
(M,G_M) \times(M',G_{M'})
\ee
is defined as the graded local-ringed space $(M\times M',
G_M\wh\ot G_{M'})$, where $G_M\wh\ot G_{M'}$ is the sheaf
determined by the presheaf
\be
&& U\times U'\to G_M(U)\wh\ot G_{M'}(U'),\\
&& \dl: G_M(U)\wh\ot G_{M'}(U')\to C^\infty_{\si(U)}\wh\ot
C^\infty_{\si(U')}= C^\infty_{\si_M(U)\times \si_M(U')},
\ee
for any open subsets $U\subset M$ and $U'\subset M'$. This product
is a $G$-supermanifold of dimension $(n+r,m+s)$ \cite{bart}.
Furthermore, there is the epimorphism
\be
  \pr_1:(M,G_M) \times(M',G_{M'})\to (M,G_M).
\ee
One may define its section over an open subset $U\subset M$ as the
$G$-supermanifold morphism
\be
s_U:(U,G_M|_U)\to (M,G_M) \times(M',G_{M'})
\ee
such that $\pr_1\circ s_U$ is the identity morphism of
$(U,G_M|_U)$. Sections $s_U$ over all open subsets $U\subset M$
determine a sheaf on $M$. This sheaf should be provided with a
suitable graded commutative $G_M$-structure.

For this purpose, let us consider the product
\mar{+72}\beq
(M,G_M)\times (B^{r\mid s}, \ccG_{r\mid s}), \label{+72}
\eeq
  where $B^{r\mid s}$
is the superspace (\ref{+70}). It is called a {\it  product
$G$-supermanifold}. \index{product $G$-supermanifold} Since the
$\La_0$-modules $B^{r\mid s}$ and $B^{r+s,r+s}$ are isomorphic,
$B^{r\mid s}$ has a natural structure of an
$(r+s,r+s)$-dimensional $G$-supermanifold. Because  $B^{r\mid s}$
is a free graded $\La$-module of the type $(r,s)$, the sheaf
$S_M^{r\mid s}$ of sections of the fibration
\mar{+71}\beq
(M,G_M)\times (B^{r\mid s}, \ccG_{r\mid s})\to (M,G_M) \label{+71}
\eeq
has the structure of the sheaf of  free graded $G_M$-modules of
rank $(r,s)$. Conversely, given a $G$-supermanifold $(M,G_M)$ and
a sheaf $S$ of
  free graded $G_M$-modules of rank
$(r,s)$ on $M$, there exists a product $G$-supermanifold
(\ref{+72}) such that $S$ is isomorphic to the sheaf of sections
of the fibration (\ref{+71}).

Let us turn now to the notion of a supervector bundle over
$G$-supermanifolds. Similarly to smooth vector bundles (see
Theorem \ref{sp60}), one can require of the category of
supervector bundles over $G$-supermanifolds to be equivalent to
the category of locally free sheaves of graded modules on
$G$-supermanifolds. Therefore, we can restrict ourselves to
locally trivial supervector bundles with the standard fibre
$B^{r\mid s}$.

\begin{defi} \label{+78} \mar{+78}
A {\it supervector bundle} \index{supervector bundle}  over a
$G$-supermanifold $(M,G_M)$ with the standard fibre $(B^{r\mid
s},\ccG_{r\mid s})$ is defined as a pair $((Y,G_Y),\pi)$ of a
$G$-supermanifold $(Y,G_Y)$ and a $G$-epimorphism
\mar{+76}\beq
\pi: (Y,G_Y)\to (M,G_M) \label{+76}
\eeq
such that $M$ admits an atlas $\{(U_\zeta,\psi_\zeta\}$ of local
$G$-isomorphisms
\be
\psi_\zeta: (\pi^{-1}(U_\zeta), G_Y|_{\pi^{-1}(U_\zeta)})\to
(U_\zeta, G_M|_{U_\zeta})\times (B^{r\mid s},\ccG_{r\mid s}).
\ee
\end{defi}

It is clear that sections of the supervector bundle (\ref{+76})
constitute a  locally free sheaf of graded $G_M$-modules. The
converse of this fact is the following \cite{bart}.

\begin{theo} \label{+77} \mar{+77}
For any locally free sheaf $S$ of graded $G_M$-modules of rank
$(r,s)$ on a $G$-supermanifold $(M,G_M)$, there exists a
supervector bundle over $(M,G_M)$ such that $S$ is isomorphic to
the structure sheaf of its sections.
\end{theo}

The fibre $Y_q$, $q\in M$, of the supervector bundle in Theorem
\ref{+77} is the quotient
\be
S_q/\cM_q\cong S_{Mq}^{r\mid s}/(\cM_q\cdot S_{Mq}^{r\mid s})\cong
B^{r\mid s}
\ee
of the stalk $S_q$ by the submodule $\cM_q$ of the germs $s\in
S_q$ whose evaluation $\dl(f)(q)$ vanishes. This fibre is a graded
$\La$-module isomorphic to $B^{r\mid s}$, and is provided with the
structure of the standard $G$-supermanifold.

\begin{rem} \label{+79} \mar{+79}
The proof of Theorem \ref{+77} is based on the fact that, given
the transition functions $\rho_{\zeta\xi}$ of the sheaf $S$, their
evaluations
\mar{+80}\beq
g_{\zeta\xi}=\dl(\rho_{\zeta\xi}) \label{+80}
\eeq
  define the morphisms
\be
U_\zeta\cap U_\xi \to  GL(r|s;\La),
\ee
and they are assembled into a cocycle of $G^\infty$-morphisms from
$M$ to the general linear graded  group $GL(r|s;\La)$. Thus, we
come to the notion of a {\it $G^\infty$-vector bundle}.
\index{$G^\infty$-vector bundle} Its definition is a repetition of
Definition \ref{+78} if one replaces $G$-supermanifolds and
$G$-morphisms with the $G^\infty$- ones. Moreover, the
$G^\infty$-supermanifold underlying a supervector bundle (see
Remark \ref{+81}) is a $G^\infty$-supervector bundle, whose
transition functions $g_{\zeta\xi}$ are related to those of the
supervector bundle by the evaluation morphisms (\ref{+80}), and
are $GL(r|s;\La)$-valued transition functions.
\end{rem}

Since the category of supervector bundles over a $G$-supermanifold
$(M,G_M)$ is equivalent to the category of locally free sheaves of
graded $G_M$-modules, one can define the usual operations of
direct sum, tensor product, etc. of supervector bundles.

Let us note that any supervector bundle admits the canonical
global zero section. Any section of the supervector bundle $\pi$
(\ref{+76}), restricted to its trivialization chart
\mar{+156}\beq
(U, G_M\mid_U)\times (B^{r\mid s},\ccG_{r\mid s}), \label{+156}
\eeq
is represented by a sum $s = s^a(q)\e_a$, where $\{\e_a\}$ is the
basis  for the graded $\La$-module $B^{r\mid s}$, while $s^a(q)$
are $G$-superfunctions on $U$. Given another trivialization chart
$U'$ of $\pi$, a transition function
\mar{+155}\beq
s'^b(q)\e'_b=s^a(q)h^b{}_a(q)\e_b, \qquad q\in U\cap U',
\label{+155}
\eeq
is given by the $(r+s)\times(r+s)$ matrix $h$ whose entries
$h^b{}_a(q)$ are $G$-superfunctions on $U\cap U'$. One can think
of this matrix as being a section of the supervector bundle over
$U\cap U$ with the above mentioned group $GL(r|s;\La)$ as a
typical fibre.

\begin{ex} \label{+84} \mar{+84}
Given a $G$-supermanifold $(M,G_M)$, let us consider the locally
free sheaf $\gd G_M$ of graded derivations of $G_M$. In accordance
with Theorem \ref{+77}, there is a supervector bundle $T(M,G_M)$,
called {\it supertangent bundle}, \index{supertangent bundle}
whose structure sheaf is isomorphic to $\gd G_M$. If
$(q^1,\ldots,q^{m+n})$ and $(q'^1,\ldots,q'^{m+n})$ are two
coordinate charts on $M$, the Jacobian matrix
\be
h^i_j=\frac{\dr q'^i}{\dr q^j}, \qquad i,j=1,\ldots, n+m,
\ee
(see the prescription (\ref{+83})) provides the transition
morphisms for $T(M,G_M)$.

It should be emphasized that the underlying $G^\infty$-vector
bundle of the supertangent bundle $T(M,G_M)$, called {\it
$G^\infty$-supertangent bundle}, \index{$G^\infty$-tangent bundle}
has the transition functions $\dl(h^i_j)$ which cannot be written
as the Jacobian matrices since the derivatives of
$G^\infty$-superfunctions with respect to odd arguments are ill
defined and the sheaf $\gd G^\infty_M$ is not locally free.
\end{ex}

\section{Superconnections}

Given a supervector bundle $\pi$ (\ref{+76}) with the structure
sheaf $S$, one can follow suit of Definition \ref{+4} and
introduce a connection on this supervector bundle  as a splitting
of the the exact sequence of sheaves
\mar{+121}\beq
0\to \gd G_M^*\ot S\to (G_M\oplus\gd G_M^*)\ot S\to S\to 0.
\label{+121}
\eeq
Its splitting is an even sheaf morphism
\mar{+122}\beq
\nabla: S\to \gd^*G_M\ot S \label{+122}
\eeq
satisfying the Leibniz rule
\mar{+158}\beq
\nabla (fs)=df\ot s + f\nabla(s), \qquad f\in G_M(U), \qquad s\in
S(U), \label{+158}
\eeq
for any open subset $U\in M$. The sheaf morphism (\ref{+122}) is
called a {\it superconnection} \index{superconnection} on the
supervector bundle $\pi$ (\ref{+76}). Its {\it curvature}
\index{curvature!of a superconnection} is given by the expression
\mar{+157}\beq
R=\nabla^2:S\to \op\w^2\gd G^*_M\ot S, \label{+157}
\eeq
similar to the expression (\ref{+105}).

The exact sequence (\ref{+121}) need not be split. One can apply
the criterion in Section 1.7 in order to study the existence of a
superconnection on supervector bundles. Namely, the exact sequence
(\ref{+121}) leads to the exact sequence of sheaves
\be
0\to \hm(S,\gd G_M^*\ot S)\to \hm(S,(G_M\oplus\gd G_M^*)\ot S) \to
\hm(S,S)\to 0
\ee
and to the corresponding exact sequence of the cohomology groups
\be
&& 0\to H^0(M; \hm(S,\gd G_M^*\ot S)) \to H^0(M;
\hm(S,(G_M\oplus\gd G_M^*)\ot S)) \\
&& \qquad \to H^0(M;\hm(S,S))\to H^1(M;\hm(S,\gd G_M^*\ot S))\to
\cdots.
\ee
The exact sequence  (\ref{+121}) defines the Atiyah class
\be
{\rm At}(\pi)\in H^1(M;\hm(S,\gd G_M^*\ot S))
\ee
 of the supervector bundle $\pi$ (\ref{+76}). If
the Atiyah class vanishes, a superconnection on this supervector
bundle exists. In particular, a superconnection exists if the
  cohomology set $H^1(M;\hm(S,\gd G_M^*\ot S))$ is trivial.
In contrast with the sheaf of smooth functions, the structure
sheaf $G_M$ of a $G$-supermanifold is not acyclic in general,
cohomology $H^*(M;\hm(S,\gd G_M^*\ot S))$ is not trivial, and a
supervector bundle need not admit a superconnection.

\begin{ex} \label{+150} \mar{+150}
In accordance with Proposition \ref{+140}, the structure sheaf of
the standard $G$-supermanifold $(B^{n,m},\ccG_{n,m})$ is acyclic,
and the trivial supervector bundle
\mar{+151}\beq
(B^{n,m},\ccG_{n,m})\times (B^{r\mid s},\ccG_{r\mid s}) \to
(B^{n,m},\ccG_{n,m}) \label{+151}
\eeq
has obviously a superconnection, e.g., the trivial
superconnection.
\end{ex}

Example \ref{+150} enables one to obtain a local coordinate
expression for a superconnection on a supervector bundle $\pi$
(\ref{+76}), whose typical fibre is $B^{r\mid s}$ and whose base
is a $G$-supermanifold locally isomorphic to the standard
$G$-supermanifold $(B^{n,m},\ccG_{n,m})$. Let $U\subset M$
(\ref{+156}) be a trivialization chart of this supervector bundle
such that every section $s$ of $\pi|_U$ is represented by a sum
$s^a(q)\e_a$, while the
  sheaf of one-superforms
$\gd^* G_M|_U$ has a local basis $\{d q^i\}$.
  Then a
superconnection $\nabla$ (\ref{+122}) restricted to this
trivialization chart can be given by a collection of coefficients
$\nabla_i{}^a{}_b$:
\mar{+160}\beq
\nabla (\e_a)=dq^i\ot (\nabla_i{}^b{}_a\e_b), \label{+160}
\eeq
which  are $G$-superfunctions on $U$. Bearing in mind the Leibniz
rule (\ref{+158}), one can compute the coefficients of the
curvature form (\ref{+157}) of the superconnection (\ref{+160}).
We have
\be
&& R(\e_a)=\frac12 dq^i\w dq^j\ot R_{ij}{}^b{}_a\e_b, \\
&& R_{ij}{}^a{}_b =(-1)^{[i][j]}\dr_i\nabla_j{}^a{}_b -\dr_j\nabla_i{}^a{}_b
+ (-1)^{[i]([j]+[a]+[k])}\nabla_j{}^a{}_k\nabla_i{}^k{}_b -\\
&& \qquad (-1)^{[j]([a]+[k])}\nabla_i{}^a{}_k\nabla_j{}^k{}_b.
\ee
In a similar way, one can obtain the transformation law of the
superconnection coefficients (\ref{+160}) under the transition
morphisms (\ref{+155}). In particular, any trivial supervector
bundle admits the trivial superconnection $\nabla_i{}^b{}_a=0$.
\newpage

\chapter{Non-commutative geometry}

Non-commutative geometry is developed in main as a generalization
of the calculus in commutative rings of smooth functions
\cite{conn,dub01,book05,grac,land,mad}.  Accordingly, a
non-commutative generalization of differential geometry is phrased
in terms of the differential calculus over a non-commutative ring
which replaces the exterior algebra of differential forms. The
Chevalley--Eilenberg differential calculus over a commutative ring
is straightforwardly generalized to a non-commutative $\cK$-ring
$\cA$. However, the extension of the notion of a differential
operator to modules over a non-commutative ring meets difficulties
\cite{book05,sard07}. In a general setting, any non-commutative
ring can be called into play, but one often follows the more deep
analogy to the case of commutative smooth function rings. In
Connes' commutative geometry, $\cA$ is the algebra $\Bbb
C^\infty(X)$ of smooth complex functions on a compact manifold
$X$. It is a dense subalgebra of the $C^*$-algebra of continuous
complex functions on $X$. Generalizing this case, Connes'
non-commutative geometry \cite{conn80,conn,conjmp} addresses the
differential calculus over an involutive algebra $\cA$ of bounded
operators in a Hilbert space $E$ and, furthermore, studies a
representation of this differential calculus by operators in $E$.

\section{Modules over $C^*$-algebras}

Let us point out some features of modules over non-commutative
algebras and, in particular, $C^*$-algebras.

Let $\cK$ throughout be a commutative ring and $\cA$ a $\cK$-ring
which need not be commutative. Let $\cZ_\cA$ denote its {\it
center}. \index{center of an algebra} An $\cA$-bimodule throughout
is assumed to be a commutative (central) $\cZ_\cA$-bimodule.
Sometimes, it is convenient to use the following compact
abbreviation \cite{dub}. We say that right and left $\cA$-modules,
$\cA$-bimodules and $\cZ_\cA$-bimodules are $\cA_i$-modules of
type $(1,0)$, $(0,1)$, $(1,1)$ and $(0,0)$, respectively (or
$(\cA_i-\cA_j)$-modules where $\cA_0=\cZ_\cA$ and $\cA_1=\cA$). Of
course, $\cA_i$-modules of type $(1,1)$ are also of type $(1,0)$
and $(0,1)$, while $\cA_i$-modules of type $(1,0)$, $(0,1)$
$(1,1)$ are also of type $(0,0)$. With this abbreviation, the
basic constructions of new modules from old ones are phrased as
follows.

$\bullet$ If $P$ and $P'$ are $\cA_i$-modules of the same type
$(i,j)$, so is its direct sum $P\oplus P'$.

$\bullet$ Let $P$ and $P'$ be $\cA_i$-modules of type $(i,k)$ and
$(k,j)$, respectively. Their tensor product $P\ot P'$ is an
$\cA_i$-module of type $(i,j)$.

$\bullet$ Given an $\cA_i$-module $P$ of type $(i,j)$, its
$\cA$-dual $P^*=\hm_{\cA_i-\cA_j}(P,\cA)$ is a module of type
$(i+1,j+1){\rm mod}\,2$.

Let $A$ be a complex involutive algebra. Any module over $A$ is
also a complex vector space. An $A_i$-module of type $(1,1)$ is
called an {\it involutive module} \index{module!involutive} if it
is equipped with an antilinear involution $p\mapsto p^*$ such that
\be
(apb)^*=b^*p^*a^*, \qquad  a,b\in A, \qquad  p\in P.
\ee
Due to this relation, an involutive module is reconstructed by its
right or left module structure. In particular, an involutive
module is said to be a projective module of finite rank if, seen
as a right (or left) module, it is a finite projective module.

Given a right module $P$ over an involutive algebra $A$, a {\it
Hermitian form} \index{Hermitian form!on a module} on $P$ is
defined as a sesquilinear $A$-valued form
\mar{w88}\ben
&& \lng.|.\rng:P\times P\to A, \label{w88}\\
&& \lng pa|p'a'\rng=a^*\lng p|p'\rng a', \qquad \lng p|p'\rng=\lng
p'|p\rng^*, \qquad p,p'\in P, \quad a,a'\in A. \nonumber
\een
A Hermitian form (\ref{w88}) on $P$ yields an antilinear morphism
$h$ of $P$ to its $A$-dual $P^*$ given by the formula
\mar{w87}\beq
(h p)(p')= \lng p|p'\rng, \qquad p,p'\in P. \label{w87}
\eeq
A Hermitian form (\ref{w88}) is called {\it invertible}
\index{Hermitian form!on a module!invertible} if the morphism $h$
is invertible.

Let $A$ be a $C^*$-algebra. A Hermitian form (\ref{w88}) on a
right $A$-module $P$ is called {\it positive} \index{Hermitian
form!on a module!positive} if $\lng p|p\rng$ for all $p\in P$ is a
{\it positive element} \index{positive element} of a $C^*$-algebra
$A$, i.e.,
\be
\lng p|p\rng=aa^*,\qquad a\in A.
\ee
Let $A$ be a unital $C^*$-algebra. Any projective $A$-module $P$
of finite rank admits an invertible positive Hermitian form.
Moreover, all these forms on $P$ are isomorphic \cite{mis}.

A positive Hermitian form on a right $A$-module $P$ endows $P$
with the semi-norm
\mar{w89}\beq
||p||=||\lng p|p\rng||^{1/2}, \qquad p\in P, \label{w89}
\eeq
where $||\lng p|p\rng||$ is the $C^*$-algebra norm of $\lng
p|p\rng\in A$. Equipped with this seminorm and the corresponding
topology, $P$ is called the (right) {\it pre-Hilbert module}.
\index{pre-Hilbert module} It is a {\it Hilbert module}
\index{Hilbert module} (a {\it $C^*$-module} in the terminology of
\cite{conn}) \index{$C^*$-module} if the seminorm (\ref{w89}) is a
complete norm.

\begin{ex} \label{w94} \mar{w94}
A $C^*$-algebra $A$ is provided with the structure of a Hilbert
$A$-module with respect to the action of $A$ on itself by right
multiplications and the positive Hermitian form
\mar{w95}\beq
\lng a|a'\rng= a^*a', \qquad a,a'\in A. \label{w95}
\eeq
\end{ex}

\begin{ex} \label{w200} \mar{w200}
Let $A=\Bbb C^0(X)$ be the $C^*$-algebra of continuous complex
functions on a compact space $X$, and let $E\to X$ be a
(topological) complex vector bundle endowed with a Hermitian fibre
metric $\lng.|.\rng_x$. Then the space $E(X)$ of continuous
sections of $E\to X$ is a Hilbert $\Bbb C^0(X)$-module with
respect to the $\Bbb C^0(X)$-valued Hermitian form
\be
\lng s|s'\rng (x)= \lng s(x)|s'(x)\rng_x, \qquad s,s'\in E(X).
\ee
\end{ex}

Given a Hilbert $A$-module $P$, by its {\it endomorphism}
\index{endomorphism!of a Hilbert module} $T$ is meant a continuous
$A$-linear endomorphism of a right module $P$ which admits the
{\it adjoint endomorphism} \index{endomorphism!of a Hilbert
module!adjoint} $T^*$, which is uniquely given by the relation
\be
\lng p|Tp'\rng=\lng T^*p|p'\rng, \qquad p,p'\in P.
\ee
The set $B_A(P)$ of $A$-linear endomorphisms of a Hilbert
$A$-module $P$ is a $C^*$-algebra with respect to the operator
norm. {\it Compact endomorphisms} \index{endomorphism!of a Hilbert
module!compact} of $P$ are defined as the closer of its
endomorphisms of finite rank. Let us consider endomorphisms
$T_{p,q}\in B_A(P)$ of $P$ of the form
\mar{w92}\beq
T_{p,q}p'=p\lng p'|q\rng, \qquad p,p',q\in P. \label{w92}
\eeq
They obey the relations
\be
T^*_{p,q}=T_{q,p}, \qquad T_{p,q}T_{p',q'}=T_{p\lng q|p'\rng,q'}=
T_{p,T_{q,p'}q}.
\ee
The linear span of endomorphisms (\ref{w92}) is a two-sided ideal
of $B_A(P)$ \cite{rief72}. Its closure is the set $T_A(P)$ of
compact $A$-linear endomorphisms of $P$.

In conclusion, let us turn to projective Hilbert modules of finite
rank over a unital $C^*$-algebra $A$. One can show the following
\cite{mis,rief72}.

$\bullet$ Let $P$ be a right Hilbert $A$-module such that $\id
P\in T_A(P)$. Then $P$ is a projective module of finite rank.

$\bullet$ Conversely, let $P$ be a projective right $A$-module of
finite rank. Then $P$ admits a positive Hermitian form which makes
it into a Hilbert module such that $\id P\in T_A(P)$.

$\bullet$ Given two positive Hermitian forms $\lng.|.\rng$ and
$\lng.|.\rng'$ on a projective right $A$-module $P$, there exists
an invertible $A$-linear endomorphism of $P$ such that
\be
\lng p|p'\rng'= \lng Tp|Tp'\rng, \qquad p,p'\in P.
\ee

\section{Non-commutative differential calculus}

The notion of a differential calculus in Section 1.4 has been
formulated for any $\cK$-ring $\cA$. One can generalize the
Chevalley--Eilenberg differential calculus over a commutative ring
in Section 1.4 to a non-commutative $\cK$-ring $\cA$
\cite{dub88,dub01,book05}. For this purpose, let us consider
derivations $u\in\gd\cA$ of $\cA$. They obey the Leibniz rule
\mar{ws100'}\beq
u(ab)=u(a)b+au(b), \qquad a,b\in\cA, \label{ws100'}
\eeq
(see Remark \ref{w70}).
 By virtue of the relation (\ref{ws100'}),
the set of derivations $\gd\cA$ is both a $\cZ_\cA$-bimodule and a
Lie $\cK$-algebra with respect to the Lie bracket
\mar{ws101}\beq
[u,u']=uu'-u'u. \label{ws101}
\eeq
It is readily observed that derivations preserve the center
$\cZ_\cA$ of $\cA$.

\begin{rem} \label{w161} \mar{w161}
If $\cA$ is an involutive ring, the differential calculus over
$\cA$ fulfills the additional relations
\mar{1010}\beq
(\al\cdot\bt)^*=\bt^*\cdot\al^*, \qquad (\dl\al)^*=-\dl\al^*,
\qquad \al,\bt \in \Om^*. \label{1010}
\eeq
In particular, the second relation (\ref{1010}) shows that $\dl$
is an antisymmetric derivation of an involutive ring $\cA$.
\end{rem}

Let us consider the extended Chevalley--Eilenberg complex
(\ref{ws102}) of the Lie algebra $\gd\cA$ with coefficients in the
ring $\cA$, regarded as a $\gd\cA$-module. This complex contains a
subcomplex $\cO^*[\gd\cA]$ of $\cZ_\cA$-multilinear skew-symmetric
maps (\ref{+840'}) with respect to the Chevalley--Eilenberg
coboundary operator $d$ (\ref{+840}). Its terms $\cO^k[\gd\cA]$
are $\cA$-bimodules. The graded module $\cO^*[\gd\cA]$ is provided
with the product (\ref{ws103}) which obeys the relation
(\ref{ws98}) and makes $\cO^*[\gd\cA]$ into a  differential graded
algebra. Let us note that, if $\cA$ is not commutative, there is
nothing like the graded commutativity of forms (\ref{ws99}) in
general. Since
\mar{708'}\beq
\cO^1[\gd\cA]=\hm_{\cZ_\cA}(\gd\cA,\cA), \label{708'}
\eeq
we have the following non-commutative generalizations of the
interior product
\be
(u\rfloor\f)(u_1,\ldots,u_{k-1})= k\f(u,u_1,\ldots,u_{k-1}),
\qquad u\in\gd\cA, \qquad \f\in \cO^*[\gd\cA],
\ee
and the Lie derivative
\be
\bL_u(\f)=d(u\rfloor\f) +u\rfloor f(\f).
\ee
Then one can think of elements of $\cO^1[\gd\cA]$ as being the
non-commutative generalization of exterior one-forms.

The minimal Chevalley--Eilenberg differential calculus $\cO^*\cA$
over $\cA$ consists of the monomials
\be
a_0 da_1\w\cdots \w da_k, \qquad a_i\in\cA,
\ee
whose product $\w$ (\ref{ws103}) obeys the juxtaposition rule
\be
(a_0d a_1)\w (b_0d b_1)=a_0d (a_1b_0)\w d b_1- a_0a_1d b_0\w d
b_1, \qquad a_i,b_i\in\cA.
\ee
For instance, it follows from the product (\ref{ws103}) that, if
$a,a'\in\cZ_\cA$, then
\mar{w123}\beq
da\w da'=-da'\w da, \qquad ada'=(da')a. \label{w123}
\eeq

\begin{prop} \label{w130} \mar{w130}
There is the duality relation
\mar{ws130}\beq
\gd\cA=\hm_{\cA-\cA}(\cO^1\cA,\cA),\label{ws130}
\eeq
generalizing the relation (\ref{5.81}) to non-commutative rings.
\end{prop}

\begin{proof}
It follows from the definition (\ref{+840}) of the
Chevalley--Eilenberg coboundary operator that
\mar{spr708'}\beq
(d a)(u)=u(a), \qquad a\in\cA, \qquad u\in\gd\cA. \label{spr708'}
\eeq
This equality yields the morphism
\be
\gd\cA\ni u\mapsto \f_u\in \hm_{\cA-\cA}(\cO^1\cA,\cA), \qquad
\f_u(da)=u(a), \qquad a\in \cA.
\ee
This morphism is a monomorphism  because the module $\cO^1\cA$ is
generated by elements $da$, $a\in\cA$. At the same time, any
element $\f \in \hm_{\cA-\cA}(\cO^1\cA,\cA)$ induces the
derivation $u_\f(a)=\f(da)$ of $\cA$. Thus, there is a morphism
\be
\hm_{\cA-\cA}(\cO^1\cA,\cA)\to\gd \cA,
\ee
which is a monomorphism since $\cO^1\cA$ is generated by elements
$da$, $a\in\cA$.
\end{proof}

Let us turn now to a different differential calculus over a
non-commutative ring which is often used in non-commutative
geometry \cite{conn,land}. Let $\cA$ be a (non-commutative)
$\cK$-ring over a commutative ring $\cK$. Let us consider the
tensor product $\cA\op\ot_\cK\cA$ of $\cK$-modules. It is brought
into an $\cA$-bimodule with respect to the multiplication
\be
b(a\ot a')c=(ba)\ot (a'c), \qquad a,a',b,c\in\cA.
\ee
Let us consider its submodule $\Om^1(\cA)$ generated by the
elements
\be
\bb\ot a-a\ot\bb, \qquad a\in\cA.
\ee
It is readily observed that
\mar{w110}\beq
d:\cA\ni a\mapsto \bb\ot a -a\ot\bb\in\Om^1(\cA) \label{w110}
\eeq
is a $\Om^1(\cA)$-valued derivation of $\cA$. Thus, $\Om^1(\cA)$
is an $\cA$-bimodule generated by the elements $da$, $a\in\cA$,
such that the relation
\mar{w265}\beq
(da)b=d(ab)-adb, \qquad a,b\in \cA, \label{w265}
\eeq
holds. Let us consider the tensor algebra $\Om^*(\cA)$ of the
$\cA$-bimodule $\Om^1(\cA)$. It consists of the monomials
\mar{w1005}\beq
a_0da_1\cdots da_k, \qquad a_i\in\cA, \label{w1005}
\eeq
whose product obeys the juxtaposition rule
\be
(a_0d a_1)(b_0d b_1)=a_0d (a_1b_0)d b_1- a_0a_1d b_0 b_1, \qquad
a_i,b_i\in \cA,
\ee
because of the relation (\ref{w265}). The operator $d$
(\ref{w110}) is extended to $\Om^*(\cA)$ by the law
\mar{w169}\beq
d(a_0da_1\cdots da_k)=da_0da_1\cdots da_k, \label{w169}
\eeq
that makes $\Om^*(\cA)$ into a differential graded algebra. Its de
Rham cohomology groups are
\be
H^0(\Om^*(\cA))=\cK, \qquad H^{r>0}(\Om^*(\cA))=0.
\ee

If $\cA$ is not a unital algebra, one can consider its unital
extension $\wt\cA$ in Remark \ref{w120}, and then construct the
differential graded algebra $\Om^*(\wt\cA)$. This algebra contains
the differential graded subalgebra $\Om^*(\cA)$ of monomials
(\ref{w1005}). The de Rham cohomology groups of $\Om^*(\cA)$ are
trivial.

Of course, $\Om^*(\cA)$ is a minimal differential calculus. One
calls it the {\it universal differential calculus}
\index{differential calculus universal} over $\cA$ because of the
following property \cite{land}. Let $P$ be an $\cA$-bimodule. Any
$P$-valued derivation $\Delta$ of $\cA$ factorizes as
$\Delta={\got f}^\Delta\circ d$ through some $(\cA-\cA)$-module
homomorphism
\mar{w133}\beq
{\got f}^\Delta:\Om^1(\cA)\to P.\label{w133}
\eeq
Moreover, let $\cA'$ be another $\cK$-algebra and $(\Om'^*,\dl')$
its differential calculus over a $\cK$-ring $\cA'$. Any
homomorphism $\cA\to \cA'$ is uniquely extended to a morphism of
differential graded algebras
\be
\rho^*: \Om^*(\cA) \to \Om'^*
\ee
such that $\rho^{k+1}\circ d=\dl' \circ\rho^k$. Indeed, this
morphism factorizes through the morphism of $\Om^*(\cA)$ to the
minimal differential calculus in $\Om'^*$ which sends $da\to
\dl'\rho(a)$.

Elements of the universal differential calculus $\Om^*(\cA)$ are
called {\it universal forms}. \index{universal form} However, they
can not be regarded as the non-commutative generalization of
exterior forms because, in contrast with the Chevalley--Eilenberg
differential calculus, the monomials $da$, $a\in\cZ_\cA$, of the
universal differential calculus do not satisfy the relations
(\ref{w123}). In particular, if $\cA$ is a commutative ring, the
module $\cO^1$ (\ref{mos058}) of exterior one-forms over $\cA$ is
the quotient of the module $\Om^1(\cA)$ (\ref{w110}) of universal
forms by the relations (\ref{5.53}). At the same time, if $P=\cA$,
the morphism (\ref{w133}) takes the form
\be
{\got f}^\Delta(da)=\Delta(a).
\ee
This relation defines the monomorphism of $\Om^1(\cA)$ to
$\cO^1[\gd\cA]$ (\ref{708'}) by the formula (\ref{spr708'}).
Therefore, its
 range coincides with
the term $\cO^1\cA$ of the minimal Chevalley--Eilenberg
differential calculus, i.e., there is an isomorphism
\mar{w134}\beq
\Om^1(\cA)=\cO^1\cA. \label{w134}
\eeq

\section{Differential operators in non-commutative geometry}

It seems natural to regard derivations of a non-commutative
$\cK$-ring $\cA$ and the Chevalley--Eilenberg coboundary operator
$d$ (\ref{+840}) as particular differential operators in
non-commutative geometry. Definition \ref{ws131} provides a
standard notion of differential operators on modules over a
commutative ring. However, there exist its different
generalizations to modules over a non-commutative ring
\cite{bor97,dublmp,dub01,lunts}.

Let $P$ and $Q$ be $\cA$-bimodules over a non-commutative
$\cK$-ring $\cA$. The $\cK$-module $\hm_\cK(P,Q)$ of $\cK$-linear
homomorphisms $\Phi:P\to Q$ can be provided with the left $\cA$-
and $\cA^\bll$-module structures (\ref{5.29}) and the similar
right module structures
\mar{ws105}\beq
(\Phi a)(p)=\Phi(p)a, \qquad (a\bll\Phi)(p)=\Phi(pa), \quad
a\in\cA, \qquad p\in\ P. \label{ws105}
\eeq
For the sake of convenience, we will refer to the module
structures (\ref{5.29}) and (\ref{ws105}) as the left and right
$\cA-\cA^\bll$ structures, respectively. Let us put
\mar{ws133}\beq
\ol\dl_a\Phi=\Phi a-a\bll\Phi, \qquad a\in\cA, \qquad \Phi\in
\hm_\cK(P,Q). \label{ws133}
\eeq
It is readily observed that
\be
\dl_a\circ\ol\dl_b=\ol\dl_b\circ\dl_a, \qquad a,b\in\cA.
\ee

The left $\cA$-module homomorphisms $\Delta: P\to Q$ obey the
conditions $\dl_a\Delta=0$, for all $a\in\cA$ and, consequently,
they can be regarded as left zero order $Q$-valued differential
operators on $P$. Similarly, right zero order differential
operators are defined.

Utilizing the condition (\ref{ws106}) as a definition of a first
order differential operator in non-commutative geometry, one
however meets difficulties. If $P=\cA$ and $\Delta(\bb)=0$, the
condition (\ref{ws106}) does not lead to the Leibniz rule
(\ref{+a20}), i.e., derivations of the $\cK$-ring $\cA$ are not
first order differential operators. In order to overcome these
difficulties, one can replace the condition (\ref{ws106}) with the
following one \cite{dublmp}.

\begin{defi} \label{ws120} \mar{ws120}
An element $\Delta\in \hm_\cK(P,Q)$ is called a first order {\it
differential operator} \index{differential operator!in
non-commutative geometry} of a  bimodule $P$ over a
non-commutative ring $\cA$ if it obeys the condition
\mar{ws114}\ben
&& \dl_a\circ\ol\dl_b\Delta=\ol\dl_b\circ\dl_a\Delta=0,
\qquad  a,b\in\cA, \nonumber \\
&& a\Delta(p)b -a\Delta(pb) -\Delta(ap)b +\Delta(apb)=0, \qquad p\in P.
\label{ws114}
\een
\end{defi}

First order $Q$-valued differential operators on $P$ make up a
$\cZ_\cA$-module $\dif_1(P,Q)$.

If $P$ is a commutative bimodule over a commutative ring $\cA$,
then $\dl_a=\ol\dl_a$ and Definition \ref{ws120} comes to
Definition \ref{ws131} for first order differential operators.

In particular, let $P=\cA$. Any left or right zero order
$Q$-valued differential operator $\Delta$ is uniquely defined by
its value $\Delta(\bb)$. As a consequence, there are left and
right $\cA$-module isomorphisms
\be
&& Q\ni q\mapsto \Delta^{\rm R}_q\in\dif_0^{\rm R}(\cA,Q), \qquad
\Delta^{\rm R}_q(a)=qa, \qquad a\in\cA,\\
&& Q\ni q\mapsto \Delta^{\rm L}_q\in\dif_0^{\rm L}(\cA,Q), \qquad
\Delta^{\rm L}_q(a)=aq.
\ee
A first order $Q$-valued differential operator $\Delta$ on $\cA$
fulfils the condition
\mar{ws110}\beq
\Delta(ab)=\Delta(a)b+a\Delta(b) -a\Delta(\bb)b. \label{ws110}
\eeq
It is a derivation of $\cA$ if $\Delta(\bb)=0$. One obtains at
once that any first order differential operator on $\cA$ is split
into the sums
\be
&& \Delta(a)=a\Delta(\bb) +[\Delta(a)-a\Delta(\bb)], \\
&& \Delta(a)=\Delta(\bb)a +[\Delta(a)-\Delta(\bb)a]
\ee
of the derivations $\Delta(a)-a\Delta(\bb)$ or
$\Delta(a)-\Delta(\bb)a$ and the left or right zero order
differential operators $a\Delta(\bb)$ and $\Delta(\bb)a$,
respectively. If $u$ is a $Q$-valued derivation of $\cA$, then
$au$ (\ref{5.29}) and $ua$ (\ref{ws105}) are so for any
$a\in\cZ_\cA$. Hence, $Q$-valued derivations of $\cA$ constitute a
$\cZ_\cA$-module $\gd(\cA,Q)$. There are two $\cZ_\cA$-module
decompositions
\be
&& \dif_1(\cA,Q)= \dif_0^{\rm L}(\cA,Q) \oplus \gd(\cA,Q), \\
&& \dif_1(\cA,Q)= \dif_0^{\rm R}(\cA,Q) \oplus \gd(\cA,Q).
\ee
They differ from each other in the inner derivations $a\mapsto
aq-qa$.

Let $\hm_\cA^{\rm R}(P,Q)$ and $\hm_\cA^{\rm L}(P,Q)$ be the
modules of right and left $\cA$-module homomorphisms of $P$ to
$Q$, respectively. They are provided with the left and right
$\cA-\cA^\bll$-module structures (\ref{5.29}) and (\ref{ws105}),
respectively.

\begin{prop} \label{ws113} \mar{ws113}
An element $\Delta\in\hm_\cK(P,Q)$ is a first order $Q$-valued
differential operator on $P$ in accordance with Definition
\ref{ws120} iff it obeys the condition
\mar{n21}\beq
\Delta(apb)=(\op\dr^\to a)(p)b +a\Delta(p)b + a(\op\dr^\leftarrow
b)(p), \qquad p\in P, \quad  a,b\in\cA,\label{n21}
\eeq
where $\op\dr^\to$ and $\op\dr^\leftarrow$ are $\hm_\cA^{\rm
R}(P,Q)$- and $\hm_\cA^{\rm L}(P,Q)$-valued derivations of $\cA$,
respectively \cite{book05}. Namely,
\be
(\op\dr^\to a)(pb)=(\op\dr^\to a)(p)b, \qquad (\op\dr^\leftarrow
b)(ap) =a(\op\dr^\leftarrow b)(p).
\ee
\end{prop}

For instance, let $P$ be a differential calculus over a $\cK$-ring
$\cA$ provided with an associative multiplication $\circ$ and a
coboundary operator $d$. Then $d$ exemplifies a $P$-valued first
order differential operator on $P$ by Definition \ref{ws120}. It
obeys the condition (\ref{n21}) which reads
\be
d(apb)=(da\circ p)b+a(dp)b + a((-1)^{|p|}p\circ db).
\ee
For instance, let $P=\cO^*\cA$ be the Chevalley--Eilenberg
differential calculus over $\cA$. In view of the relations
(\ref{708'}) and (\ref{ws130}), one can think of derivations
$u\in\gd\cA$ as being vector fields in non-commutative geometry. A
problem is that $\gd\cA$ is not an $\cA$-module. One can overcome
this difficulty as follows \cite{bor97}.

Given a non-commutative $\cK$-ring $\cA$ and an $\cA$-bimodule
$Q$, let $d$ be a $Q$-valued derivation of $\cA$. One can think of
$Q$ as being a first degree term of a differential calculus over
$\cA$. Let $Q^*_{\rm R}$ be the right $\cA$-dual of $Q$. It is an
$\cA$-bimodule:
\be
(bu)(q)= bu(q), \qquad (ub)(q)=u(bq), \qquad  b\in\cA, \qquad q\in
Q.
\ee
One can associate to each element $u\in Q^*_{\rm R}$ the
$\cK$-module morphism
\mar{ws121}\beq
\wh u:\cA\in a\mapsto u(da)\in\cA. \label{ws121}
\eeq
This morphism obeys the relations
\mar{ws123}\beq
\wh{(bu)}(a) =bu(da), \qquad \wh u(ba)=\wh u(b)a+\wh{(ub)}(a).
\label{ws123}
\eeq
One calls $(Q^*_{\rm R},u\mapsto\wh u )$ the $\cA$-right {\it
Cartan pair}, \index{Cartan pair} and regards $\wh u$
(\ref{ws121}) as an $\cA$-valued first order differential operator
on $\cA$ \cite{bor97}. Let us note that $\wh u$ (\ref{ws121}) need
not be a derivation of $\cA$ and fails to satisfy Definition
\ref{ws120}, unless $u$ belongs to the two-sided $\cA$-dual
$Q^*\subset Q^*_{\rm R}$ of $Q$. Morphisms $\wh u$ (\ref{ws121})
are called into play in order to describe (left) vector fields in
non-commutative geometry \cite{bor97,jara}.

In particular, if $Q=\cO^1\cA$, then $au$ for any $u\in\gd\cA$ and
$a\in\cA$ is a left non-commutative vector field in accordance
with the relation (\ref{spr708}).

Similarly, the $\cA$-left Cartan pair is defined. For instance,
$ua$ for any $u\in\gd\cA$ and $a\in\cA$ is a right {\it
non-commutative vector field}. \index{non-commutative vector
field}

If $\cA$-valued derivations $u_1,\ldots u_r$ of a non-commutative
$\cK$-ring $\cA$ or the above mentioned non-commutative vector
fields $\wh u_1,\ldots \wh u_r$ on $\cA$ are regarded as first
order differential operators on $\cA$, it seems natural to think
of their compositions $u_1\circ\cdots u_r$ or $\wh u_1\circ\cdots
\wh u_r$ as being particular higher order differential operators
on $\cA$. Let us turn to the general notion of a differential
operator on $\cA$-bimodules.

By analogy with Definition \ref{ws131}, one may try to generalize
Definition \ref{ws120} by means of the maps $\dl_a$ (\ref{spr172})
and $\ol\dl_a$ (\ref{ws133}). A problem lies in the fact that, if
$P=Q=\cA$, the compositions $\dl_a\circ\dl_b$ and
$\ol\dl_a\circ\ol\dl_b$ do not imply the Leibniz rule and, as a
consequence, compositions of derivations of $\cA$ fail to be
differential operators \cite{book05,sard07}.

This problem can be solved if $P$ and $Q$ are regarded as left
$\cA$-modules \cite{lunts}. Let us consider the $\cK$-module
$\hm_\cK (P,Q)$ provided with the left $\cA-\cA^\bll$ module
structure (\ref{5.29}). We denote by $\cZ_0$ its center, i.e.,
$\dl_a\Phi=0$ for all $\Phi\in\cZ_0$ and $a\in\cA$. Let $\cI_0=\ol
\cZ_0$ be the $\cA-\cA^\bll$ submodule of $\hm_\cK (P,Q)$
generated by $\cZ_0$. Let us consider:

(i) the quotient $\hm_\cK (P,Q)/\cI_0$,

(ii) its center $\cZ_1$,

(iii) the $\cA-\cA^\bll$ submodule $\ol \cZ_1$ of $\hm_\cK
(P,Q)/\cI_0$ generated by $\cZ_1$,

(iv) the $\cA-\cA^\bll$ submodule $\cI_1$ of $\hm_\cK (P,Q)$ given
by the relation $\cI_1/\cI_0=\ol \cZ_1$.

\noindent Then we define the $\cA-\cA^\bll$ submodules $\cI_r$,
$r=2,\ldots$, of $\hm_\cK (P,Q)$ by induction as
$\cI_r/\cI_{r-1}=\ol \cZ_r$, where $\ol \cZ_r$ is the
$\cA-\cA^\bll$ module generated by the center $\cZ_r$ of the
quotient $\hm_\cK (P,Q)/\cI_{r-1}$.

\begin{defi} \label{ws135} \mar{ws135}
Elements of the submodule $\cI_r$ of $\hm_\cK (P,Q)$ are said to
be left $r$-order $Q$-valued {\it differential operators}
\index{differential operator!in non-commutative geometry} of an
$\cA$-bimodule $P$ \cite{lunts}.
\end{defi}

\begin{prop} \label{ws137} \mar{ws137}
An element $\Delta\in \hm_\cK (P,Q)$ is a differential operator of
order $r$ in accordance with Definition \ref{ws135} iff it is a
finite sum
\mar{ws138}\beq
\Delta(p)=b_i\Phi^i(p) +\Delta_{r-1}(p), \qquad b_i\in\cA,
\label{ws138}
\eeq
where $\Delta_{r-1}$ and $\dl_a\Phi^i$ for all $a\in\cA$ are
$(r-1)$-order differential operators if $r>0$, and they vanish if
$r=0$ \cite{book05}.
\end{prop}

If $\cA$ is a commutative ring, Definition \ref{ws135} comes to
Definition \ref{ws131}. Indeed, the expression (\ref{ws138}) shows
that $\Delta\in \hm_\cK (P,Q)$ is an $r$-order differential
operator iff $\dl_a\Delta$ for all $a\in\cA$ is a differential
operator of order $r-1$.

\begin{prop} \label{ws143} \mar{ws143}
If $P$ and $Q$ are $\cA$-bimodules, the set $\cI_r$ of $r$-order
$Q$-valued differential operators on $P$ is provided with the left
and right $\cA-\cA^\bll$ module structures \cite{book05}.
\end{prop}

Let $P=Q=\cA$. Any zero order differential operator on $\cA$ in
accordance with Definition \ref{ws135} takes the form $a\mapsto
cac'$ for some $c,c'\in\cA$.

\begin{prop} \label{ws146} \mar{ws146}
Let $\Delta_1$ and $\Delta_2$ be $n$- and $m$-order $\cA$-valued
differential operators on $\cA$, respectively. Then their
composition $\Delta_1\circ\Delta_2$ is an $(n+m)$-order
differential operator \cite{book05}.
\end{prop}

Any derivation $u\in\gd\cA$ of a $\cK$-ring $\cA$ is a first order
differential operator in accordance with Definition \ref{ws135}.
Indeed, it is readily observed that
\be
(\dl_au)(b)= au(b)-u(ab)=-u(a)b, \qquad b\in\cA,
\ee
is a zero order differential operator for all $a\in\cA$. The
compositions $au$, $u\bll a$ (\ref{5.29}), $ua$, $a\bll u$
(\ref{ws105}) for any $u\in\gd\cA$, $a\in\cA$ and the compositions
of derivations $u_1\circ\cdots\circ u_r$ are also differential
operators on $\cA$ in accordance with Definition \ref{ws135}.

At the same time, non-commutative vector fields do not satisfy
Definition \ref{ws135} in general. First order differential
operators by Definition \ref{ws120} also need not obey Definition
\ref{ws135}, unless $P=Q=\cA$.

By analogy with Definition \ref{ws135} and Proposition
\ref{ws137}, one can define differential operators on right
$\cA$-modules as follows.

\begin{defi} \label{ws151} \mar{ws151}
Let $P$ and $Q$ be seen as right $\cA$-modules over a
non-commutative $\cK$-ring $\cA$. An element
$\Delta\in\hm_\cK(P,Q)$ is said to be a right zero order
$Q$-valued differential operator on $P$ if it is a finite sum
$\Delta=\Phi^i b_i$, $b_i\in\cA$, where $\ol\dl_a\Phi^i=0$ for all
$a\in\cA$. An element $\Delta\in\hm_\cK(P,Q)$ is called a right
differential operator of order $r>0$ on $P$ if it is a finite sum
\mar{ws150}\beq
\Delta(p)=\Phi^i(p)b_i +\Delta_{r-1}(p), \qquad b_i\in\cA,
\label{ws150}
\eeq
where $\Delta_{r-1}$ and $\ol\dl_a\Phi^i$ for all $a\in\cA$ are
right $(r-1)$-order differential operators.
\end{defi}

Definition \ref{ws135} and Definition \ref{ws151} of left and
right differential operators on $\cA$-bimodules are not
equivalent, but one can combine them as follows.

\begin{defi} \label{ws152} \mar{ws152}
Let $P$ and $Q$ be bimodules over a non-commutative $\cK$-ring
$\cA$. An element $\Delta\in\hm_\cK(P,Q)$ is a two-sided zero
order $Q$-valued differential operator on $P$ if it is either a
left or right zero order differential operator. An element
$\Delta\in\hm_\cK(P,Q)$ is said to be a two-sided differential
operator of order $r>0$ on $P$ if it is brought both into the form
\be
\Delta=b_i\Phi^i +\Delta_{r-1},\qquad b_i\in\cA,
\ee
and
\be
\Delta=\ol\Phi^i\ol b_i +\ol\Delta_{r-1}, \qquad \ol b_i\in\cA,
\ee
where $\Delta_{r-1}$, $\ol\Delta_{r-1}$ and $\dl_a\Phi^i$,
$\ol\dl_a\ol\Phi^i$ for all $a\in\cA$
  are two-sided $(r-1)$-order differential operators.
\end{defi}

One can think of this definition as a generalization of Definition
\ref{ws120} to higher order differential operators.

It is readily observed that two-sided differential operators
described by Definition \ref{ws152} need not be left or right
differential operators, and {\it vice versa}. At the same time,
$\cA$-valued derivations of a $\cK$-ring $\cA$ and their
compositions obey Definition \ref{ws152}.

\section{Connections in non-commutative geometry}

This Section is devoted to the definitions of a connection in
non-commutative geometry. We follow the notion of an algebraic
connection in Section 1.3 generalized to modules over
non-commutative rings \cite{conn80,conn,dub,dub2,dub01}.

Let $(\Om^*,\dl)$ be a differential calculus over a $\cK$-ring
$\cA$, and let $P$ be a left $\cA$-module. Following  Definition
\ref{+181}, one can construct the tensor product of modules
$\Om^1\ot P$  and then define a {\it left connection}
\index{connection!left} on $P$ as a $\cK$-module morphism
\mar{+866}\beq
\nabla^L: P\to \Om^1\op\ot_\cA P, \label{+866}
\eeq
which obeys the Leibniz rule
\be
\nabla^L(ap)=\dl a\ot p +a\nabla^L(p), \qquad p\in P, \qquad a\in
\cA,
\ee
\cite{land,var}. If $\Om^*=\Om^*(\cA)$ is the universal
differential calculus over $\cA$, the connection (\ref{+866}) is
called a {\it universal connection}.
\index{connection!non-commutative!universal}

Similarly, a {\it right non-commutative connection}
\index{connection!right} on a right $\cA$-module $P$ is defined as
a $\cK$-module morphism
\be
\nabla^R: P\to P\op\ot_\cA\Om^1,
\ee
which obeys the Leibniz rule
\be
\nabla^R(pa)=p\ot\dl a +\nabla^R(p) a, \qquad p\in P, \qquad a\in
\cA.
\ee

The forthcoming theorem shows that a connection on a left (or
right) module over a non-commutative ring need not exist
\cite{cunt,land}.

\begin{theo} \label{+893} \mar{+893}
A left (resp. right) universal connection on a left (resp. right)
module $P$ of finite rank exists iff $P$ is projective.
\end{theo}

In contrast with connections on left and right modules, there is a
problem  how to define a connection on an $\cA$-bimodule over a
non-commutative ring. If $\cA$ is a commutative ring, then the
tensor products $\Om^1\ot P$ and $P\ot\Om^1$ are naturally
identified by the permutation morphism
\be
\varrho: \al\ot p\mapsto p\ot\al, \qquad  \al\in\Om^1, \quad
  p\in
P,
\ee
so that any left connection $\nabla^L$ is also the right one
$\varrho\circ \nabla^L$, and {\it vice versa}. If $\cA$ is a
non-commutative ring,  neither left nor right connections  are
connections on an $\cA$-bimodule $P$  since $\nabla^L(P)\in
\Om^1\ot P$, whereas $\nabla^R(P)\in P\ot\Om^1$. As a palliative,
one may assume that there exists some $\cA$-bimodule isomorphism
\mar{+894}\beq
\varrho: \Om^1\ot P\to P\ot\Om^1. \label{+894}
\eeq
Then a couple $(\nabla^L,\nabla^R)$ of right and left
non-commutative connections on $P$ is said to be a
$\varrho$-compatible if $\varrho\circ\nabla^L= \nabla^R$
\cite{dub2,land,mour} (see also \cite{dabr96} for a weaker
condition). However, this couple is not a true connection on an
$\cA$-bimodule. The problem is not simplified even if $P=\Om^1$,
together with the natural permutation
\be
\f\ot\f'\mapsto \f'\ot\f, \qquad \f,\f'\in \Om^1.
\ee

A different construction of a connection on $\cA$-bimodules has
been suggested in \cite{dub}. Let $(\Om^*\cA,\dl)$ be the minimal
Chevalley--Eilenberg differential calculus over a $\cK$-ring
$\cA$. Due to the isomorphism (\ref{ws130}), any element
$u\in\gd\cA$ determines the morphisms
\mar{1013,4}\ben
&& u: \cO^1\cA\ot P\ni \al\ot p\mapsto u(\al)p\in P, \label{1013}\\
&& u: P\ot\cO^1\cA\ni \p\ot\al\mapsto pu(\al)\in P. \label{1014}
\een

\begin{defi} \label{+845} \mar{+845}
A connection on an $\cA$-bimodule $P$ with respect to the minimal
Chevalley--Eilenberg differential calculus $\cO^*\cA$ over $\cA$
is defined as a $\cZ_\cA$-bimodule morphism
\mar{+846}\beq
\nabla: \gd\cA\ni u\mapsto \nabla_u\in \hm_\cK (P,P), \label{+846}
\eeq
which obeys the Leibniz rule
\mar{+847}\beq
\nabla_u(a_ipb_j)=u(a_i)pb_j +a_i\nabla_u(p)b_j +a_ipu(b_j), \quad
   p\in P, \quad
   a_k,b_k\in A_k. \label{+847}
\eeq
\end{defi}

Comparing the formulae (\ref{n21}) and (\ref{+847}) shows that, by
virtue of Proposition \ref{ws113}, $\nabla_u$ (\ref{+847}) is a
first order differential operator on $P$ in accordance with
Definitions \ref{ws120} and \ref{ws151}.

Let us recall that, if $a\in \cZ_\cA$, then $\dl a$ belongs to the
center $\cZ_P$ of the module $P$ and $u(\dl a)\in \cZ_\cA$ for all
$u\in\gd\cA$. If $\cA$ is a commutative ring, Definition
\ref{+845} is identic to Definition \ref{1016}.

\begin{rem}
Due to the isomorphism (\ref{w134}), Definition \ref{+845} can be
also extended to connections with respect to the universal
differential calculus.
\end{rem}

\begin{rem}
Definition \ref{+845} can be applied to the case of a non-minimal
differential calculus $\Om^*$ over $\cA$ by replacing the
derivation module $\gd\cA$ with the dual $(\Om^1)^*$ of the
$\cA$-bimodule $\Om^1$. Definition \ref{+845} is straightforwardly
applied to the Chevalley--Eilenberg differential calculus
$\cO^*[\gd\cA]$ over $\cA$ due to the inclusion
\be
\gd\cA\subset \gd\cA^{**}=(\cO^*[\gd\cA])^*.
\ee
\end{rem}

We agree to call a connection in Definition \ref{+845} the {\it
Dubois--Violette connection}. \index{Dubois--Violette connection}
It is readily observed that any left (resp. right) connection
$\nabla$ on an $\cA$-module of type $(1,0)$ (resp. $(0,1)$)
determines the Dubois--Violette connection
$\nabla_u=u\circ\nabla$, $u\in\gd\cA$, on $P$ due to the morphism
(\ref{1013}) (resp. (\ref{1014})).

\begin{ex}
If $P=\cA$, the morphisms
\mar{+870}\beq
\nabla_u(a)=u(a), \quad  u\in \gd\cA, \quad  a\in \cA,
\label{+870}
\eeq
define the canonical Dubois--Violette connection on the ring
$\cA$. Then, due to the Leibniz rule (\ref{+847}), a
Dubois--Violette connection on any $\cA$-bimodule $P$ is also a
Dubois--Violette connection on $P$, seen as a $\cZ_\cA$-bimodule.
\end{ex}

\begin{ex}
If $P$ is an $\cA$-bimodule and the ring $\cA$ has only inner
derivations
\be
{\rm ad}\, b(a)= ba-ab,
\ee
the morphisms
\mar{+871}\beq
\nabla_{{\rm ad}\, b}(p)= bp-pb, \quad  b\in\cA, \quad p\in P,
\label{+871}
\eeq
define the canonical Dubois--Violette connection on $P$.
\end{ex}

The  curvature $R$ of a Dubois--Violette connection $\nabla$
(\ref{+846}) on an $\cA$-module $P$ is defined as the
$\cZ_\cA$-module morphism
\mar{+874}\ben
&& R:\gd\cA\times\gd\cA \ni(u,u')\mapsto R_{u,u'}\in
\hm_{A_i-A_j}(P,P), \label{+874}\\
&& R_{u,u'}(p)=\nabla_u(\nabla_{u'}(p)) -\nabla_{u'}(\nabla_u(p))
-\nabla_{[u,u']}(p), \quad p\in P,\nonumber
\een
\cite{dub}. We have the relations
\be
&& R_{au,a'u'}=aa' R_{u,u'}, \qquad a,a'\in\cZ_\cA,\\
&& R_{u,u'}(a_ipb_j)= a_iR_{u,u'}(p)b_j, \quad a_k,b_k\in
A_k.
\ee
For instance, the curvature of the connections (\ref{+870}) and
(\ref{+871}) vanishes.

There are the following standard constructions of new
Dubois--Violette connections from old ones.

(i) Given $\cA$-bimodules $P$ and $P'$ and two Dubois--Violette
connections $\nabla$ and $\nabla'$ on them, there is an obvious
Dubois--Violette connection $\nabla\oplus\nabla'$ on the direct
sum $P\oplus P'$.

(ii) Let $P$ be an $\cA$-bimodule and $P^*$ its $\cA$-dual. For
any Dubois--Violette connection $\nabla$ on $P$, there is a unique
{\it dual Dubois--Violette connection} \index{Dubois--Violette
connection!dual} $\nabla'$ on $P^*$ such that
\be
u(\lng p,p'\rng)=\lng \nabla_u(p),p'\rng +\lng p,\nabla'(p')\rng,
\quad
   p\in P,\quad\ p'\in P^*, \quad u\in\gd\cA.
\ee

(iii) Let $P$ and $P'$ be $\cA$-bimodules, and let $\nabla$ and
$\nabla'$ be Dubois--Violette connections on these modules. For
any $u\in\gd\cA$, let us consider the endomorphism
\mar{+848}\beq
(\nabla\ot\nabla')_u= \nabla_u\ot \id P' + \id P\ot\nabla'_u
\label{+848}
\eeq
of the tensor product $P\ot P'$ of $\cK$-modules $P$ and $P'$.
This endomorphism preserves the subset of $P\ot P'$ generated by
elements
\be
pa\ot p'-p\ot ap', \qquad p\in P, \quad p'\in P' \quad a\in A_k.
\ee
Consequently, the endomorphisms (\ref{+848}) define a
Dubois--Violette connection on the tensor product $P\ot P'$ of
modules $P$ and $P'$.

(iv) Let $\cA$ be a unital $*$-algebra and $P$ an involutive
module over $\cA$. Let us recall that, in this case, the
derivation module $\gd\cA$ consists of only symmetric derivations
of $\cA$. For any Dubois--Violette connection $\nabla$ on $P$, the
conjugate Dubois--Violette connection $\nabla^*$ on $P$ is defined
by the relation
\mar{+849}\beq
\nabla_u^*(p)=(\nabla_u(p^*))^*. \label{+849}
\eeq
A Dubois--Violette connection $\nabla$ on an involutive module $P$
is said to be {\it real} \index{Dubois--Violette connection!real}
if $\nabla=\nabla^*$.

Let now $P=\cO^1\cA$. Any Dubois--Violette connection on an
$\cA$-bimodule $\cO^1\cA$ is called a
   {\it linear connection} \index{Dubois--Violette connection!linear}
\cite{dub}. Let us note that this is not the term for a left or
right connection on $\cO^1\cA$ \cite{dub2}. If $\cO^1\cA$ is an
involutive module, a linear connection on it is assumed to be
real. Given a linear connection $\nabla$ on $\cO^1\cA$, there is
an $\cA$-module homomorphism
\mar{+873}\ben
&& T:\cO^1\cA\to \cO^2\cA, \nonumber\\
&& (T\f)(u,u')=(d\f)(u,u') - \nabla_u(\f)(u') +\nabla_{u'}(\f)(u),
\label{+873}
\een
called the {\it torsion} \index{torsion!of a non-commutative
connection} of the linear connection $\nabla$.

\section{Matrix geometry}

{\it Matrix geometry} \index{matrix geometry} over the algebra
$\cA=M_n$ of complex $n\times n$ matrices provides an important
example of a non-commutative system of finite degrees of freedom
\cite{dub90,mad}. Let  $\{\ve_r\}$, $1\leq r\leq n^2-1$, be an
anti-Hermitian basis for the (right) Lie algebra $su(n)$. All
derivations of the algebra $M_n$ are inner, and $u_r={\rm
ad}\,\ve_r$ constitute a basis for the complex Lie algebra $\gd
M_n$ of derivations of $M_n$, together with the commutation
relations
\be
[u_r,u_q]=c^s_{rq}u_s,
\ee
where $c^s_{rq}$ are structure constants of the Lie algebra
$su(n)$. Since the center $\cZ_{M_n}$ of $M_n$ consists of
matrices $c\bb$, $c\in \Bbb C$, the derivation module $\gd M_n$ is
an $(n^2-1)$-dimensional complex vector space. Let us consider the
minimal Chevalley--Eilenberg
   differential calculus $(\cO^*M_n,d)$ over
the algebra $M_n$ with respect to the Chevalley--Eilenberg
coboundary operator $d$ (\ref{+840}). In particular,
$\cO^0M_n=M_n$, while $\cO^1M_n$ is a free left $M_n$-module of
rank $n^2-1$ whose basis $\{\thh^r\}$ is the dual of the basis
$\{u_r\}$ for the complex Lie algebra $\gd M_n$, i.e.,
\be
\thh^r(u_q)= \dl^r_q\bb.
\ee
It is readily observed that elements $\thh^r$ of this basis belong
to the center of the $M_n$-bimodule $\cO^1M_n$, i.e.,
\mar{+930}\beq
a\thh^r=\thh^r a, \qquad  a\in M_n. \label{+930}
\eeq
It also follows that
\mar{+931}\beq
\thh^r\w\thh^q=-\thh^q\w\thh^r. \label{+931}
\eeq
The morphism
\be
d:M_n\to \cO^1M_n,
\ee
given by the formula (\ref{spr708}), reads
\be
d\ve_r(u_q)={\rm ad}\,\ve_q(\ve_r)=c^s_{qr}\ve_s,
\ee
that is,
\mar{+922}\beq
d\ve_r=c^s_{qr}\ve_s\thh^q. \label{+922}
\eeq
The formula (\ref{+921}) leads to the Maurer--Cartan equations
\mar{+934}\beq
d\thh^r=-\frac12c^r_{qs}\thh^q\w\thh^s. \label{+934}
\eeq
If we put $\thh=\ve_r\thh^r$, the equality (\ref{+922}) can be
brought into the form
\be
da=a\thh-\thh a, \qquad  a\in M_n.
\ee
It follows that the $M_n$-bimodule $\cO^1M_n$ is generated by only
one element $\thh$.

Let us consider linear connections on the $M_n$-bimodule
$\cO^1M_n$ in matrix geometry. Such a connection $\nabla$ is given
by the relations
\mar{+932}\ben
&& \nabla_{u=c^ru_r}=c^r\nabla_r, \nonumber \\
&& \nabla_r(\thh^p)=\om^p_{rq}\thh^q, \qquad \om^p_{rq}\in M_n.
\label{+932}
\een
Bearing in mind the equalities (\ref{+930}) -- (\ref{+931}), we
obtain from the Leibniz rule (\ref{+847}) that
\be
a\nabla_r(\thh^p)=\nabla_r(\thh^p)a, \qquad  a\in M_n.
\ee
It follows that elements $\om^p_{rq}$ in the expression
(\ref{+932}) are proportional to the identity $\bb\in M_n$, i.e.,
are complex numbers. Then the relations
\mar{+933}\beq
\nabla_r(\thh^p)=\om^p_{rq}\thh^q, \qquad \om^p_{rq}\in \Bbb C,
\label{+933}
\eeq
determine any linear connection on $\cO^1M_n$. Let us point out
the following two particular linear connections on $\cO^1M_n$.

(i) Since all derivations of the algebra $M_n$ are inner, we have
the curvature-free connection (\ref{+871}) given by the relations
$\nabla_r(\thh^p)=0$. However, this connection is not
torsion-free. The expressions (\ref{+934}) and (\ref{+873}) result
in
\be
(T\thh^p)(u_r,u_q)=-c^p_{rq}.
\ee

(ii) In matrix geometry, there is a unique torsion-free linear
connection
\be
\nabla_r(\thh^p)=-c^p_{rq}\thh^q.
\ee

\section{Connes' non-commutative geometry}

Connes' non-commutative geometry addresses the representation of
the universal differential calculus $(\Om^*(\cA),d)$ over an
involutive algebra $\cA$ in a Hilbert space $E$ so that the
coboundary operator $da$, $a\in\cA$, is represented by the bracket
$[\cD,\pi(a)]$, $a\in A$, where $\cD$ is a certain operator in
$E$. Thus, one comes to the notion of a spectral triple
\cite{conn,land,var}.

\begin{defi} \label{w168} \mar{w168}
A {\it spectral triple} $(\cA,E,\cD)$ \index{spectral triple} is
given by a unital $*$-algebra $\cA\subset B(E)$ of bounded
operators in a separable Hilbert space $E$ and a (unbounded)
self-adjoint operator $\cD$ in $E$ such that:

(i) the brackets $[\cD,a]$, $a\in A$, are bounded operators in
$E$,

(ii) the resolvent $(\cD-\la)^{-1}$, $\la\in \Bbb C\setminus \Bbb
R$, is a compact operator in $E$,
\end{defi}

The triple $(\cA,E,\cD)$ is also called the {\it $K$-cycle}
\index{$K$-cycle} over $\cA$ \cite{conn}.

\begin{rem} \label{w171} \mar{w171}
In Connes' commutative geometry (see Example \ref{w170} below),
$\cA$ is the algebra $\Bbb C^\infty(X)$ of smooth complex
functions on a compact manifold $X$. It is a dense subalgebra of
the $C^*$-algebra of continuous complex functions on $X$.
Generalizing this case, one usually assume that an algebra $\cA$
in Connes' spectral triple is a dense involutive subalgebra of
some $C^*$-algebra.
\end{rem}

\begin{rem} \label{w172} \mar{w172}
In many cases, $E$ is a $\Bbb Z_2$-graded Hilbert space whose
grading automorphism $\g$ obeys the conditions
\be
\g \cD+\cD\g=0, \qquad [a,\g]=0, \qquad a\in\cA,
\ee
i.e., elements of $\cA$ are even operators, while $\cD$ is the odd
one.
 The spectral triple is called {\it even}
\index{spectral triple!even} if such a grading exists and {\it
odd} \index{spectral triple!odd} otherwise.
\end{rem}

Given a spectral triple $(\cA,E,\cD)$ in Definition \ref{w168},
let $(\Om^*(\cA),d)$ be the universal differential calculus over
the algebra $\cA$. Let us consider the representation of the
graded algebra $\Om^*(\cA)$ by bounded operators
\mar{+940}\beq
\pi (a_0d a_1\cdots d a_k)= a_0[\cD,a_1]\cdots [\cD,a_k]
\label{+940}
\eeq
in the Hilbert space $E$. Since
\be
[\cD,a]^*=-[\cD,a^*],
\ee
we have
\be
\pi(\om)^*=\pi(\om^*), \qquad \om\in \Om^*(\cA).
\ee
However, the representation (\ref{+940}) fails to be a
representation of the differential graded algebra $\Om^*(\cA)$
because $\pi(\om)=0$ need not imply that $\pi(d\om)=0$. Therefore,
one should construct the appropriate quotient of $\Om^*(\cA)$ in
order to obtain a differential graded algebra of operators in $E$.

Let $J_0^*$ be the graded two-sided ideal of $\Om^*(\cA)$ where
\be
J^k_0=\{\om\in\Om^k(\cA)\, :\, \pi(\om)=0\}.
\ee
Then it is readily observed that
\be
J^*=J_0^* +d J_0^*
\ee
is a differential graded two-sided ideal of $\Om^*(\cA)$. By {\it
Connes' differential calculus} \index{Connes' differential
calculus} is meant the pair $(\Om^*_\cD\cA, d)$ such that
\be
\Om^*_\cD\cA =\Om^*\cA/J^*, \qquad  d[\om]=[d \om],
\ee
where $[\om]$ denotes the class of $\om\in \Om^*(\cA)$ in
$\Om^*_\cD\cA$. It is a differential calculus over
$\Om_\cD^0\cA=\cA$. Its representation $\pi(\Om^*_\cD\cA)$  is
given by the classes $\pi[\om]$ of operators
\be
\op\sum_j a_0^j[\cD,a_1^j]\cdots [\cD,a_k^j], \qquad a_i^j\in\cA,
\ee
modulo the operators
\be
\op\sum_j [\cD,b_0^j][\cD,b_1^j]\cdots [\cD,b_{k-1}^j]
\ee
such that
\be
\op\sum_j b_0^j[\cD,b_1^j]\cdots [\cD,b_{k-1}^j]=0\}.
\ee
It should be emphasized that $\pi([\om])$ are not operators in
$E$. The coboundary operator in $\pi(\Om^*_\cD\cA)$ reads
\be
d(a_0[\cD,a_1]\cdots [\cD,a_k])=[\cD,a_0][\cD,a_1]\cdots
[\cD,a_k].
\ee

\begin{ex} \label{w170} \mar{w170}
Connes' commutative geometry is characterized by the spectral
triple $(\cA,E,\cD)$ where:

$\bullet$  $\cA=\Bbb C^\infty(X)$ is the above mentioned algebra
of smooth complex functions on a compact manifold $X$,

$\bullet$ $E=L^2(X,S)$ is the Hilbert space of square integrable
sections of a spinor bundle $S\to X$,

$\bullet$ $\cD$ is the Dirac operator on this spinor bundle.

\noindent In this case, the representation $\pi$ (\ref{+940}) of
the universal differential calculus $\Om^*(\cA)$ over $\cA$ (up to
a possible twisting by a complex line bundle) is the complex
Clifford algebra ${\rm Cliff}(T^*X)\ot\Bbb C$ of the cotangent
bundle $T^*X$ and
\be
\pi(\Om^*_\cD\cA)=\cO^*(X)\ot\Bbb C
\ee
is the algebra of complex exterior forms on $X$
\cite{conn1,frol1,ren2}.
\end{ex}

Spectral triples have been studied, e.g., for non-commutative
tori, the Moyal deformations of $\Bbb R^n$, non-commutative
spheres $2$-, $3$- and $4$-spheres \cite{chak,con02}, and quantum
Heisenberg manifolds \cite{chak2}.

Given Connes' differential calculus $(\Om^*_\cD\cA,d)$ over an
algebra $\cA$, let $P$ be a right projective $\cA$-module of
finite rank. In the spirit of the Serre--Swan theorem \ref{sp60},
one can think of $P$ as being a non-commutative vector bundle. By
virtue of Theorem \ref{+893}, it admits a connection. Let us
construct this connection in an explicit form.

Given a generic right finite projective module $P$ over a complex
ring $\cA$, let
\be
\bp: \Bbb C^N\op\ot_\Bbb C \cA \to P, \qquad i_P: P\to \Bbb
C^N\op\ot_\Bbb C \cA
\ee
be the corresponding projection and injection, where $\ot_\Bbb C$
denotes the tensor product over $\Bbb C$, while $\ot_\cA$ stands
for the tensor product over the ring $\cA$. There is the
composition of morphisms
\mar{+941}\beq
P\ar^{i_p} \Bbb C^N\op\ot_\Bbb C\cA \ar^{\id\ot d}\Bbb
C^N\op\ot_\Bbb C \Om^1(\cA)\ar^{\bp} P\op\ot_\cA\Om^1(\cA),
\label{+941}
\eeq
where the canonical module isomorphism
\be
\Bbb C^N\op\ot_\Bbb C\Om^1(\cA)=(\Bbb C^N\op\ot_\Bbb
C\cA)\op\ot_\cA\Om^1(\cA)
\ee
is used. It is readily observed that the composition (\ref{+941})
denoted briefly as $\bp\circ d$ is a right universal connection on
the module $P$.

Given the universal connection $\bp\circ d$ on a right finite
projective module $P$ over a $*$-algebra $\cA$, let us consider
the morphism
\be
P\ar^{\bp\circ d} P\op\ot_\cA \Om^1(\cA)\ar^{\id\ot\pi}P\op\ot_\cA
\Om^1_\cD\cA.
\ee
This is a right connection $\nabla_0$ on the module $P$ with
respect to Connes' differential calculus. Any other right
connection $\nabla$ on $P$ with respect to Connes' differential
calculus takes the form
\mar{+945}\beq
\nabla = \nabla_0 + \si=(\id\ot\pi)\circ\bp\circ d +\si
\label{+945}
\eeq
where $\si$ is an $\cA$ module morphism
\be
\si: P\to P\op\ot_\cA\Om^1_\cD\cA.
\ee
The term $\si$ in the connection $\nabla$ (\ref{+945}) is called a
{\it non-commutative gauge field}. \index{non-commutative gauge
field}

\section{Differential calculus over Hopf algebras}

Hopf algebras make a contribution to many quantum models
\cite{chan,char,majid2}. For instance, quantum groups are
particular Hopf algebras. In a general setting, any
non-cocommutative Hopf algebra can be treated as a quantum group
\cite{chan}. However, the development of differential calculus and
differential geometry over Hopf algebras has met some problems
\cite{book05}.

Let $\cA$ be a complex unital algebra, i.e., a $\Bbb C$-ring. The
{\it tensor product} $\cA\ot\cA$ of algebras $\cA$ \index{tensor
product!of algebras} is defined as that of vector spaces $\cA$
provided with the multiplication
\be
(a\ot b)\ot (a'\ot b')=(aa')\ot (bb').
\ee
Let us write the multiplication operation of the algebra $\cA$ as
a $\Bbb C$-linear morphism
\be
m:a\ot b\to ab, \qquad a,b\in\cA.
\ee

A {\it coalgebra} \index{coalgebra} $\cA$ is defined as a vector
space $\cA$ provided with the following linear morphisms:

$\bullet$ a coassociative {\it comultiplication}
\index{comultiplication} $\Delta: \cA\to \cA\ot\cA$,

$\bullet$ a {\it counit} \index{counit} $\e: \cA\to\Bbb C$,

\noindent which obey the relations
\mar{w35}\beq
(\Delta\ot\id)\circ\Delta=(\id\ot\Delta)\circ\Delta, \quad
(\e\ot\id)\circ\Delta=(\id\ot\e)\circ\Delta=\id. \label{w35}
\eeq
As a shorthand, one writes
\mar{w38}\beq
\Delta(a)=\sum a_{(1)}\ot a_{(2)}, \quad
(\Delta\ot\id)\circ\Delta(a)= \sum a_{(1)}\ot a_{(2)}\ot a_{(3)}.
\label{w38}
\eeq
A comultiplication and a counit extend to tensor products
pairwise, i.e.,
\be
\Delta(a\ot a')=\Delta(a)\ot\Delta(a'), \qquad \e(a\ot
a')=\e(a)\e(a').
\ee

A {\it bi-algebra} \index{bi-algebra} $(\cA,m,\Delta,\e)$ is
defined as an associative algebra $\cA$ which is also a coalgebra
so that
\be
&&\Delta(\bb)=\bb\ot\bb, \\
&& \sum (ab)_{(1)}\ot (ab)_{(2)}=
\sum (a_{(1)}b_{(1)})\ot (a_{(2)}b_{(2)}), \qquad a,b\in\cA,\\
&& \e(\bb)=1, \qquad \e(ab)=\e(a)\e(b).
\ee

A {\it Hopf algebra} \index{Hopf algebra} $(\cA,m,\Delta,\e, S)$
is a bi-algebra $\cA$ endowed with a linear morphism
$S:\cA\to\cA$, called the {\it antipode}, \index{antipode} such
that
\be
m((S\ot\id)\Delta(a))=m((\id\ot S)\Delta(a))=\e(a)\bb.
\ee
It obeys the relations
\be
&& S(\bb)=\bb, \qquad S(ab)=S(b)S(a), \quad a,b\in\cA,\\
&& \e\circ S=\e, \qquad \Delta \circ S=P\circ (S\ot S)\circ \Delta,
\ee
where  $P: a\ot b \mapsto b\ot a$ is the transposition operator.
Let us note that, given a bi-algebra $\cA$, there is a unique
antipode, if any, such that $\cA$ becomes a Hopf algebra.

For the sake of brevity, we call $\Delta$, $\e$ and $S$ the {\it
co-operations} \index{co-operation} of a Hopf algebra.

A Hopf algebra is said to be {\it cocommutative} \index{Hopf
algebra!cocommutative} if $P\circ\Delta=\Delta$. If a Hopf algebra
is commutative or cocommutative, then $S^2=\id$. If
$(\cA,m,\Delta,\e,S)$ is a Hopf algebra whose antipode $S$ is
invertible, then
\mar{w54}\beq
\cA^t=(\cA,m,P\circ\Delta,\e,S^{-1}) \label{w54}
\eeq
is also a Hopf algebra.

Let $\cA$ be an involutive algebra. A Hopf algebra
$(\cA,m,\Delta,\e, S)$ is said to be {\it involutive} \index{Hopf
algebra!involutive} if
\be
&& \Delta(a^*)=\Delta(a)^*=\sum a^*_{(1)}\ot a^*_{(2)},\\
&& \e(a^*)=\ol{\e(a)}, \qquad S(S(a^*)^*)=a.
\ee

A Hopf algebra is called {\it quasi-triangular} \index{Hopf
algebra!quasi-triangular} if there exists an invertible element
$\cR=r_i\ot r^i\in\cA\ot\cA$ which obeys the relations
\mar{ws200}\bea
&& (\Delta\ot\id)(\cR)=\cR_{13}\cR_{23}, \label{ws200a}\\
&& (\id\ot\Delta)(\cR)=\cR_{12}\cR_{23}, \label{ws200b}\\
&& (P\circ\Delta)(a)=\cR\Delta(a)\cR^{-1}, \qquad a\in\cA, \label{ws200c}
\eea
where
\be
\cR_{12}=r_k\ot r^k\ot\bb, \qquad \cR_{13}=r_k\ot \bb\ot r^k,
\qquad \cR_{23}=\bb\ot r_k\ot r^k.
\ee
The element $\cR$ is said to be the {\it universal $R$-matrix}
\index{universal $R$ matrix} of a Hopf algebra $\cA$. It satisfies
the {\it quantum Yang--Baxter equation} \index{Yang--Baxter
equation!quantum}
\mar{ws201}\beq
\cR_{12}\cR_{13}\cR_{23}=\cR_{23}\cR_{13}\cR_{12}. \label{ws201}
\eeq
One can show that
\be
(S\ot\id)(\cR)=\cR^{-1}, \qquad (\id\ot S)(\cR^{-1})=\cR, \qquad
(S\ot S)(\cR)=\cR.
\ee
The antipode of a quasi-triangular Hopf algebra is always
invertible.

\begin{ex} \label{ws202} \mar{ws202} Let $\cG$ be a finite-dimensional
Lie algebra whose basis is $\{e_k\}$. Its universal enveloping
algebra $\ol\cG$ is provided with the structure of a cocommutative
Hopf algebra $U\cG$, called the {\it classical Hopf algebra},
\index{Hopf algebra!classical} with respect to the co-operations
\be
\Delta(e_k)=e_k\ot \bb +\bb\ot e_k, \qquad \e(e_k)=0,\qquad
S(e_k)=-e_k,
\ee
extended by linearity to $\ol\cG$. It is a quasi-triangular Hopf
algebra where $\cR=\bb\ot\bb$.
\end{ex}

\begin{ex} \label{w31} \mar{w31} Let $G$ be a finite group and $\Bbb
C G$ the complex group ring \cite{book05}. It is brought into the
cocommutative Hopf algebra, called the {\it group Hopf algebra},
\index{group Hopf algebra}  characterized by the co-operations
\be
\Delta(g)=g\ot g, \qquad \e(g)=1, \qquad S(g)=g^{-1}, \qquad g\in
G,
\ee
extended by linearity to $\Bbb C G$.
\end{ex}

\begin{ex} \label{w32} \mar{w32} Given a finite group $G$, let us consider
the algebra $\Bbb C(G)=\Bbb C^G$ of complex functions on $G$. Let
us identify
\be
\Bbb C(G)\ot \Bbb C(G)=\Bbb C(G\times G).
\ee
Then the algebra $\Bbb C(G)$ is a commutative Hopf algebra with
respect to the co-operations
\mar{w32'}\beq
\Delta(f)(g,g')=f(gg'), \qquad \e(f)=f(\bb), \qquad
S(f)(g)=f(g^{-1})  \label{w32'}
\eeq
for all $f\in \Bbb C(G)$. It is called the {\it group function
Hopf algebra}. \index{group function Hopf algebra}
\end{ex}

\begin{ex} \label{w33} \mar{w33} The Hopf algebra $U_q(b_+)$, where $q$
is a non-zero real number, is generated by the elements $\bb$,
$a$, $g$ and $g^{-1}$ obeying the relations
\be
&& gg^{-1}=g^{-1}g=\bb, \qquad ga=qag, \\
&& \Delta(a)=a\ot\bb +g\ot a, \qquad \Delta(g)=g\ot g, \qquad
\Delta(g^{-1})= g^{-1}\ot g^{-1},\\
&& \e(a)=0, \qquad \e(g)=\e(g^{-1})=1, \\
&& S(a)=-g^{-1}a, \qquad S(g)=g^{-1}, \qquad S(g^{-1})=g.
\ee
\end{ex}

\begin{ex} \label{w45} \mar{w45}
Let $\cG$ be a finite-dimensional semi-simple Lie algebra. It
yields an associated complex Lie group $G$ of $n\times n$ complex
matrices $b$ which obey a family of polynomial equations $p(b)=0$.
Correspondingly, we have an algebraic variety with the coordinate
algebra $\Bbb C[G]$ of polynomials $\Bbb C[b^i_j]$ in $n^2$
variables modulo the relations $p(b)=0$. It is a commutative Hopf
algebra with respect to the co-operations
\be
\Delta (b^i_j)=b^i_k\ot b^k_j, \qquad \e(b^i_j)=\dl^i_j.
\ee
Its antipode is given algebraically via a matrix of cofactors of
the matrix $b^i_j$, i.e.,
\be
S(b^i_k)b^k_j=b^i_kS(b^k_j)=\dl^i_j\bb.
\ee
For instance, the Hopf algebra $\Bbb C[SL(2)]$ is generated by
four elements $a$, $b$, $c$ and $d$ which are assembled into the
matrix
\be
A=\left(\begin{array}{cc}
a & b\\
c & d
\end{array}\right)
\ee
and obey the polynomial equation
\be
A\,{\rm det}\,A=({\rm det}\,A)A=A,
\ee
i.e., ${\rm det}\,A=ad-cb$ is the unit element. The co-operations
of this Hopf algebra can be written in the compact form
\be
\Delta (A)=A\ot A, \qquad \e (A)=\left(\begin{array}{cc}
1 & 0\\
0 & 1
\end{array}\right),\qquad
S(A)=A^{-1} = \left(\begin{array}{cc}
d & -b\\
-c & a
\end{array}\right).
\ee
It is quasi-triangular Hopf algebra where $\cR=\bb\ot\bb$.
\end{ex}

Given a Hopf algebra $(\cA,m,\Delta,\e, S)$, let us consider the
space $\cA^*=\hm(A,\Bbb C)$ of complex linear forms on $\cA$. It
is a complex ring with respect to the {\it convolution product}
\index{convolution product}
\mar{w260}\ben
&& f*f'= (f\ot f')\circ \Delta, \qquad f,f'\in\cA^*, \label{w260}\\
&& (f*f')(a)=\sum f(a_{(1)})f'(a_{(2)}), \qquad a\in A. \nonumber
\een
The unit of $\cA^*$ is the counit $\e$ of $\cA$. Sometimes, the
notion of the convolution product (\ref{w260}) includes the linear
morphisms
\mar{w261}\ben
&& f*a=[(\Id\ot f)\circ\Delta](a)=\sum a_{(1)}f(a_{(2)}), \label{w261}\\
&& a*f=[(f\ot\Id)\circ\Delta](a)=\sum f(a_{(1)})a_{(2)},
\quad a\in \cA, \nonumber
\een
of the complex space $\cA$ \cite{woron87}. There are the relations
\mar{w262} \beq
(f*f')(a)=f'(a*f)=f(f'*a), \qquad a\in \cA, \qquad f,f'\in\cA^*.
\label{w262}
\eeq

A Hopf algebra $(\cA',m',\Delta',\e', S')$ is said to be {\it
dually paired} \index{Hopf algebra!dually paired} with a Hopf
algebra $(\cA,m,\Delta,\e, S)$ if there exists a non-degenerate
interior product
\be
\lng.,.\rng: \cA\ot\cA'\to\Bbb C
\ee
which satisfies the relations
\mar{qm850}\bea
&& \lng a,a'b'\rng= \lng\Delta(a),a'\ot b'\rng, \qquad
a,b\in\cA, \quad a',b'\in\cA',
\label{qm850a}\\
&& \lng a\ot b,\Delta'(a')\rng=\lng ab,a'\rng,  \label{qm850b}\\
&&  \e'(a')=\lng\bb,a'\rng,\qquad \e(a)=\lng a,\bb\rng, \label{qm850c}\\
&& \lng a,S'(a')\rng= \lng S(a),a'\rng,\label{qm850d}
\eea
where $\lng,\rng$ extends to tensor products pairwise, i.e.,
\mar{w50}\beq
\lng a\ot b,a'\ot b'\rng=\lng a,a'\rng \lng b,b'\rng. \label{w50}
\eeq
If $\cA$ is an involutive Hopf algebra, the involutive Hopf
algebra $\cA'$ dually paired with $\cA$ should satisfy the
additional condition
\be
\lng a^*,a'\rng=\lng a, (S'(a'))^*\rng.
\ee

For instance, the Hopf algebra $U_q(b_+)$ in Example \ref{w33} is
dually paired with itself by
\be
\lng a,a\rng=1, \qquad \lng g,g\rng=q, \qquad \lng a,g\rng=\lng
g,a\rng=0.
\ee
If $\cG$ is a finite-dimensional complex semi-simple Lie algebra
represented by $n\times n$ matrices $\rho(a)$, $a\in\cG$, the Hopf
algebras $U\cG$ in Example \ref{ws202} and $\Bbb C[G]$ in Example
\ref{w45} are dually paired with respect to the interior product
\mar{w48}\beq
\lng a,b^i_j\rng=\rho(a)^i_j. \label{w48}
\eeq

Let a Hopf algebra $(\cA,m,\Delta,\e, S)$ be a finite-dimensional
vector space, and let $\cA'$ be its algebraic dual. Then the
equality (\ref{w50}) yields an isomorphism
\be
(\cA\ot\cA)'= \cA'\ot\cA'.
\ee
In this case, the relations (\ref{qm850a}) -- (\ref{qm850d})
provide $\cA'$ with a unique Hopf algebra structure
$(m',\Delta',\e', S')$, called the {\it dual Hopf algebra}.
\index{dual Hopf algebra}

One says that a Hopf algebra $\cH$ acts on a complex algebra $A$
on the left (or $A$ is a {\it module over a Hopf algebra} $\cH$)
\index{module!over a Hopf algebra} if $\cH$ acts $h\triangleright
a$  on $A$
 as a vector space
such that
\be
&& h\triangleright(ab)=\sum (h_{(1)}\triangleright a)
(h_{(2)}\triangleright b), \qquad a,b\in A,\\
&& h\triangleright\bb=\e(h)\bb, \qquad h\in\cH,
\ee
i.e., the multiplication and the unit in $A$ commute with the
action of $\cH$.

In particular, any Hopf algebra acts on itself as an algebra.
Namely, the left and right {\it adjoint action} \index{adjoint
action!of a Hopf algebra} of a Hopf algebra $\cH$ on itself is
given by the formulae
\mar{w41}\ben
&& h\triangleright h': h\ot h'\mapsto \sum h_{(1)}h'S(h_{(2)}),
\label{w41}\\
&& h'\triangleleft h: h\ot h'\mapsto \sum S(h_{(1)})h'h_{(2)}.
\een
For instance, the left and right adjoint action of the classical
Hopf algebra $U\cG$ in Example \ref{ws202} reads
\mar{w42}\beq
a\triangleright b =[a,b], \qquad b\triangleleft a=[b,a], \qquad
a,b\in \cG. \label{w42}
\eeq
The left adjoint action of the group Hopf algebra  $\Bbb C G$ in
Example \ref{w31} coincides with the adjoint action of the group
$G$ on itself. Since the Hopf algebras $\Bbb C(G)$ in Example
\ref{w32} and $\Bbb C[G]$ in Example \ref{w45} are commutative,
their adjoint action is trivial.

If a Hopf algebra $\cH'$ is dually paired with a Hopf algebra
$\cH$, its action on $\cH$ reads
\mar{w40}\beq
h'\triangleright h=\sum h_{(1)}\lng h',h_{(2)}\rng, \qquad h'\in
\cH', \qquad h\in\cH. \label{w40}
\eeq
For instance, the Hopf algebra $\Bbb C[G]$ acts on the classical
Hopf algebra $U\cG$ by the law
\be
b^i_j\triangleright a=\bb\rho(a)^i_j +a\dl^i_j.
\ee

A coalgebra $(\cA,\Delta,\e)$ is said to be a module of a Hopf
algebra $\cH$ if it is an $\cH$-module as an algebra and,
additionally, the relations
\be
\Delta(h\triangleright a)=\sum (h_{(1)}\triangleright a_{(1)})\ot
(h_{(2)}\triangleright a_{(2)}), \qquad \e(h\triangleright
a)=\e(h)\e(a)
\ee
hold for all $h\in\cH$ and $a\in \cA$.

A left {\it coaction} \index{coaction} of a Hopf algebra $\cH$ on
a vector space $V$  is given by a map $\bt: V\to \cH\ot V$ such
that
\mar{w36}\beq
(\id \ot\bt)\circ\bt=(\Delta\ot\id)\circ\bt, \qquad
(\e\ot\id)\circ\bt=\id. \label{w36}
\eeq
A right coaction $V\to V\ot\cH$ of $\cH$ on $V$ is similarly
defined. We write
\mar{w37}\beq
\bt(v)=\sum \ol v^{(1)}\ot v^{(2)}, \qquad \ol v^{(1)}\in \cH,
\qquad v^{(2)}\in V.
\eeq
One says that $V$ is a (left) {\it comodule} \index{comodule!over
a Hopf algebra} over a Hopf algebra $\cH$. For instance the {\it
trivial comodule} \index{comodule!trivial} is $\Bbb C$ where
\be
\bt(\la)=\bb\ot\la, \qquad \la\in\Bbb C.
\ee
If $V$ is a $\cH$-bimodule, the additional relations
\mar{w270}\beq
\bt(hvh')=\Delta(h)\bt(v)\Delta(h'),  \qquad h,h'\in\cH,
\label{w270}
\eeq
hold. If an element $v\in V$ satisfies the equality
\mar{w271}\beq
\bt(v)=\bb\ot v, \label{w271}
\eeq
one says that $v$ is left-invariant.

A Hopf algebra $\cH$ coacts on a unital algebra $A$ if $A$ is a
$\cH$-comodule and $\bt$ is an algebra homomorphism of $A$ to the
algebra $\cH\ot A$, i.e.,
\be
\bt(ab)=\bt(a)\bt(b), \qquad \bt(\bb_A)=\bb\ot\bb_A, \qquad a,b\in
A.
\ee
Then $A$ is called a (left) {\it $\cH$-comodule algebra}.
\index{comodule algebra}  A glance at the formulae (\ref{w35}) and
(\ref{w36}) shows that any Hopf algebra coacts
\be
\Delta:\cH\to \cH\ot\cH
\ee
on itself as an algebra.

A Hopf algebra coacts $\cH$ on a Hopf algebra $\cA$ if it coacts
on $\cA$ as an algebra and, additionally, this coaction commutates
with co-operations of $\cA$. In particular, any Hopf algebra
coacts on itself on the left as a coalgebra by the law
\mar{w39}\beq
{\rm ad}_l(h)= \sum h_{(1)}S(h_{(3)})\ot h_{(2)} \label{w39}
\eeq
(see the shorthand (\ref{w38})). Accordingly, the right coaction
reads
\mar{w250}\beq
{\rm ad}_r(h)= \sum h_{(2)}\ot S(h_{(1)})h_{(3)}. \label{w250}
\eeq

Turn now to the differential calculus over Hopf algebras.

Let $(\cA,m,\Delta,\e,S)$ be a Hopf algebra with an invertible
antipode $S$. A differential calculus over $\cA$ is that over
$\cA$ as a non-commutative ring which obeys some additional
conditions related to the co-operations in $\cA$
\cite{klim,woron}.

Here, we focus on a minimal first order differential calculus
(henceforth {\it FODC}) \index{FODC} over $\cA$ (see Remark
\ref{w310} below). It is a left $\cA$-module $\Om^1$ generated by
elements $da$, $a\in\cA$, and provided with the $\cA$-bimodule
structure by the rule
\mar{w266}\beq
(da)b=d(ab)-adb, \qquad a,b \in \cA. \label{w266}
\eeq
The universal FODC $\Om^1(\cA)$ over $\cA$ is isomorphic to the
sub-bimodule $\Ker m$ of the $\cA$-bimodule $\cA\ot\cA$. Then any
FODC $\Om^1$ over $\cA$ is uniquely characterized by a
sub-bimodule
\mar{w267}\beq
\cN=\{\sum (a_k\ot b_k -a_kb_k\ot\bb)\in \Ker m \, :\, \Om^1\ni
\sum a_kdb_k=0\} \label{w267}
\eeq
of the $\cA$-bimodule $\Ker m$. Obiously, $\Om^1=\Om^1(\cA)/\cN$.

A FODC $(\Om^1,d)$ over a Hopf algebra $\cA$ is called {\it
left-covariant} \index{FODC!left-covariant} if it possesses the
structure of a left $\cA$-comodule
\be
\Delta_l:\Om^1\to \cA\ot\Om^1
\ee
(see relations (\ref{w36}) and (\ref{w270})) such that
\mar{w268}\beq
\Delta_l(adb)=\Delta(a)(\id\ot d)\Delta(b), \qquad a,b\in \cA.
\label{w268}
\eeq
Let us note that the original notion of a left-covariant FODC in
\cite{woron} implies the above mentioned one. Similarly, a FODC
over a Hopf algebra $\cA$ is called {\it right-covariant}
\index{FODC!right-covariant} if it possesses the structure of a
right $\cA$-comodule $\Delta_r:\Om^1\to \Om^1\ot\cA$ such that
\be
\Delta_r(adb)=\Delta(a)(d\ot \id)\Delta(b), \qquad a,b\in \cA.
\ee
A FODC over $\cA$ is called {\it bicovariant}
\index{FODC!bicovariant} if it is both a left- and right-covariant
FODC satisfying the relation
\be
(\id\ot\Delta_r)\circ\Delta_l=(\Delta_l\ot \id)\circ \Delta_r.
\ee

We here restrict our consideration to left-covariant FODCs over a
Hopf algebra $\cA$. An $\cA$-bimodule $V$ which is also an
$\cA$-comodule is called a {\it covariant bimodule}.
\index{bimodule!covariant} Covariant bimodules possess the
following important properties.

\begin{theo} \label{w287} \mar{w287}
Let $(V,\Delta_l)$ be a left-covariant bimodule over a Hopf
algebra $\cA$, and let $V_{\rm inv}$ denote the complex vector
space of left-invariant elements of $V$ obeying the condition
(\ref{w271}). There exists an epimorphism of complex spaces
$\rho:V\to V_{\rm inv}$ such that
\mar{w309}\beq
\rho(av)=\e(a)\rho(v), \qquad v\in V, \qquad a\in \cA.
\label{w309}
\eeq
Moreover, if
\mar{w288}\beq
\Delta_l(v)=\sum a_k\ot v_k, \qquad a_k\in\cA, \qquad v_k\in V,
\label{w288}
\eeq
then
\mar{w289}\beq
\rho(v)=\sum S(a_k)v_k. \label{w289}
\eeq
\end{theo}

\begin{theo} \label{w272} \mar{w272}
Let $\{w_i\}_{i\in I}$ be a basis for the vector space $V_{\rm
inv}$ of left-invariant elements of $V$ in Theorem \ref{w287}.
Then the following hold.

$\bullet$ Any element of $V$ is uniquely decomposed into the
finite sums
\mar{w280}\beq
v=\sum a_iw_i, \qquad v=\sum w_i b_i, \qquad a_i,b_i\in\cA.
\label{w280}
\eeq

$\bullet$ There exist complex linear forms $f_{ij}\in \cA^*$ on
the complex vector space $\cA$ such that
\mar{w281,2}\ben
&& w_ib=\op\sum_j (f_{ij}* b)w_j, \qquad b\in\cA, \label{w281}\\
&& aw_i=\op\sum_j w_j(f_{ij}\circ S^{-1}*a), \qquad a\in\cA,\label{w282}
\een
where $*$ is the convolution product (\ref{w260}) -- (\ref{w261}).

$\bullet$ The forms $f_{ij}$ are uniquely defined. They obey the
relations
\mar{w283}\beq
f_{ij}(ab)=\op\sum_k f_{ik}(a)f_{kj}(b), \qquad
f_{ij}(\bb)=\dl_{ij}, \qquad a,b\in\cA. \label{w283}
\eeq
\end{theo}

Theorem \ref{w272} describes all left-covariant bimodules. Namely,
given a family of complex linear forms $\{f_{ij}\}_{i,j\in I}$
obeying the conditions (\ref{w281}) -- (\ref{w282}), one can
consider a free left $\cA$-module $V$ generated by elements $v_i$
indexed by the set $I$ and provided it with the operations
\mar{w286}\beq
w_ib= \op\sum_j(f_{ij}*b)w_j, \qquad \Delta_l(aw_i)= \Delta
(a)(\id \ot w_i). \label{w286}
\eeq

For instance, if $V=\Om^1$ is a left-covariant FODC over $\cA$,
the condition
\be
\Delta_l(db)=\sum b_{(1)}\ot db_{(2)}, \qquad b\in\cA,
\ee
(\ref{w268}) is exactly the relation (\ref{w288}) in Theorem
\ref{w287}. Therefore, the projection $\rho$ of an element $db\in
\Om^1$ to the space $V_{\rm inv}$ is
\mar{w291}\beq
\rho(db)=\sum S(b_{(1)}) db_{(2)} \label{w291}
\eeq
in accordance with the formula (\ref{w309}). It is called the {\it
Maurer--Cartan form}. \index{Maurer--Cartan form!on a Hopf
algebra} By virtue of the equality (\ref{w288}), the projection
$\rho$ of an arbitrary element $\f\in\Om^1$ reads
\be
\rho(\f)=\rho(\op\sum_k a_kdb_k)=\op\sum_k\e(a_k)\sum
S(b_{k(1)})db_{k(2)}.
\ee
It follows that, for any element $\f\in\Om^1$, there exists an
element $b\in \Ker\e$ such that $\rho(\f)=\rho(db)$. Namely, we
have
\be
&& \rho(\op\sum_k a_kdb_k)= \rho(db), \\
&& b=\op\sum_k\e(a_k)(b_k-\e(b_k)\bb).
\ee
Consequently, any invariant element of a left-covariant FODC
$\Om^1$ is of the form (\ref{w291}) where $b\in \Ker\e$.

The covariance conditions impose some restrictions on the
sub-bimodule $\cN$ (\ref{w267}). The well-known Woronowicz theorem
states the following \cite{woron}.

\begin{theo} \label{w292} \mar{w292}
Let $\cR$ be a right ideal of $\cA$ contained in $\Ker\e$ and
\mar{w293}\beq
\cN=\{\sum aS(b_{(1)})\ot b_{(2)}\, :\, b\in\cR,\,a\in\cA\}.
\eeq
Then $\cN$ is a sub-bimodule of the universal FODC $\Om^1(\cA)$,
and
\mar{w294}\beq
\Om^1_\cN=\Om^1(\cA)/\cN \label{w294}
\eeq
 is a left-covariant FODC. Moreover, any left-covariant
FODC can be obtained in this way.
\end{theo}

It follows that
\be
\Om^1_\cN\ni da=0
\ee
iff either $a=\bb$ of $a\in \cR$. In particular, the universal
FODC over a Hopf algebra is always left-covariant, right-covariant
and bicovariant.

Given the right ideal $\cR$ in Theorem \ref{w292}, let us consider
the subspace
\mar{w300}\beq
T=\{\chi\in\cA^*\, :\, \chi(\bb)=0;\,\chi(a)=0,\, a\in\cR\}
\label{w300}
\eeq
of the complex dual $\cA^*$ of $\cA$. Somebody calls $T$ the {\it
quantum tangent space} \index{quantum tangent space} \cite{heck}.
Moreover, $T$ is a complex Lie algebra with respect to the bracket
\be
[\chi,\chi']=\chi*\chi'-\chi'*\chi.
\ee
Therefore, it is also called the {\it quantum Lie algebra}
\index{quantum Lie algebra of a Hopf algebra} \cite{schmud}. One
can show the following.

$\bullet$ There is the interior product
\mar{w303}\ben
&& \lng.,.\rng: \Om^1_\cN\times
T\to\Bbb C, \label{w303}\\
&& \lng\f,\chi\rng=\chi(a), \qquad \rho(da)=\rho(\f), \qquad a\in\Ker\e,
\nonumber
\een
which obeys the relations
\be
\lng a\f,\chi\rng=\e(a)\lng\f.\chi\rng, \qquad \lng
da,\chi\rng=\chi(a).
\ee

$\bullet$ The interior product (\ref{w303}) makes the space
$V_{\rm inv}$ of left-invariant elements of the FODC $\Om^1_\cN$
(\ref{w294}) and the quantum tangent space $T$ (\ref{w300}) into a
dual pair.

$\bullet$ In order to simplify the notation, let us assume that
$T$ is finite-dimensional. Let $\{\chi_i\}_{i\in I}$ be a basis
for $T$, $\{w_i\}$ the dual basis for $V_{\rm inv}$  and $\{a_j\}$
the set of elements of $\Ker \e$ defined by the condition
$\chi_i(a_j)=\dl_{ij}$. Then the relations
\mar{w305}\ben
&& w_i=\rho(da_i), \qquad f_{ij}(b)=\chi_i(a_jb), \qquad b\in\cA, \nonumber\\
&& \f=\op\sum_i\lng\f,\chi_i\rng w_i, \qquad da=\op\sum_i \chi_i(a)w_i,
\qquad a\in\cA, \nonumber\\
&& \chi_i*ab=\op\sum_j (\chi_i*a)(f_{ij}*b) +a(\chi_i*b) \label{w305}
\een
hold (see Theorem \ref{w272} for the notation).

In view of the relation (\ref{w305}), one can think of the
elements
\mar{w306}\beq
u\in\hm(\cA,\cA), \qquad u(a)=\chi*a=\sum a_{(1)}\chi(a_{(2)})
\label{w306}
\eeq
as being vector fields for the left-covariant FODC $\Om^1_\cN$
which obey the deformed Leibniz rule (\ref{w305}). They are called
{\it invariant vector fields} \index{invariant vector field}
\cite{ash}, and can be defined in an intrinsic way as follows.

\begin{prop} \label{w307} \mar{w307}
An element $u\in \hm(\cA,\cA)$ is an invariant vector field iff it
satisfies the equality
\be
u=(\id \ot \e)\circ(\id\ot u)\circ\Delta.
\ee
\end{prop}

However, the notion of invariant vector fields meets the standard
problem of non-commutative geometry investigated in Section 4.3.
Namely, the space of invariant vector fields fails to be an
$\cA$-module. For instance, one has suggested to overcome this
difficulty by appealing to Cartan pairs in Section 4.3. The
general problem of differential operators on Hopf algebras has
been studied in \cite{lunts}.

\begin{rem} \label{w310} \mar{w310}
Let us say a few words on a higher order differential calculus
over a Hopf algebra. A minimal differential calculus $(\Om^*,d)$
over a Hopf algebra $\cA$ seen as an associative algebra  is said
to be left-covariant if $\Om^*$ possesses the structure of a left
$\cA$-comodule
\be
\Delta_l:\Om^*\to \cA\ot\Om^*
\ee
such that
\be
\Delta_l(a_0da_1\cdots da_k)=\Delta (a_0)[(\id\ot
d)\Delta(a_1)]\cdots [(\id\ot d)\Delta(a_k)], \quad a_i\in\cA.
\ee
Let $\Om^*(\cA)$ be the universal differential calculus over $\cA$
and $\cN^*$ its differential graded ideal such that
\be
\Delta_l(\cN^*)\subset \cA\ot\cN^*.
\ee
Then, the differential calculus $\Om^*(\cA)/\cN^*$ is
left-covariant. Similarly, right-covariant differential calculi
over a Hopf algebra $\cA$ are considered. Let us note that
different ideals $\cN^*$ may lead to isomorphic quotients
$\Om^*(\cA)/\cN^*$. We always choose $\cN^*$ the kernel of the
morphism $\Om^*(\cA)\to \Om^*$.
\end{rem}

\chapter{Appendix. Cohomology}

For the sake of convenience of the reader, several topics on
cohomology are compiled in this Chapter.

\section{Cohomology of complexes}

This Section summarizes the relevant basics on homology and
cohomology of complexes of modules over a commutative ring
\cite{mcl,massey}.

Let $\cK$ be a commutative ring.  A sequence
\mar{b3256}\beq
0\to B^0 \ar^{\dl^0} B^1 \ar^{\dl^1}\cdots B^p\ar^{\dl^p}\cdots
\label{b3256}
\eeq
of modules $B^p$ and their homomorphisms $\dl^p$ is said to be a
{\it cochain complex} \index{cochain complex} (henceforth, simply,
a {\it complex}) \index{complex} if
\be
\dl^{p+1}\circ \dl^p =0, \qquad  p\in \Bbb N,
\ee
i.e., $\im \dl^p\subset \Ker \dl^{p+1}$. The homomorphisms $\dl^p$
are called {\it coboundary operators}. \index{coboundary operator}
For the sake of convenience, let us denote $B^{-1}=0$ and
$\dl^{-1}:0\to B^0$. Elements of the module $B^p$ are said to be
{\it $p$-cochains}, \index{cochain} while elements of its
submodules $\Ker \dl^p\subset B^p$ and  $\im \dl^{p-1}\subseteq
\Ker \dl^p$ are called {\it $p$-cocycles} \index{$p$-cocycle} and
{\it $p$-coboundaries}, \index{$p$-coboundary} respectively. The
{\it $p$-th cohomology group} \index{cohomology group} of the
complex $B^*$ (\ref{b3256}) is the factor module
\be
H^p(B^*)= \Ker \dl^p/\im \dl^{p-1}.
\ee
It is a $\cK$-module. In particular, $H^0(B^*)=\Ker \dl^0$.

A complex (\ref{b3256}) is said to be {\it exact}
\index{complex!exact} at a term $B^p$ if $H^p(B^*)=0$. It is an
exact sequence if all cohomology groups are trivial.

A complex $(B^*,\dl^*)$ is called {\it acyclic}
\index{complex!acyclic} if its cohomology groups $H^{p>0}(B^*)$
are trivial. It is acyclic if there exists a {\it homotopy
operator} \index{homotopy operator} ${\bf h}$, defined as a set of
module morphisms
\be
{\bf h}^{p+1}: B^{p+1}\to B^p, \qquad p\in\Bbb N,
\ee
such that
\be
{\bf h}^{p+1}\circ \dl^p + \dl^{p-1}\circ {\bf h}^p=\id B^p,
\qquad p\in\Bbb N_+.
\ee
Indeed, if $\dl^pb^p=0$, then $b^p=\dl^{p-1}({\bf h}^pb^p)$, and
$H^{p>0}(B^*)=0$.

A complex $(B^*,\dl^*)$ is said to be a {\it resolution}
\index{resolution} of a module $B$ if it is acyclic and
$H^0(B^*)=B$.

The following are the standard constructions of new complexes from
old ones.

$\bullet$ Given complexes $(B^*_1,\dl^*_1)$ and $(B^*_2,\dl^*_2)$,
their {\it direct sum} \index{direct sum!of complexes}
$B^*_1\oplus B^*_2$ is a complex of modules
\be
(B^*_1\oplus B^*_2)^p=B^p_1\oplus B^p_2
\ee
with respect to the coboundary operators
\be
\dl^p_\oplus(b^p_1 + b^p_2)=\dl^p_1b^p_1 +\dl^p_2 b^p_2.
\ee

$\bullet$ Given a subcomplex $(C^*,\dl^*)$ of a complex
$(B^*,\dl^*)$, the {\it factor complex} \index{factor complex}
$B^*/C^*$ is defined as a complex of factor modules $B^p/C^p$
provided with the coboundary operators
\be
\dl^p[b^p]=[\dl^p b^p],
\ee
where $[b^p]\in B^p/C^p$ denotes the coset of the element $b^p$.

$\bullet$ Given complexes $(B^*_1,\dl^*_1)$ and $(B^*_2,\dl^*_2)$,
their {\it tensor product} \index{tensor product!of complexes}
$B^*_1\ot B^*_2$ is a complex of modules
\be
(B^*_1\ot B^*_2)^p=\op\oplus_{k+r=p} B^k_1\ot B^r_2
\ee
with respect to the coboundary operators
\be
\dl^p_\ot(B^k_1\ot B^r_2)=(\dl_1^kb^k_1)\ot b^r_2 +
(-1)^kb^k_1\ot(\dl_2^rb^r_2).
\ee

A {\it cochain morphism} \index{cochain morphism} of complexes
\mar{spr32'}\beq
\g:B^*_1\to B^*_2 \label{spr32'}
\eeq
is defined as a family of degree-preserving homomorphisms
\be
\g^p: B^p_1\to B^p_2, \qquad p\in\Bbb N,
\ee
which commute with the coboundary operators, i.e.,
\be
\dl^p_2\circ\g^p=\g^{p+1}\circ\dl^p_1, \qquad p\in\Bbb N.
\ee
It follows that if $b^p\in B^p_1$ is a cocycle or a coboundary,
then $\g^p(b^p)\in B^p_2$ is so. Therefore, the cochain morphism
of complexes (\ref{spr32'}) yields an induced homomorphism of
their cohomology groups
\mar{spr33'}\beq
[\g]^*: H^*(B^*_1) \to H^*(B^*_2). \label{spr33'}
\eeq

Let us consider a short sequence of complexes
\mar{spr34'}\beq
0\to C^*\ar^\g B^* \ar^\zeta F^*\to 0, \label{spr34'}
\eeq
represented by the commutative diagram
\be
\begin{array}{rcrccrccl}
& & & 0 & & & 0 & &\\
& & & \put(0,10){\vector(0,-1){20}} & & &
\put(0,10){\vector(0,-1){20}}
   & &\\
\cdots &\longrightarrow & & C^p & \op\longrightarrow^{\dl^p_C} & &
C^{p+1} & \longrightarrow &  \cdots \\
& & _{\g_p} & \put(0,10){\vector(0,-1){20}} & & _{\g_{p+1}}&
\put(0,10){\vector(0,-1){20}}
   &  &\\
\cdots &\longrightarrow & & B^p & \op\longrightarrow^{\dl^p_B} & &
B^{p+1} & \longrightarrow &  \cdots \\
& &_{\zeta_p} & \put(0,10){\vector(0,-1){20}} & & _{\zeta_{p+1}}
& \put(0,10){\vector(0,-1){20}} & & \\
\cdots &\longrightarrow & & F^p & \op\longrightarrow^{\dl^p_F} & &
F^{p+1} & \longrightarrow  & \cdots \\
& & & \put(0,10){\vector(0,-1){20}} & & &
\put(0,10){\vector(0,-1){20}}
   & & \\
& & & 0 & & & 0 & &
\end{array}
\ee
It is said to be exact if all columns of this diagram are exact,
i.e., $\g$ is a cochain monomorphism and $\zeta$ is a cochain
epimorphism onto the quotient $F^*=B^*/C^*$.

\begin{theo} \label{spr36'} \mar{spr36'}
The short exact sequence  of complexes (\ref{spr34'}) yields the
long exact sequence of their cohomology groups
\mar{spr35'}\ben
&& 0\to H^0(C^*)\ar^{[\g]^0} H^0(B^*)\ar^{[\zeta]^0} H^0(F^*)\ar^{\tau^0}
H^1(C^*)\ar\cdots \label{spr35'}\\
&& \qquad \ar H^p(C^*)\ar^{[\g]^p} H^p(B^*)\ar^{[\zeta]^p}
H^p(F^*)\ar^{\tau^p} H^{p+1}(C^*)\ar\cdots.\nonumber
\een
\end{theo}

\begin{theo} \label{spr37'} \mar{spr37'}
A direct sequence of complexes
\mar{spr55}\beq
B^*_0\ar B^*_1\ar\cdots B^*_k\ar^{\g^k_{k+1}} B^*_{k+1}\ar \cdots
\label{spr55}
\eeq
admits a direct limit $B^*_\infty$ which is a complex whose
cohomology $H^*(B^*_\infty)$ is a direct limit of the direct
sequence of cohomology groups
\be
H^*(B^*_0)\ar H^*(B^*_1)\ar \cdots H^*(B^*_k)\ar^{[\g^k_{k+1}]}
H^*(B^*_{k+1})\ar \cdots.
\ee
This statement is also true for a direct system of complexes
indexed by an arbitrary directed set.
\end{theo}

\section{Cohomology of Lie algebras}

One can associate to an arbitrary Lie algebra the
Chevalley--Eilenberg complex. In this Section, $\cG$ denotes a Lie
algebra (not necessarily finite-dimensional) over a commutative
ring $\cK$.

Let $\cG$ act on a $\cK$-module $P$ on the left by endomorphisms
\be
&& \cG\times P\ni (\ve,p)\mapsto \ve p\in P, \\
&& [\ve,\ve']p=(\ve\circ
\ve'-\ve'\circ \ve)p, \qquad \ve,\ve'\in\cG.
\ee
One says that $P$ is a {\it $\cG$-module}. \index{module!over a
Lie algebra}  A $\cK$-multilinear skew-symmetric map
\be
c^k:\op\times^k\cG\to P
\ee
is called a $P$-valued $k$-cochain on the Lie algebra $\cG$. These
cochains form a $\cG$-module $C^k[\cG;P]$. Let us put
$C^0[\cG;P]=P$. We obtain the cochain complex
\mar{spr997}\beq
0\to P\ar^{\dl^0} C^1[\cG;P]\ar^{\dl^1} \cdots C^k[\cG;P]
\ar^{\dl^k} \cdots, \label{spr997}
\eeq
with respect to the {\it Chevalley--Eilenberg coboundary
operators} \index{Chevalley--Eilenberg!coboundary operator}
\mar{spr132}\ben
&& \dl^kc^k (\ve_0,\ldots,\ve_k)=\op\sum_{i=0}^k(-1)^i\ve_ic^k(\ve_0,\ldots,
\wh\ve_i, \ldots, \ve_k)+ \label{spr132}\\
&& \qquad \op\sum_{1\leq i<j\leq k}
(-1)^{i+j}c^k([\ve_i,\ve_j], \ve_0,\ldots, \wh\ve_i, \ldots,
\wh\ve_j,\ldots, \ve_k), \nonumber
\een
where the caret $\,\wh{}\,$ denotes omission \cite{fuks}. The
complex (\ref{spr997}) is called the {\it Chevalley--Eilenberg
complex} \index{Chevalley--Eilenberg!complex} with coefficients in
a module $P$. It is finite if the Lie algebra $\cG$ is
finite-dimensional.

For instance,
\mar{spr133,4}\ben
&& \dl^0p(\ve_0)=\ve_0p, \label{spr133}\\
&& \dl^1c^1(\ve_0,\ve_1)=\ve_0c^1(\ve_1)-\ve_1c^1(\ve_0) -
c^1([\ve_0,\ve_1]). \label{spr134}
\een

Cohomology $H^*(\cG;P)$ of the complex $C^*[\cG;P]$ is called the
{\it Chevalley--Eilenberg cohomology of the Lie algebra $\cG$ with
coefficients in the module} $P$.
\index{Chevalley--Eilenberg!cohomology!with coefficients in a
module}

The following are two standard variants of the
Chevalley--Eilenberg complex.

(i) Let $P=\cG$  be regarded as a $\cG$-module with respect to the
adjoint representation
\be
\ve:\ve'\mapsto [\ve,\ve']\qquad, \ve,\ve'\in\cG.
\ee
We abbreviate with $C^*[\cG]$ the Chevalley--Eilenberg complex of
$\cG$-valued cochains on $\cG$. Cohomology $H^*(\cG)$ of this
complex is called the {\it Chevalley--Eilenberg cohomology}
\index{Chevalley--Eilenberg!cohomology} or, simply, the {\it
cohomology of a Lie algebra} $\cG$. \index{cohomology!of a Lie
algebra}

In particular, $C^0[\cG]=\cG$, while $C^1[\cG]$ consists of
endomorphisms of the Lie algebra $\cG$. Accordingly, the
coboundary operators (\ref{spr133}) and (\ref{spr134}) read
\mar{spr135,6}\ben
&&\dl^0\ve(\ve_0)=[\ve_0,\ve], \label{spr135} \\
&& \dl^1 c^1(\ve_0,\ve_1)=[\ve_0,c^1(\ve_1)]-[\ve_1,c^1(\ve_0)] -
c^1([\ve_0,\ve_1]). \label{spr136}
\een
A glance at the expression (\ref{spr136}) shows that a one-cocycle
$c^1$ on $\cG$ obeys the relation
\be
c^1([\ve_0,\ve_1])=[c^1(\ve_0),\ve_1]+ [\ve_0,c^1(\ve_1)]
\ee
and, thus, it is a derivation of the Lie algebra $\cG$.
Accordingly, any one-coboundary (\ref{spr135}) is an inner
derivation of $\cG$ up to the sign minus. Therefore, one can think
of the cohomology $H^1(\cG)$ as being the set of outer derivations
of $\cG$.

(ii) Let $P=\cK$ and $\cG:\cK\to 0$. Then the Chevalley--Eilenberg
complex $C^*[\cG;\cK]$ is the exterior algebra $\w\cG^*$ of the
dual Lie algebra $\cG^*$. The Chevalley--Eilenberg coboundary
operators (\ref{spr132}) on this algebra read
\mar{spr940}\beq
\dl^kc^k (\ve_0,\ldots,\ve_k)= \op\sum_{i<j}^k
(-1)^{i+j}c^k([\ve_i,\ve_j], \ve_0,\ldots, \wh\ve_i, \ldots,
\wh\ve_j,\ldots, \ve_k). \label{spr940}
\eeq
In particular,
\mar{020}\ben
&&\dl^0c^0(\ve_0)=0, \qquad c^0\in \cK, \nonumber \\
&& \dl^1 c^1(\ve_0,\ve_1)= -
c^1([\ve_0,\ve_1]), \qquad c^1\in \cG^*, \label{020}\\
&& \dl^2c^2(\ve_0,\ve_1,\ve_2)=-c^2([\ve_0,\ve_1],\ve_2) +
c^2([\ve_0,\ve_2],\ve_1) - c^2([\ve_1,\ve_2],\ve_0).\nonumber
\een

Cohomology $H^*(\cG;\cK)$ of the complex $C^*[\cG;\cK]$ are called
the {\it Chevalley--Eilenberg cohomology with coefficients in the
trivial representation}. \index{Chevalley--Eilenberg
cohomology!with coefficients in the trivial representation} It is
provided with the {\it cup-product} \index{cup-product!of the
Chevalley--Eilenberg cohomology}
\mar{ws80}\beq
[c]\smile [c']=[c\w c'], \label{ws80}
\eeq
where $[c]$ denotes the cohomology class of a cocycle $c$. This
product makes $H^*(\cG;\cK)$ into a graded commutative algebra.

For instance, let $\cG$ be the right Lie algebra of a
finite-dimensional real Lie group $G$. Then the relation
(\ref{020}) is the well-known Maurer--Cartan equation. Written
with respect to the basis $\ve_i$ for $\cG$ and the dual basis
$\thh^i$ for $\cG^*$, this equation reads
\mar{021}\beq
\dl\thh^k=-\frac12c^k_{ij}\thh^i\w \thh^j. \label{021}
\eeq
There is a monomorphism of the complex $C^*[\cG;\Bbb R]$ onto the
subcomplex of right-invariant exterior forms of the de Rham
complex (\ref{t37}) of exterior forms on $G$. This monomorphism
induces an isomorphism of the Chevalley--Eilenberg cohomology
$H^*(\cG;\Bbb R)$ to the de Rham cohomology of $G$ \cite{fuks}.
For instance, if $G$ is semi-simple, then
\be
H^1(\cG;\Bbb R)=H^2(\cG;\Bbb R)=0.
\ee

\section{Sheaf cohomology}

In this Section, we follow the terminology of \cite{bred,hir}.

A {\it sheaf} \index{sheaf} on a topological space $X$ is a
topological fibre bundle $\pi:S\to X$ in modules over a
commutative ring $\cK$, where the surjection $\pi$ is a local
homeomorphism and fibres $S_x$, $x\in X$, called the {\it stalks},
are provided with the discrete topology. Global sections of a
sheaf $S$ make up a $\cK$-module $S(X)$, called the {\it structure
module} \index{structure module!of a sheaf} of $S$.

Any sheaf is generated by a presheaf. A {\it presheaf}
\index{presheaf} $S_\sU$ on a topological space $X$ is defined if
a module $S_U$ over a commutative ring $\cK$ is assigned to every
open subset $U\subset X$ $(S_\emptyset=0)$ and if, for any pair of
open subsets $V\subset U$, there exists the {\it restriction
morphism} \index{restriction morphism} $r_V^U:S_U\rightarrow S_V$
such that
\be
r_U^U=\id S_U,\qquad  r_W^U=r_W^Vr_V^U, \qquad W\subset V\subset
U.
\ee

Every presheaf $S_\sU$ on a topological space $X$ yields a sheaf
on $X$ whose stalk $S_x$ at a point $x\in X$ is the direct limit
of the modules $S_U,\,x\in U$, with respect to the restriction
morphisms $r_V^U$. It means that, for each open neighborhood $U$
of a point $x$, every element $s\in S_U$ determines an element
$s_x\in S_x$, called the {\it germ} of $s$ at $x$. Two elements
$s\in S_U$ and $s'\in S_V$ belong to the same germ at $x$ iff
there exists an open neighborhood $W\subset U\cap V$ of $x$ such
that $r_W^Us=r_W^Vs'$.

\begin{ex} \label{spr7} \mar{spr7}
Let $C^0_\sU$ be the presheaf of continuous real functions on a
topological space $X$. Two such functions $s$ and $s'$ define the
same germ $s_x$ if they coincide on an open neighborhood  of $x$.
Hence, we obtain the {\it sheaf $C^0_X$ of continuous functions}
\index{sheaf!of continuous functions} on $X$. Similarly, the {\it
sheaf $C^\infty_X$ of smooth functions} \index{sheaf!of smooth
functions} on a smooth manifold $X$ is defined. Let us also
mention the presheaf of real functions which are constant on
connected open subsets of $X$. It generates the {\it constant
sheaf} on $X$ \index{sheaf!constant} denoted by $\Bbb R$.
\end{ex}

Two different presheaves may generate the same sheaf. Conversely,
every sheaf $S$ defines a presheaf $S(\sU)$ of modules $S(U)$ of
its local sections. It is called the {\it canonical presheaf}
\index{preasheaf!canonical} of the sheaf $S$. Global sections of
$S$ make up the {\it structure module} $S(X)$ \index{structure
module!of a sheaf} of $S$. If a  sheaf $S$ is constructed from a
presheaf $S_\sU$, there are natural module morphisms
\be
S_U\ni s\mapsto s(U)\in S(U), \qquad s(x)= s_x, \quad x\in U,
\ee
which are neither monomorphisms nor epimorphisms in general. For
instance, it may happen that a non-zero presheaf defines a zero
sheaf. The sheaf generated by the canonical presheaf of a sheaf
$S$ coincides with $S$.

A direct sum and a tensor product of presheaves (as families of
modules)  and sheaves (as fibre bundles in modules) are naturally
defined. By virtue of Theorem \ref{spr170}, a direct sum (resp. a
tensor product) of presheaves generates a direct sum (resp. a
tensor product) of the corresponding sheaves.

\begin{rem} \label{spr190'} \mar{spr190'}
In the terminology of \cite{tenn}, a sheaf is introduced as a
presheaf which satisfies the following additional axioms.

(S1) Suppose that $U\subset X$ is an open subset and $\{U_\al\}$
is its open cover. If $s,s'\in S_U$ obey the condition
\be
r^U_{U_\al}(s)=r^U_{U_\al}(s')
\ee
for all $U_\al$, then $s=s'$.

(S2) Let $U$ and $\{U_\al\}$ be as in previous item. Suppose that
we are given a family of presheaf elements $\{s_\al\in
S_{U_\al}\}$ such that
\be
r^{U_\al}_{U_\al\cap U_\la}(s_\al)=r^{U_\la}_{U_\al\cap
U_\la}(s_\la)
\ee
for all $U_\al$, $U_\la$. Then there exists a presheaf element
$s\in S_U$ such that $s_\al=r^U_{U_\al}(s)$.

Canonical presheaves are in one-to-one correspondence with
presheaves obeying these axioms. For instance, the presheaves of
continuous, smooth and locally constant functions in Example
\ref{spr7} satisfy the axioms (S1) -- (S2).
\end{rem}

\begin{rem}
The notion of a sheaf can be extended to sets, but not to
non-commutative groups. One can consider a presheaf of such
groups, but it generates a sheaf of sets because a direct limit of
non-commutative groups need not be a group.
\end{rem}

There is a useful construction of a sheaf on a topological space
$X$ from local sheaves on open subsets which make up a cover of
$X$.

\begin{prop} \label{+25} \mar{+25}
Let $\{U_\zeta\}$ be an open cover of a topological space $X$ and
$S_\zeta$ a sheaf on $U_\zeta$ for every $U_\zeta$. Let us suppose
that, if $U_\zeta\cap U_\xi\neq\emptyset$, there is a sheaf
isomorphism
\be
\rho_{\zeta\xi}:S_\xi\mid_{U_\zeta\cap U_\xi}\to
S_\zeta\mid_{U_\zeta\cap U_\xi}
\ee
and, for every triple $(U_\zeta,U_\xi, U_\iota)$, these
isomorphisms fulfil the cocycle condition
\be
\rho_{\xi\zeta}\circ\rho_{\zeta\iota}(S_\iota|_{U_\zeta\cap
U_\xi\cap U_\iota})=\rho_{\xi\iota}(S_\iota|_{U_\zeta\cap
U_\xi\cap U_\iota}).
\ee
Then there exists a sheaf $S$ on $X$ together with the sheaf
isomorphisms $\f_\zeta: S|_{U_\zeta}\to S_\zeta$ such that
\be
\f_\zeta|_{U_\zeta\cap
U_\xi}=\rho_{\zeta\xi}\circ\f_\xi|_{U_\zeta\cap U_\xi}.
\ee
\end{prop}

A {\it morphism of a presheaf} \index{morphism!of presheaves}
$S_\sU$ to a presheaf $S'_\sU$ on the same topological space $X$
is defined as a set of module morphisms $\g_U:S_U\to S'_U$ which
commute with restriction morphisms. A morphism of presheaves
yields a {\it morphism of sheaves} \index{morphism!of sheaves}
generated by these presheaves. This is a bundle morphism over $X$
such that $\g_x: S_x\to S'_x$ is the direct limit of morphisms
$\g_U$, $x\in U$. Conversely, any morphism of sheaves $S\to S'$ on
a topological space $X$ yields a morphism of canonical presheaves
of local sections of these sheaves. Let $\hm(S|_U,S'|_U)$ be the
commutative group of sheaf morphisms $S|_U\to S'|_U$ for any open
subset $U\subset X$. These groups are assembled into a presheaf,
and define the sheaf $\hm(S,S')$ on $X$. There is a monomorphism
\mar{+212}\beq
\hm(S,S')(U)\to \hm(S(U),S'(U)), \label{+212}
\eeq
which need not be an isomorphism.

By virtue of Theorem \ref{dlim1}, if a presheaf morphism is a
monomorphism or an epimorphism, so is the corresponding sheaf
morphism. Furthermore, the following holds.

\begin{theo} \label{spr29} \mar{spr29}
A short exact sequence
\mar{spr208}\beq
0\to S'_\sU\to S_\sU\to S''_\sU\to 0 \label{spr208}
\eeq
of presheaves on the same topological space yields the short exact
sequence of sheaves generated by these presheaves
\mar{ms0102}\beq
0\to S'\to S\to S''\to 0, \label{ms0102}
\eeq
where the {\it factor sheaf} \index{factor sheaf} $S''=S/S'$ is
isomorphic to that generated by the factor presheaf
\be
S''_\sU=S_\sU/S'_\sU.
\ee
If the exact sequence of presheaves (\ref{spr208}) is split, i.e.,
\be
S_\sU\cong S'_\sU\oplus S''_\sU,
\ee
the corresponding splitting
\be
S\cong S'\oplus S''
\ee
of the exact sequence of sheaves (\ref{ms0102}) holds.
\end{theo}

The converse is more intricate. A sheaf morphism induces a
morphism of the corresponding canonical presheaves. If $S\to S'$
is a monomorphism, $S(\sU)\to S'(\sU)$ is also a monomorphism.
However, if $S\to S'$ is an epimorphism, $S(\sU)\to S'(\sU)$ need
not be so. Therefore, the short exact sequence (\ref{ms0102}) of
sheaves yields the exact sequence of the canonical presheaves
\mar{ms0103'}\beq
0\to S'(\sU)\to S(\sU)\to S''(\sU), \label{ms0103'}
\eeq
where $S(\sU)\to S''(\sU)$ is not necessarily an epimorphism. At
the same time, there is the short exact sequence of presheaves
\mar{ms0103}\beq
0\to S'(\sU)\to S(\sU)\to S''_\sU \to 0, \label{ms0103}
\eeq
where the factor presheaf
\be
S''_\sU=S(\sU)/S'(\sU)
\ee
generates the factor sheaf $S''=S/S'$, but need not be its
canonical presheaf.

\begin{theo} \label{spr30} \mar{spr30}
Let the exact sequence of sheaves (\ref{ms0102}) be split. Then
\be
S(\sU)\cong S'(\sU) \oplus S''(\sU),
\ee
and the canonical presheaves make up the short exact sequence
\mar{+218}\beq
0\to S'(\sU)\to S(\sU)\to S''(\sU)\to 0. \label{+218}
\eeq
\end{theo}

Let us turn now to sheaf cohomology. We follow its definition in
\cite{hir}. In the case of paracompact topological spaces, it
coincides with a different definition of sheaf cohomology based on
the canonical flabby resolution (see Remark \ref{spr250} below).

Let $S_\sU$ be a presheaf of modules on a topological space $X$,
and let $\gU=\{U_i\}_{i\in I}$ be an open cover of $X$. One
constructs a cochain complex where a $p$-cochain is defined as a
function $s^p$ which associates an element
\be
s^p(i_0,\ldots,i_p)\in S_{U_{i_0}\cap\cdots\cap U_{i_p}}
\ee
to each $(p+1)$-tuple $(i_0,\ldots,i_p)$ of indices in $I$. These
$p$-cochains are assembled into a module $C^p(\gU,S_\sU)$. Let us
introduce the coboundary operator
\mar{spr180}\ben
&& \delta^p:C^p(\gU,S_\sU)\to C^{p+1}(\gU,S_\sU), \nonumber\\
&& \dl^ps^p(i_0,\ldots,i_{p+1})=\op\sum_{k=0}^{p+1}(-1)^kr_W^
{W_k}s^p(i_0,\ldots,\wh i_k,\ldots,i_{p+1}), \label{spr180}\\
&&
W=U_{i_0}\cap\ldots\cap U_{i_{p+1}},\qquad
W_k=U_{i_0}\cap\cdots\cap\wh U_{i_k}\cap\cdots\cap
U_{i_{p+1}}.\nonumber
\een
One can easily check that $\delta^{p+1}\circ\delta^p=0$. Thus, we
obtain the cochain complex of modules
\mar{spr181}\beq
0\to C^0(\gU,S_\sU)\ar^{\dl^0}\cdots C^p(\gU,S_\sU)\ar^{\dl^p}
C^{p+1}(\gU,S_\sU)\ar\cdots. \label{spr181}
\eeq
Its cohomology groups
\be
H^p(\gU;S_\sU)=\Ker\dl^p/\im\dl^{p-1}
\ee
are modules. Of course, they depend on an open cover $\gU$ of the
topological space $X$.

\begin{rem}
Throughout the Lectures, only {\it proper covers} \index{proper
cover} are considered, i.e., $U_i\neq U_j$ if $i\neq j$. A cover
$\gU'$ is said to be a {\it refinement} \index{refinement} of a
cover $\gU$ if, for each $U'\in\gU'$, there exists $U\in\gU$ such
that $U'\subset U$.
\end{rem}

Let $\gU'$ be a refinement of the cover $\gU$. Then there is a
morphism of cohomology groups
\be
H^*(\gU;S_\sU)\rightarrow H^*(\gU';S_\sU).
\ee
Let us take the direct limit of cohomology groups $H^*(\gU;S_\sU)$
with respect to these morphisms, where $\gU$ runs through all open
covers of $X$. This limit $H^*(X;S_\sU)$ is called the {\it
cohomology of $X$ with coefficients in the presheaf} $S_\sU$.
\index{cohomology!with coefficients in a presheaf}

Let $S$ be a sheaf on a topological space $X$. {\it Cohomology of
$X$ with coefficients in $S$} or, simply, {\it sheaf cohomology}
of $X$ \index{sheaf cohomology} is defined as cohomology
\be
H^*(X;S)=H^*(X;S(\sU))
\ee
with coefficients in the canonical presheaf $S(\sU)$ of the sheaf
$S$.

In this case, a $p$-cochain $s^p\in C^p(\gU,S(\sU))$ is a
collection
\be
s^p=\{s^p(i_0,\ldots,i_p)\}
\ee
of local sections $s^p(i_0,\ldots,i_p)$ of the sheaf $S$ over
$U_{i_0}\cap\cdots\cap U_{i_p}$ for each $(p+1)$-tuple
$(U_{i_0},\ldots,U_{i_p})$ of elements of the cover $\gU$. The
coboundary operator (\ref{spr180}) reads
\mar{spr192}\beq
\dl^ps^p(i_0,\ldots,i_{p+1})=\op\sum_{k=0}^{p+1}(-1)^k
s^p(i_0,\ldots,\wh i_k,\ldots,i_{p+1})|_{U_{i_0}\cap\cdots\cap
U_{i_{p+1}}}. \label{spr192}
\eeq
For instance,
\mar{spr188,9}\ben
&& \dl^0s^0(i,j)=[s^0(j) -s^0(i)]|_{U_i\cap U_j},
\label{spr188}\\
&& \dl^1s^1(i,j,k)=[s^1(j,k)-s^1(i,k)
   +s^1(i,j)]|_{U_i\cap U_j\cap U_k}. \label{spr189}
\een
A glance at the expression (\ref{spr188}) shows that a
zero-cocycle is a collection $s=\{s(i)\}_I$ of local sections of
the sheaf $S$ over $U_i\in\gU$ such that $s(i)=s(j)$ on $U_i\cap
U_j$. It follows from the axiom (S2) in Remark \ref{spr190'} that
$s$ is a global section of the sheaf $S$, while each $s(i)$ is its
restriction $s|_{U_i}$ to $U_i$. Consequently, the cohomology
group $H^0(\gU,S(\sU))$ is isomorphic to the structure module
$S(X)$ of global sections of the sheaf $S$. A one-cocycle is a
collection $\{s(i,j)\}$ of local sections of the sheaf $S$ over
overlaps $U_i\cap U_j$ which satisfy the {\it cocycle condition}
\index{cocycle condition}
\mar{spr192'}\beq
[s(j,k)-s(i,k) +s(i,j)]|_{U_i\cap U_j\cap U_k}=0. \label{spr192'}
\eeq

If $X$ is a paracompact space, the study of its sheaf cohomology
is essentially simplified due to the following fact \cite{hir}.

\begin{theo} \label{spr225} \mar{spr225}
Cohomology of a paracompact space $X$ with coefficients in a sheaf
$S$ coincides with cohomology of $X$ with coefficients in any
presheaf generating the sheaf $S$.
\end{theo}

\begin{rem} \label{spr200} \mar{spr200}
We follow the definition of a {\it paracompact topological space}
\index{paracompact space} in \cite{hir} as a Hausdorff space such
that any its open cover admits a {\it locally finite}
\index{locally finite covering} open refinement, i.e., any point
has an open neighborhood which intersects only a finite number of
elements of this refinement. A topological space $X$ is
paracompact iff any cover $\{U_\xi\}$ of $X$ admits a subordinate
{\it partition of unity} \index{partition of unity} $\{f_\xi\}$,
i.e.:

(i) $f_\xi$ are real positive continuous functions on $X$;

(ii) supp$\,f_\xi\subset U_\xi$;

(iii) each point $x\in X$ has an open neighborhood which
intersects only a finite number of the sets supp$\,f_\xi$;

(iv) $\op\sum_\xi f_\xi(x)=1$ for all $x\in X$.
\end{rem}

The key point of the analysis of sheaf cohomology is that short
exact sequences of presheaves and sheaves yield long exact
sequences of sheaf cohomology groups.

Let $S_\sU$ and $S'_\sU$ be presheaves on the same topological
space $X$. It is readily observed that, given an open cover $\gU$
of $X$, any morphism $S_\sU\to S'_\sU$  yields a cochain morphism
of complexes
\be
C^*(\gU,S_\sU)\to C^*(\gU,S'_\sU)
\ee
and the corresponding morphism
\be
H^*(\gU,S_\sU)\to H^*(\gU,S'_\sU)
\ee
of cohomology groups of these complexes. Passing to the direct
limit through all refinements of $\gU$, we come to a morphism of
the cohomology groups
\be
H^*(X,S_\sU)\to H^*(X,S'_\sU)
\ee
of $X$ with coefficients in the presheaves $S_\sU$ and $S'_\sU$.
In particular, any sheaf morphism $S\to S'$ yields a morphism of
canonical presheaves $S(\{U\})\to S'(\{U\})$ and the corresponding
cohomology morphism
\be
H^*(X,S)\to H^*(X,S').
\ee

By virtue of Theorems \ref{spr36'} and \ref{spr37'}, every short
exact sequence
\mar{spr220}\beq
0\to S'_\sU\ar S_\sU\ar S''_\sU\to 0 \label{spr220}
\eeq
of presheaves on the same topological space $X$ and the
corresponding exact sequence of complexes (\ref{spr181}) yield the
long exact sequence
\mar{spr221}\ben
&& 0\to H^0(X;S'_\sU)\ar H^0(X;S_\sU)\ar H^0(X;S''_\sU)\ar \label{spr221}\\
&& \qquad
H^1(X;S'_\sU) \ar\cdots  H^p(X;S'_\sU)\ar H^p(X;S_\sU)\ar \nonumber\\
&& \qquad  H^p(X;S''_\sU)\ar
H^{p+1}(X;S'_\sU) \ar\cdots \nonumber
\een
of the cohomology groups of $X$ with coefficients in these
presheaves. This result however is not extended to an exact
sequence of sheaves, unless $X$ is a paracompact space. Let
\mar{spr226}\beq
0\to S'\ar S\ar S'' \to 0 \label{spr226}
\eeq
be a short  exact sequence of sheaves on $X$. It yields the short
exact sequence of presheaves (\ref{ms0103}) where the presheaf
$S''_\sU$ generates the sheaf $S''$. If $X$ is paracompact,
\be
H^*(X;S''_\sU)=H^*(X;S'')
\ee
in accordance with Theorem \ref{spr225}, and we have the exact
sequence of sheaf cohomology
\mar{spr227}\ben
&& 0\to H^0(X;S')\ar H^0(X;S)\ar H^0(X;S'')\ar
\label{spr227}\\
&& \qquad H^1(X;S') \ar\cdots   H^p(X;S')\ar H^p(X;S)\ar \nonumber
\\
&& \qquad H^p(X;S'')\ar H^{p+1}(X;S') \ar\cdots\,. \nonumber
\een

Let us point out the following isomorphism between sheaf
cohomology and singular (\v Chech and Alexandery) cohomology of a
paracompact space \cite{bred,span}.

\begin{theo} \label{spr257} \mar{spr257}
The sheaf cohomology $H^*(X;\Bbb Z)$ (resp. $H^*(X;\Bbb Q)$,
$H^*(X;\Bbb R)$) of a paracompact topological space $X$ with
coefficients in the constant sheaf $\Bbb Z$ (resp. $\Bbb Q$, $\Bbb
R$) is isomorphic to the singular cohomology of $X$  with
coefficients in the ring $\Bbb Z$ (resp. $\Bbb Q$, $\Bbb R$).
\end{theo}

Since singular cohomology is a {\it topological invariant}
\index{topological invariant} (i.e., homotopic topological spaces
have the same singular cohomology) \cite{span}, the sheaf
cohomology groups $H^*(X;\Bbb Z)$, $H^*(X;\Bbb Q)$, $H^*(X;\Bbb
R)$ of a paracompact space are also topological invariants.

Let us turn now to the abstract de Rham theorem which provides a
powerful tool of studying algebraic systems on paracompact spaces.

Let us consider an exact sequence of sheaves
\mar{spr228}\beq
0\to S\ar^h S_0\ar^{h^0} S_1\ar^{h^1}\cdots S_p\ar^{h^p}\cdots.
\label{spr228}
\eeq
It is said to be a {\it resolution of the sheaf}
\index{resolution!of a sheaf} $S$ if each sheaf $S_{p\geq 0}$ is
{\it acyclic}, \index{sheaf!acyclic} i.e., its cohomology groups
$H^{k>0}(X;S_p)$ vanish.

Any exact sequence of sheaves (\ref{spr228}) yields the sequence
of their structure modules
\mar{spr229}\beq
0\to S(X)\ar^{h_*} S_0(X)\ar^{h^0_*} S_1(X)\ar^{h^1_*}\cdots
S_p(X)\ar^{h^p_*}\cdots \label{spr229}
\eeq
which is always exact at terms $S(X)$ and $S_0(X)$ (see the exact
sequence (\ref{ms0103'})). The sequence (\ref{spr229}) is a
cochain complex because $h^{p+1}_*\circ h^p_*=0$. If $X$ is a
paracompact space and the exact sequence (\ref{spr228}) is a
resolution of $S$, the {\it abstract de Rham theorem} \index{de
Rham theorem!abstract} establishes an isomorphism of cohomology of
the complex (\ref{spr229}) to cohomology of $X$ with coefficients
in the sheaf $S$ as follows \cite{hir}.

\begin{theo} \label{spr230} \mar{spr230}
Given a resolution (\ref{spr228}) of a sheaf $S$ on a paracompact
topological  space $X$ and the induced complex (\ref{spr229}),
there are isomorphisms
\mar{spr231}\beq
H^0(X;S)=\Ker h^0_*, \qquad H^q(X;S)=\Ker h^q_*/\im h^{q-1}_*,
\qquad q>0. \label{spr231}
\eeq
\end{theo}

We will also refer to the following minor modification of Theorem
\ref{spr230} \cite{jmp01,tak2}.

\begin{theo} \label{spr232} \mar{spr232}
Let
\mar{+131'}\beq
0\to S\ar^h S_0\ar^{h^0} S_1\ar^{h^1}\cdots\ar^{h^{p-1}}
S_p\ar^{h^p} S_{p+1}, \qquad p>1, \label{+131'}
\eeq
be an exact sequence of sheaves on a paracompact topological space
$X$, where the sheaves $S_q$, $0\leq q<p$, are acyclic, and let
\mar{+130}\beq
0\to S(X)\ar^{h_*} S_0(X)\ar^{h^0_*}
S_1(X)\ar^{h^1_*}\cdots\ar^{h^{p-1}_*} S_p(X)\ar^{h^p_*}
S_{p+1}(X) \label{+130}
\eeq
be the corresponding cochain complex of structure modules of these
sheaves. Then the isomorphisms (\ref{spr231}) hold for $0\leq
q\leq p$.
\end{theo}

Any sheaf on a topological space admits the canonical resolution
by flabby sheaves as follows.

A sheaf $S$ on a topological space $X$ is called {\it flabby}
\index{sheaf!flabby} (or {\it flasque} \index{sheaf!flasque} in
the terminology of \cite{tenn}), if the restriction morphism
$S(X)\to S(U)$ is an epimorphism for any open $U\subset X$, i.e.,
if any local section of the sheaf $S$ can be extended to a global
section. A flabby sheaf is acyclic. Indeed, given an arbitrary
cover $\gU$ of $X$, let us consider the complex $C^*(\gU,S(\sU))$
(\ref{spr181}) for its canonical presheaf $S(\sU)$. Since $S$ is
flabby, one can define a morphism
\mar{spr251}\ben
&& h: C^p(\gU,S(\sU))\to C^{p-1}(\gU,S(\sU)), \qquad p>0, \nonumber \\
&& hs^p(i_0,\ldots,i_{p-1})= j^*s^p(i_0,\ldots,i_{p-1}, j),
\label{spr251}
\een
where $U_j$ is a fixed element of the cover $\gU$ and $j^*s^p$ is
an extension of $s^p(i_0,\ldots,i_{p-1}, j)$ onto
$U_{i_0}\cap\cdots \cap U_{i_{p-1}}$. A direct verification shows
that $h$ (\ref{spr251}) is a homotopy operator for the complex
$C^*(\gU,S(\sU))$ and, consequently, $H^{p>0}(\gU, S(\sU))=0$.

Given an arbitrary sheaf $S$ on a topological space $X$, let
$S_F^0(\sU)$ denote the presheaf of all (not-necessarily
continuous) sections of the sheaf $S$. It generates a sheaf
$S_F^0$ on $X$, and coincides with the canonical presheaf of this
sheaf. There are the natural monomorphisms $S(\sU)\to S_F^0(\sU)$
and $S\to S_F^0$. It is readily observed that the sheaf $S_F^0$ is
flabby. Let us take the quotient $S_F^0/S$ and construct the
flabby sheaf
\be
S^1_F=(S^0_F/S)^0_F.
\ee
Continuing the procedure, we obtain the exact sequence of sheaves
\mar{spr253}\beq
0\to S\ar S_F^0\ar S_F^1\ar\cdots, \label{spr253}
\eeq
which is a resolution of $S$ since all sheaves are flabby and,
consequently, acyclic. It is called the {\it canonical flabby
resolution} \index{resolution!flabby} of the sheaf $S$. The exact
sequence of sheaves (\ref{spr253}) yields the complex of structure
modules of these sheaves
\mar{spr254}\beq
0\to S(X)\ar S_F^0(X)\ar S_F^1(X)\ar\cdots\,. \label{spr254}
\eeq
If $X$ is paracompact, the cohomology of $X$ with coefficients in
the sheaf $S$ coincides with cohomology of the complex
(\ref{spr254}) by virtue of Theorem \ref{spr230}.

\begin{rem} \label{spr250} \mar{spr250}
An important peculiarity of flabby sheaves is that a short exact
sequence of flabby sheaves on an arbitrary topological space
provides the short exact sequence of their structure modules.
Therefore, there is a different definition of sheaf cohomology.
Cohomology of a topological space $X$ with coefficients in a sheaf
$S$ is defined directly as cohomology of the complex
(\ref{spr254}) \cite{bred}. For a paracompact space, this
definition coincides with above mentioned one due to Theorem
\ref{spr230}.
\end{rem}

In the sequel, we also refer to a {\it fine resolution}
\index{resolution!fine} of sheaves, i.e., a resolution by fine
sheaves.

A sheaf $S$  on a paracompact space $X$ is called {\it fine}
\index{sheaf!fine} if, for each locally finite open cover $\gU
=\{U_i\}_{i\in I}$ of $X$,  there exists a system $\{h_i\}$ of
endomorphisms $h_i:S\to S$ such that:

(i) there is a closed subset $V_i\subset U_i$ and $h_i(S_x)=0$ if
$x\not\in V_i$,

(ii) $\op\sum_{i\in I}h_i$ is the identity map of $S$. A fine
sheaf on a paracompact space is acyclic.

\noindent Indeed, given an arbitrary locally finite cover
$\gU=\{U_i\}_{i\in I}$
  of $X$ and a $p$-cochain $s^p$,
let us define the $(p-1)$-cochain
\mar{spr255}\beq
{\bf h}s^p(i_0,\ldots,i_{p-1})=\op\sum_{i\in I}
i^*s^p(i,i_0,\ldots,i_{p-1}) \label{spr255}
\eeq
where $i^*s^p$, by definition, is equal to $h_is^p$ on the set
$U_i\cap U_{i_0}\cap \cdots U_{i_{p-1}}$ and to 0 outside this
set. Then the morphism ${\bf h}$ (\ref{spr255}) is a homotopy
operator.

There are the following important examples of fine sheaves
\cite{book05}.

\begin{prop}  \label{spr256} \mar{spr256}
Let $X$ be a paracompact topological  space which admits a
partition of unity performed by elements of the structure module
$\gA(X)$ of some sheaf $\gA$ of real functions on $X$. Then any
sheaf $S$ of $\gA$-modules on $X$, including $\gA$ itself, is
fine.
\end{prop}

In particular, the sheaf $C^0_X$ of continuous functions on a
paracompact topological space is fine, and so is any sheaf of
$C^0_X$-modules. A smooth manifold $X$ admits a partition of unity
performed by smooth real functions. It follows that the sheaf
$C^\infty_X$ of smooth real functions on $X$ is fine, and so is
any sheaf of $C^\infty_X$-modules, e.g., the sheaves of sections
of smooth vector bundles over $X$.

We complete our exposition of sheaf cohomology with the following
useful theorem \cite{bart}.

\begin{theo} \label{+144} \mar{+144}
Let $f:X\to X'$ be a continuous map and $S$ a sheaf on $X$. Let
either $f$ be a closed immersion or every point $x'\in X'$ have a
base of open neighborhoods $\{U\}$ such that the sheaves
$S\mid_{f^{-1}(U)}$ are acyclic. Then the cohomology groups
$H^*(X;S)$ and $H^*(X';f_*S)$ are isomorphic.
\end{theo}

\bibliographystyle{alpha}
\bibliographystyle{plain}
\addcontentsline{toc}{chapter}{Bibliography}
\bibliography{conn99}

%------------------------------------------------------------------------------
%%%       To produce INDEX: after running Latex 3 times, run
%%%                      makeindex gov99
%%%       and after this one more time run Latex

\addcontentsline{toc}{chapter}{Index}

%\printindex

\begin{theindex}

  \item $C^*$-algebra
    \subitem defined by a continuous field of $C^*$-algebras, 41
    \subitem elementary, 41
  \item $C^*$-dynamic system, 46
  \item $C^*$-module, 82
  \item $G$-superfunction, 72
  \item $G$-supermanifold, 73
  \item $G$-supermanifold standard, 73
  \item $GH^\infty$-superfunctions, 71
  \item $G^\infty$-superfunctions, 71
  \item $G^\infty$-tangent bundle, 77
  \item $G^\infty$-vector bundle, 77
  \item $H^\infty$-superfunctions, 71
  \item $K$-cycle, 97
  \item $\cK$-algebra, 6
  \item $p$-coboundary, 111
  \item $p$-cocycle, 111

  \indexspace

  \item  Atiyah class, 27

  \indexspace

  \item adjoint action
    \subitem of a Hopf algebra, 104
  \item algebra, 5
    \subitem $\Bbb N$-graded, 17
      \subsubitem commutative, 17
    \subitem $\Bbb Z_2$-graded, 53
    \subitem unital, 5
  \item algebra $\Bbb Z_2$-graded commutative, 53
  \item antiholomorphic function on a Hilbert space, 36
  \item antipode, 101

  \indexspace

  \item Banach manifold, 30
  \item Banach vector bundle, 40
  \item basis
    \subitem for a graded manifold, 63
    \subitem for a module, 7
  \item Batchelor's theorem, 62
  \item Berry connection, 49
  \item bi-algebra, 100
  \item bimodule, 6
    \subitem covariant, 107
  \item bimodule commutative, 6
  \item bimodule graded, 54
  \item body
    \subitem of a graded manifold, 62
    \subitem of a ringed space, 19
  \item body map, 55

  \indexspace

  \item Cartan pair, 89
  \item category ${\bf Vect}(\cB)$, 32
  \item category {\bf Bnh}, 32
  \item center of an algebra, 81
  \item characteristic vector bundle, 63
  \item Chevalley--Eilenberg
    \subitem coboundary operator, 114
    \subitem cohomology, 114
      \subsubitem with coefficients in a module, 114
    \subitem complex, 114
  \item Chevalley--Eilenberg cohomology
    \subitem with coefficients in the trivial representation, 115
  \item Chevalley--Eilenberg differential calculus, 18
    \subitem minimal, 19
  \item co-operation, 101
  \item coaction, 105
  \item coalgebra, 100
  \item coboundary operator, 111
  \item cochain, 111
  \item cochain complex, 111
  \item cochain morphism, 112
  \item cocycle condition, 120
  \item cohomology
    \subitem of a Lie algebra, 114
    \subitem with coefficients in a presheaf, 120
  \item cohomology group, 111
  \item comodule
    \subitem over a Hopf algebra, 105
    \subitem trivial, 105
  \item comodule algebra, 105
  \item complex, 111
    \subitem acyclic, 111
    \subitem exact, 111
  \item complex vector field
    \subitem on a Hilbert manifold, 36
  \item comultiplication, 100
  \item connection
    \subitem left, 92
    \subitem non-commutative
      \subsubitem universal, 92
    \subitem on a bundle of $C^*$-algebras, 43
    \subitem on a graded commutative ring, 62
    \subitem on a Hilbert bundle, 44
    \subitem on a Hilbert manifold, 38
    \subitem on a module, 15, 16
    \subitem on a ring, 16
    \subitem on a sheaf, 26
    \subitem right, 92
  \item connection on a Banach manifold, 33
  \item connection on a graded module, 61
  \item Connes' differential calculus, 98
  \item convolution product, 103
  \item cotangent bundle
    \subitem  complex, 37
    \subitem antiholomorphic, 37
    \subitem holomorphic, 37
    \subitem of a Banach manifold, 32
  \item counit, 100
  \item covariant differential
    \subitem on a module, 16
  \item cup-product, 17
    \subitem of the Chevalley--Eilenberg cohomology, 115
  \item curvature of a connection
    \subitem on a module, 16
  \item curvature of a connection on sheaves, 26
  \item curvature of a graded connection, 66
  \item curvature of a superconnection, 78

  \indexspace

  \item de Rham cohomology, 17
    \subitem of a ring, 19
  \item de Rham cohomology of a manifold, 22
  \item de Rham complex, 17
    \subitem of a ring, 19
    \subitem of sheaves, 22
  \item de Rham complex of exterior forms, 22
  \item de Rham theorem, 22
  \item de Rham theorem abstract, 122
  \item derivation
    \subitem inner, 12
    \subitem of a non-commutative algebra, 12
    \subitem with values in a module, 11
  \item derivation graded, 64
  \item derivation module, 11
  \item derivative
    \subitem on a Banach space, 30
  \item differentiable function
    \subitem between Banach spaces, 30
    \subitem on a Hilbert space, 35
  \item differential
    \subitem on a Banach space, 30
  \item differential bigraded algebra, 60
  \item differential calculus, 17
    \subitem minimal, 18
  \item differential calculus universal, 86
  \item differential graded algebra, 17
  \item differential operator
    \subitem in non-commutative geometry, 88, 90
    \subitem on a module, 10
  \item direct image
    \subitem of a sheaf, 19
  \item direct limit, 8
  \item direct sequence, 8
  \item direct sum
    \subitem of complexes, 112
    \subitem of modules, 6
  \item direct system
    \subitem of modules, 8
  \item directed set, 8
  \item division algebra, 5
  \item Dixmier--Douady class, 43
  \item dual Hopf algebra, 104
  \item Dubois--Violette connection, 93
    \subitem linear, 95
    \subitem real, 95
  \item Dubois--Violette connection dual, 94

  \indexspace

  \item endomorphism
    \subitem of a Hilbert module, 83
      \subsubitem adjoint, 83
      \subsubitem compact, 83
  \item evaluation morphism, 71
  \item even element, 53
  \item even morphism, 55
  \item evolution operator, 48
  \item exact sequence
    \subitem of modules, 7
    \subitem short, 7
  \item exterior algebra, 9
    \subitem bigraded, 54
  \item exterior form
    \subitem antiholomorphic
      \subsubitem on a Hilbert manifold, 37
    \subitem on a Hilbert manifold, 37
  \item exterior form on a Banach manifold, 33

  \indexspace

  \item factor algebra, 5
  \item factor complex, 112
  \item factor module, 7
  \item factor sheaf, 118
  \item fibre bundle
    \subitem of $C^*$-algebras, 41
  \item fibred manifold Banach, 32
  \item field, 5
  \item field of $C^*$-algebras, 41
  \item FODC, 106
    \subitem bicovariant, 106
    \subitem left-covariant, 106
    \subitem right-covariant, 106
  \item Fr\'echet ring, 21
  \item Fubini--Studi metric, 40
  \item fundamental form
    \subitem of a Hermitian metric on a Hilbert manifold, 38

  \indexspace

  \item general linear graded group, 56
  \item generalized connection, 44
  \item geometric space, 19
  \item graded commutative ring, 55
    \subitem Banach, 55
  \item graded connection, 66
  \item graded de Rham cohomology, 68
  \item graded derivation, 57
  \item graded derivation module, 58
  \item graded differential calculus, 60
  \item graded differential operator, 57
  \item graded envelope, 56
  \item graded exterior differential, 68
  \item graded exterior form, 67
  \item graded exterior product, 54
  \item graded function, 62
  \item graded jet module, 58
  \item graded manifold, 62
    \subitem simple, 63
  \item graded module, 54
    \subitem dual, 55
    \subitem free, 54
  \item graded vector field, 64
  \item graded vector space, 54
  \item grading automorphism, 53
  \item Grassmann algebra, 55
  \item Grothendieck's topology, 21
  \item group function Hopf algebra, 102
  \item group Hopf algebra, 102

  \indexspace

  \item Heisenberg equation, 46
  \item Hermitian form
    \subitem on a module, 82
      \subsubitem invertible, 82
      \subsubitem positive, 82
  \item Hermitian manifold
    \subitem infinite-dimensional, 38
  \item Hermitian metric
    \subitem on a Hilbert manifold, 38
  \item Hilbert bundle, 41
  \item Hilbert manifold, 36
  \item Hilbert module, 41, 82
  \item holomorphic exterior form on a Hilbert manifold, 37
  \item holomorphic function on a Hilbert space, 36
  \item homogeneous element, 54
  \item homotopy operator, 111
  \item Hopf algebra, 101
    \subitem classical, 102
    \subitem cocommutative, 101
    \subitem dually paired, 103
    \subitem involutive, 101
    \subitem quasi-triangular, 101

  \indexspace

  \item ideal, 5
    \subitem maximal, 5
    \subitem of nilpotents, 55
    \subitem prime, 5
  \item ideal proper, 5
  \item inductive limit, 9
  \item instantwise quantization, 45
  \item interior product
    \subitem graded, 60
  \item invariant vector field, 109
  \item inverse image
    \subitem of a sheaf, 20
  \item inverse limit, 9
  \item inverse mapping theorem, 30
  \item inverse sequence, 9

  \indexspace

  \item jet module, 12
  \item jet sheaf, 25
  \item juxtaposition rule, 17

  \indexspace

  \item K\"ahler form
    \subitem on a Hilbert manifold, 38
  \item K\"ahler manifold
    \subitem infinite-dimensional, 38
  \item K\"ahler metric
    \subitem on a Hilbert manifold, 38
  \item Koszul connection, 16

  \indexspace

  \item Leibniz rule, 11, 16
  \item Leibniz rule graded, 17
  \item Leibniz rule non-commutative, 12
  \item Lie derivative
    \subitem of a bigraded algebra, 60
  \item Lie superalgebra, 56
  \item local ring, 5
  \item local-ringed space, 19
  \item locally
    \subitem free module, 20
  \item locally finite covering, 120

  \indexspace

  \item matrix geometry, 95
  \item Maurer--Cartan form
    \subitem on a Hopf algebra, 108
  \item metric connection
    \subitem on a Hilbert manifold, 39
  \item module, 6
    \subitem dual, 7
    \subitem finitely generated, 7
    \subitem free, 7
    \subitem involutive, 82
    \subitem of finite rank, 7
    \subitem over a Hopf algebra, 104
    \subitem over a Lie algebra, 114
    \subitem projective, 7
  \item morphism
    \subitem of presheaves, 117
    \subitem of ringed spaces, 20
    \subitem of sheaves, 117

  \indexspace

  \item non-commutative gauge field, 100
  \item non-commutative vector field, 89

  \indexspace

  \item odd element, 53
  \item odd morphism, 55
  \item operator of a parallel displacement, 47

  \indexspace

  \item paracompact space, 120
  \item partition of unity, 121
  \item positive element, 82
  \item pre-Hilbert module, 82
  \item preaheaf
    \subitem canonical, 116
  \item presheaf, 116
  \item product
    \subitem of $G$-supermanifolds, 75
  \item product $G$-supermanifold, 75
  \item projective Hilbert space, 39
  \item proper cover, 119

  \indexspace

  \item quantum Lie algebra of a Hopf algebra, 109
  \item quantum tangent space, 109

  \indexspace

  \item refinement, 119
  \item resolution, 112
    \subitem fine, 124
    \subitem flabby, 124
    \subitem of a sheaf, 122
  \item restriction morphism, 116
  \item ring, 5
  \item ringed space, 19

  \indexspace

  \item Schr\"odinger equation, 47
  \item second-countable topological space, 21
  \item separable topological space, 21
  \item Serre--Swan theorem, 23
  \item sheaf, 116
    \subitem acyclic, 122
    \subitem constant, 116
    \subitem fine, 124
    \subitem flabby, 123
    \subitem flasque, 123
    \subitem locally free, 20
      \subsubitem of constant rank, 20
      \subsubitem of finite type, 20
    \subitem of continuous functions, 116
    \subitem of derivations, 20
    \subitem of jets, 25
    \subitem of modules, 20
    \subitem of smooth functions, 116
  \item sheaf cohomology, 120
  \item smooth manifold, 21
  \item smooth supermanifold, 72
  \item soul map, 55
  \item spectral triple, 97
  \item spectral triple even, 97
  \item spectral triple odd, 97
  \item split (injection), 29
  \item split (subspace), 29
  \item splitting domain, 62
  \item splitting of an exact sequence, 8
  \item structure module
    \subitem of a sheaf, 116
    \subitem of a simple graded manifold, 63
    \subitem of a vector bundle, 23
  \item structure sheaf
    \subitem of a graded manifold, 62
    \subitem of a ringed space, 19
    \subitem of a vector bundle, 22
  \item structure sheaf of a supermanifold, 72
  \item superbracket, 56
  \item superconnection, 78
  \item superform, 73
  \item superfunction, 69
    \subitem smooth, 71
  \item supermatrix, 56
  \item supermatrix even, 56
  \item supermatrix odd, 56
  \item supertangent bundle, 77
  \item supervector bundle, 76
  \item supervector field, 73
  \item supervector space, 56

  \indexspace

  \item tangent bundle
    \subitem antiholomorphic, 36
    \subitem complexified , 36
    \subitem holomorphic, 36
  \item tangent bundle of a Banach manifold, 31
  \item tangent space
    \subitem antiholomorphic, 36
  \item tangent space complex, 36
  \item tangent space holomorphic, 36
  \item tangent space to a Banach manifold, 31
  \item tensor algebra, 9
  \item tensor product
    \subitem of algebras, 100
    \subitem of complexes, 112
    \subitem of graded modules, 54
    \subitem of modules, 6
    \subitem topological, 21
  \item time-ordered exponential, 48
  \item topological invariant, 122
  \item torsion
    \subitem of a non-commutative connection, 95
  \item trivial extension of a sheaf, 20

  \indexspace

  \item underlying $G^\infty$-supermanifold, 73
  \item unital extension, 6
  \item universal $R$ matrix, 101
  \item universal form, 86

  \indexspace

  \item vector field
    \subitem on a Banach manifold, 32
  \item vector space, 6

  \indexspace

  \item Yang--Baxter equation
    \subitem quantum, 101

\end{theindex}

\end{document}